\documentclass[10pt, oneside]{article}

\input ./math_headers.sty

\title{Topological Twists of Supersymmetric Algebras of Observables}
\author{Chris Elliott \and Pavel Safronov}

\bibliography{Twist}

\renewcommand{\O}{\mathrm{O}}
\newcommand{\C}{\mathrm{C}}
\newcommand{\ialg}{\mathcal{A}\mathrm{lg}}
\newcommand{\cC}{\mathcal{C}}
\newcommand{\cD}{\mathcal{D}}
\newcommand{\cE}{\mathcal{E}}
\newcommand{\cF}{\mathcal{F}}
\newcommand{\cO}{\mathcal{O}}
\newcommand{\cP}{\mathcal{P}}
\newcommand{\cS}{\mathcal{S}}

\newcommand{\BD}{\mathbb{BD}}

\newcommand{\Cat}{\mathrm{Cat}}
\newcommand{\Conf}{\mathrm{Conf}}
\newcommand{\Fun}{\mathrm{Fun}}
\newcommand{\ML}{\mathrm{ML}}

\begin{document}

\maketitle
\begin{abstract}
We explain how to perform topological twisting of supersymmetric field theories in the language of factorization algebras. Namely, given a supersymmetric factorization algebra with a choice of a topological supercharge we construct an algebra over the operad of little disks. We also explain the role of the twisting homomorphism allowing us to construct an algebra over the operad of framed little disks. Finally, we give a complete classification of topological supercharges and twisting homomorphisms in dimensions $1$ through $10$.
\end{abstract}

\section{Introduction}

Quantum field theories contain a large amount of information, part of which has been formalized mathematically. For instance, given a quantum field theory on a manifold $M$, we may consider the set of observables $\obs(U)$ supported on an open subset $U\sub M$. Given a disjoint collection of open subsets $U_1,\dots, U_n\sub V\sub M$ we may combine observables supported on each open set to get an operator product map $\obs(U_1)\otimes \dots \otimes \obs(U_n)\rightarrow \obs(V)$. In this way, Costello and Gwilliam gave a formalization of the structure of observables in perturbative quantum field theories as a \emph{factorization algebra}, see \cite{Book1,Book2}.

Topological field theories do not depend on the precise shape of the manifold, so we expect that the maps $\obs(U)\rightarrow \obs(V)$ are equivalences for an embedding of open subsets $U\sub V$ which is a homotopy equivalence. In other words, observables in topological field theories are described by locally-constant factorization algebras. According to \cite[Theorem 5.4.5.9]{HA}, in the case where $M=\RR^n$ locally-constant factorization algebras are the same as $\bb E_n$-algebras, i.e. algebras over the operad of little $n$-disks.

Most quantum field theories are not topological, since they explicitly depend on a metric on $M$. To obtain a topological quantum field theory from a more general quantum field theory, and as a first example, to obtain a topological field theoretic description for Donaldson invariants of 4-manifolds, Witten \cite{WittenTQFT} introduced the idea of a topological twisting. We may start with a non-topological quantum field theory and consider its deformations to a topological quantum field theory (note that local constancy is an open condition). Deformations of factorization algebras are parametrized by their derivations of degree 1 and a large source of such derivations is given by quantum field theories which in addition carry an action of a super Poincar\'{e} algebra.

More explicitly, consider a translation-invariant factorization algebra $\obs$ on $\RR^n$ together with a degree 1 derivation $Q$ such that the generators of translations are $Q$-exact. Then we expect that the twisted factorization algebra $\obs^Q$, i.e. $\obs$ equipped with the differential $\d + Q$, corresponds to a topological field theory. In particular, one should be able to extract an $\bb E_n$-algebra from it.

In this paper we will formalize this idea, explaining in what sense we can recover $\bb E_n$-algebras and more refined structures from topologically twisted factorization algebras.  We begin with an explanatory outline of our results, followed by a more detailed summary in which we will make more precise statements.

\begin{itemize}
 \item We can construct $\bb E_n$-algebras from topologically twisted factorization algebras subject to a single condition: that the map $\obs(B_r(0)) \to \obs(B_R(0))$ on concentric balls of radii $R > r$ is an equivalence.
 \item This condition is automatically satisfied in many natural examples, including those coming from superconformal field theories.
 \item In many of these natural examples these $\bb E_n$-algebras are additionally compatible with the action of $\SO(n)$.  This means the topological field theories can be defined on any oriented $n$-manifold.
\end{itemize}

\begin{remark}
We do not know whether the twisted factorization algebra $\obs^Q$ is locally-constant, so our arguments will not use Lurie's result comparing locally-constant factorization algebras and $\bb E_n$-algebras.  Instead we will equip the local observables of the topologically twisted factorization algebra directly with an action of a specific model for the $\bb E_n$-operad.
\end{remark}

In addition to understanding twisted factorization algebras abstractly we will give a complete classification of the possible twists coming from super Poincar\'e algebra actions in dimensions 1 to 10.

\subsection{Topological twists}
Let us summarize our results in more precise terms.  Suppose $\obs$ is a translation-invariant factorization algebra on $\RR^n$ where translations are homotopically trivial (we say $\obs$ is a \emph{de Rham translation-invariant factorization algebra}: we define this precisely in Definition \ref{dR_translation_invariance_def}). There is a coloured operad $\disk_n^{\col}$ with the set of colours $\RR_{>0}$ which is defined in the same way as the operad of little $n$-disks, but where the colours correspond to the radii of the disks. Let $B_r(0)\sub \RR^n$ be the standard open disk of radius $r$ centered at the origin.

\begin{theorem}[See Theorem \ref{dR_translation_invariance_thm}]
Let $\obs$ be a de Rham translation-invariant factorization algebra on $\RR^n$. Then the collection $\{\obs(B_r(0))\}_{r>0}$ becomes an algebra over $\disk_n^{\col}$.
\end{theorem}

It is then natural to ask when a $\disk_n^{\col}$-algebra $\{A(r)\}_{r>0}$ is an $\bb E_n$-algebra, i.e. when it is possible to forget the parameterization by radii of disks. The unary operations in $\disk_n^{\col}$ give maps $A(r)\rightarrow A(R)$ for $R\geq r$. Moreover, we have a natural forgetful functor from $\bb E_n$-algebras to $\disk_n^{\col}$-algebras where the family $\{A(r)\}$ is constant.

\begin{theorem}[See Theorem \ref{E_n_theorem}]
The forgetful functor of $\infty$-categories $\ialg_{\bb E_n}\rightarrow \ialg_{\disk_n^{\col}}$ is fully faithful. Its essential image is given by those families $\{A(r)\}$ where the map $A(r)\rightarrow A(R)$ for $R\geq r$ is an equivalence.
\end{theorem}

Supersymmetric field theories provide a natural family of examples of de Rham translation-invariant factorization algebras.  
We review the definitions of supersymmetry algebras and supersymmetric field theories in Section \ref{twist_section}. In particular, a supertranslation algebra $\mf{A}$ is a super Lie algebra with even part $V=\CC^n$, the Lie algebra of translations. The group of \emph{R-symmetries} is a subgroup of the group of automorphisms of $\mf{A}$ which preserve the even part. A \emph{topological supercharge} is a square-zero odd element $Q\in\mf{A}$ such that the bracket $[Q, -]$ is surjective onto the even part of $\mf{A}$. We then have the following statement.

\begin{theorem}[See Proposition \ref{topologically_twisted_theories_are_dR_translation_invt_prop}]
Let $\obs$ be an $\mf{A}$-supersymmetric factorization algebra on $\RR^n$ and $Q\in\mf{A}$ a topological supercharge with an abelian section. Moreover, suppose $\obs^Q(B_r(0))\rightarrow \obs^Q(B_R(0))$ is a quasi-isomorphism for $R\geq r$. Then the collection $\{\obs^Q(B_r(0))\}$ forms an $\bb E_n$-algebra.
\end{theorem}

In particular this justifies the term ``topological supercharge'': by twisting a supersymmetric factorization algebra by a topological supercharge we produce a topological field theory.

Note that in the previous claim we had to choose an abelian section for $Q$, i.e. an abelian subalgebra of $\mf{A}_{\mr{odd}}$ which maps under $[Q, -]$ isomorphically onto the even part. We show in Proposition \ref{commuting_potential_prop} that in all examples of interest such a choice can be made.

Let us also remark that the factorization algebra $\obs$ we start with has a $\ZZ\times \ZZ/2$-grading corresponding to the ghost number and the fermionic charge. If we can extend a topological supercharge $Q$ to a twisting datum (meaning we choose a circle action on $\obs$ compatible with $Q$, see Definition \ref{def:twistingdatum}), then the twisted factorization algebra $\obs^Q$ will be $\ZZ$-graded. Without this choice $\obs^Q$ is merely $\ZZ/2$-graded. When we classify square-zero supercharges in Section \ref{twisted_theories_section}, we will specify for each supercharge whether we can extend it to a twisting datum.

We may also consider square-zero supercharges which are not topological, i.e. when the image of $[Q, -]$ is a proper subspace of $V$. We show in Proposition \ref{at_least_holo_prop} that the image is at least half-dimensional. We call square-zero supercharges with a half-dimensional image \emph{holomorphic}. For instance, in dimension $2$ factorization algebras twisted by holomorphic supercharges give rise to vertex algebras as explained in \cite[Chapter 5, Section 2]{Book1}. There are also intermediate cases between holomorphic and topological supercharges which we call \emph{holomorphic-topological} supercharges.  One would like to think of an $n$-dimensional field theory twisted by a holomorphic-topological supercharge whose image has codimension $k$ as a field theory on $\CC^k \times \RR^{n-2k}$ which is holomorphic along $\CC^k$ and topological along $\RR^{n-2k}$.  Such theories were originally considered by Kapustin \cite{KapustinHolo}, and the term ``holomorphic-topological'' first appeared in \cite{KapustinSaulina2}.

It is natural to ask if there are ways to ensure that $\obs^Q(B_r(0))\rightarrow \obs^Q(B_R(0))$ is a quasi-isomorphism. If $\obs^Q$ admits a homotopically trivial dilation action, the map $\obs^Q(B_r(0))\rightarrow \obs^Q(B_R(0))$ is a quasi-isomorphism by Proposition \ref{dR_dilation_prop}. In particular, we show in Lemma \ref{dilation_potential_lemma} that such a homotopically trivial dilation action exists in many \emph{superconformal} field theories. Moreover, it is enough to have this action only on the classical level (see Remark \ref{remark:classicaldilationaction}).

\subsection{Oriented field theories}

If $M$ is a framed $n$-manifold and $A$ is an $\bb E_n$-algebra, we may consider its factorization homology $\int_M A$ (see e.g. \cite{AyalaFrancis}) which is a cochain complex. In particular, this allows us to define observables on any framed manifold. If we want to define factorization homology for a manifold $M$ which is only oriented, we need to extend $A$ to a \emph{framed $\bb E_n$-algebra}, i.e. an algebra over the operad of framed little $n$-disks. In other words, $\bb E_n$ is an operad in $\SO(n)$-spaces and a framed $\bb E_n$-algebra is an $\bb E_n$-algebra equipped with a compatible action of the chains $\C_\bullet(\SO(n))$ (see e.g. \cite{SalvatoreWahl}). More generally, for a Lie group $G$ with a homomorphism $G\rightarrow \O(n)$ we may define the factorization homology $\int_M A$ over $M$ when $M$ carries a $G$-reduction of the frame bundle and $A$ is an $\bb E_n^G$-algebra. We recover the oriented case for $G=\SO(n)$ and the framed case for $G=\pt$.

We may ask when it is possible to extend the twisted factorization algebra $\obs^Q$ described above to a framed $\bb E_n$-algebra. In general, we consider supersymmetric quantum field theories, so the untwisted factorization algebra $\obs$ carries a compatible action of rotations. However, the supercharge $Q$ is never preserved by this action, so this action of rotations does not survive the twist. To correct this problem, Witten \cite{WittenTQFT} modifies the action of rotations by means of a \emph{twisting homomorphism}. A twisting homomorphism is a homomorphism $\phi\colon \SO(n)\rightarrow G_R$ from the group of rotations to the group of R-symmetries which allows us to define a new action of rotations on the factorization algebra $\obs$ using the embedding $(\id, \phi)\colon \SO(n)\rightarrow \SO(n)\times G_R$. If the supercharge $Q$ is scalar under the $\phi$-twisted $\SO(n)$-action, we say that $Q$ is compatible with the twisting homomorphism $\phi$. When we classify square-zero supercharges in Section \ref{twisted_theories_section}, we will also classify possible twisting homomorphisms and say which supercharges are compatible. Note that if $Q$ is any square-zero supercharge (not necessarily topological) in dimension $n\geq 3$ compatible with a twisting homomorphism $\phi$, then $Q$ is automatically a topological supercharge by Theorem \ref{compatible_twisting_hom_thm} (this is usually taken as the definition of a topological supercharge).

To summarize, given a topological supercharge $Q$ compatible with a twisting homomorphism $\phi$ we get an $\SO(n)$-action on $\obs^Q$. However, to define a framed $\bb E_n$-algebra we need to have an action of the chains $\C_\bullet(\SO(n))$ on $\obs^Q$. Explicitly, we are looking for $\SO(n)$-actions in the de Rham sense (see Definition \ref{dR_action_def}), i.e. where the Lie algebra action is homotopically trivial. Then we get the following statement (see Theorem \ref{twisted_theory_with_twisting_hom_thm}).

\begin{theorem}
Let $\obs$ be an $\mf{A}$-supersymmetric factorization algebra on $\RR^n$, $Q$ a topological supercharge compatible with a twisting homomorphism $\phi\colon \SO(n)\rightarrow G_R$. Suppose that the following two conditions hold.
\begin{itemize}
\item The factorization map $\obs^Q(B_r(0))\rightarrow \obs^Q(B_R(0))$ is a quasi-isomorphism.
\item The $\phi$-twisted $\SO(n)$-action on $\obs^Q$ extends to a de Rham $\SO(n)$-action.
\end{itemize}
Then the collection $\{\obs^Q(B_r(0))\}$ forms a framed $\bb E_n$-algebra.
\end{theorem}

Note that if we did not have an extension of the $\phi$-twisted $\SO(n)$-action to a de Rham $\SO(n)$-action, we would get the structure of a $\bb E_n^{\SO(n)_\delta}$-algebra, where $\SO(n)_\delta$ is the Lie group $\SO(n)$ considered with the discrete topology. To such an algebra we can apply factorization homology over an oriented manifold $M$ if the manifold carries a flat connection on its tangent bundle (for instance, it is an affine manifold). There are also versions of this theorem for any Lie group $G$ with a homomorphism $G\rightarrow \O(n)$.

One might then wonder how one can construct an extension of the $\phi$-twisted $\SO(n)$-action to a de Rham $\SO(n)$-action. There are two natural mechanisms for that:
\begin{itemize}
\item In the case of superconformal theories in dimension $2, 3$ or $4$, the homotopy trivialization may be obtained from the generators of the superconformal algebra as shown in Propositions \ref{2d_superconformal_potential}, \ref{3d_superconformal_potential} and \ref{4d_superconformal_potential} (in the case of dimension 4, this works only for certain twisting homomorphisms and topological supercharges).

\item If $\obs$ is an $\hbar$-family of factorization algebras obtained by deformation quantization of a classical field theory and the $\phi$-twisted $\SO(n)$-action on $\obs^Q$ is inner, then the $\phi$-twisted $\SO(n)$-action on the quantum observables $\obs^Q[\hbar^{-1}]$ is automatically homotopically trivial (See Remark \ref{remark:inneractionhomotopytrivial}).
\end{itemize}

One expects that a topological field theory should assign diffeomorphism invariants to manifolds. For observables this means that the complex of observables on $\RR^n$ carries a $\mr{Diff}(\RR^n)$-action in the de Rham sense. Given a supersymmetric factorization algebra $\obs$, it usually only carries a Poincar\'{e} action (i.e. an action of Killing vector fields) since the action of the theory involves the choice of a metric. The obstruction to extending this action to an action of all vector fields is given by the gravitational stress-energy tensor (see Section \ref{stressenergy_section}). If $Q$ is a topological supercharge compatible with a twisting homomorphism $\phi$, then one still has an action of Killing vector fields on $\obs^Q$. Moreover, in many supersymmetric field theories the gravitational stress-energy tensor is $Q$-exact which shows that there is an action of all vector fields on $\obs^Q$ (Lemma \ref{exact_SE_tensor_lemma}). Note, however, that we do not expect this to be a de Rham $\mr{Diff}(\RR^n)$-action.

\subsection{Structure of the paper}

Section \ref{Factorization_section} takes place in the world of homological and homotopical algebra -- the abstract mathematical machinery that we will use is concentrated entirely in this section.  We remind the reader the basics of factorization algebras (in Section \ref{perturbative_section}) and group actions thereon (in Section \ref{group_action_section}). In particular, we explain what it means to have a group action in the de Rham sense. We study factorization algebras with an action of the group of translations (Section \ref{translation_invariant_section}), and investigate how this condition can be reformulated operadically.  We then study de Rham translation-invariant factorization algebras (Section \ref{E_n_section}) and show how to obtain the data $\disk_n^{\col}$-algebras from them.  Analyzing the natural pushout of coloured operads, we show that $\bb E_n$-algebras sit fully faithfully in the $\infty$-category of $\disk_n^{\col}$-algebras (Theorem \ref{E_n_theorem}). Finally, in Section \ref{Gstructure_section} we study analogues of these results for theories with a $G$-structure.

Section \ref{twist_section} is devoted to supersymmetric field theories.  This section is much less high-tech than the previous section: other than the definition of the twist we will perform calculations relying only on ordinary Lie theory. We begin in Section \ref{susy_algebra_section} with a recollection on (complex) supersymmetry algebras in various dimensions and define what it means for a factorization algebra on $\RR^n$ to be supersymmetric. We then define topological twists of supersymmetric factorization algebras by topological supercharges (Section \ref{twist_def_section}) and explain when we get a de Rham translation-invariant theory (Proposition \ref{topologically_twisted_theories_are_dR_translation_invt_prop}).  When we additionally have the data of a twisting homomorphism from a group $G$, we explain how we obtain de Rham translation-invariant theories with $G$-structure (Section \ref{twisting_hom_section}).  In Section \ref{pure_spinor_section} we provide a recollection on pure spinors which will be later used in the classification of twists in dimensions 7 through 10. In Section \ref{stressenergy_section} we explain what the gravitational stress-energy tensor is and why supersymmetric field theories carry an action of all vector fields (which is possibly not homotopically trivial). In Section \ref{dilation_section} we explain how the additional data of an $\bb E_n$-algebra arises from a twisted theory where we additionally have a compatible action of the group of dilations.  Finally, Section \ref{superconformal_section} is devoted to superconformal algebras which are certain extensions of supersymmetry algebras. We explain there how superconformal symmetry can be used to construct nullhomotopies for rotations and to show that the factorization maps $\obs^Q(B_r(0))\rightarrow \obs^Q(B_R(0))$ are quasi-isomorphisms.

In Section \ref{twisted_theories_section} we provide a classification of possible twists of supersymmetric theories.  This section engages most closely with the techniques used in the physics literature on topological twists. We give a classification of twists in dimensions 1 through 10:
\begin{itemize}
\item In each dimension we classify orbits of square-zero supercharges under the complexified Lorentz group, R-symmetry and scaling. For each class of square-zero supercharges we specify the dimension of its image, i.e. we explain whether it is topological, holomorphic, or one of the intermediate holomorphic-topological possibilities.

\item We also classify possible twisting homomorphisms from $\SO(n)$ and their compatible topological supercharges. It turns out that these exist only in dimensions $\leq 5$ and in higher dimensions we list some twisting homomorphisms from subgroups of $\SO(n)$ which might be of interest. We also explain whether a square-zero supercharge can be promoted to a twisting datum and whether this can be done compatibly with the twisting homomorphism.

\item Finally, we provide references where twisted supersymmetric theories for a given square-zero supercharge have been studied. Let us note that there are several cases (in particular, in dimensions above 5) which have not been analyzed from either the mathematical or physical perspective to our best knowledge.
\end{itemize}

Let us mention a recent article of Eager, Saberi and Walcher \cite{EagerSaberiWalcher} appearing subsequent to the first version of this paper that also studies the possible square zero supercharges in supersymmetric field theories, and includes discussion of the geometry of the varieties of such supercharges.

\subsection*{Acknowledgements}
We would like to thank Kevin Costello, Owen Gwilliam, Vasily Pestun and Brian Williams for helpful discussions and the anonymous referees for many comments and suggestions. This research was supported in part by Perimeter Institute for Theoretical Physics. Research at Perimeter Institute is supported by the Government of Canada through Industry Canada and by the Province of Ontario through the Ministry of Economic Development \& Innovation. CE acknowledges the support of IH\'ES.  The research of CE on this project has received funding from the European Research Council (ERC) under the European Union's Horizon 2020 research and innovation programme (QUASIFT grant agreement 677368). PS was supported by the NCCR SwissMAP grant of the Swiss National Science Foundation.

\section{Factorization Algebras and $\bb E_n$-Algebras} \label{Factorization_section}

Throughout the paper we work over the ground field $\CC$ of complex numbers. Thus, all complexes, Lie algebras and so on are over $\CC$. The chain complexes of observables will be differentiable chain complexes, see \cite[Appendices B, C]{Book1}.

\subsection{Models for Perturbative Field Theory} \label{perturbative_section}

We begin by outlining the theory of factorization algebras as a model for classical and quantum field theory on $\RR^n$.  One can model a classical perturbative field theory and its quantization starting from either of two ``dual'' points of view.
\begin{enumerate}
 \item The BV--BRST complex of \emph{fields}: a sheaf of $L_\infty$-algebras on $\RR^n$ that gives a derived model for the space of gauge-equivalence classes of solutions to the equations of motion.
 \item The factorization algebra of \emph{observables}: a multiplicative version of a cosheaf of cochain complexes on $\RR^n$ built from the BV--BRST complex by taking Chevalley-Eilenberg cochains.
\end{enumerate}

Our focus for this article will be the second point of view, that of factorization algebras, but we will conclude this section with some remarks on the BV--BRST complex.

Let us begin by explaining what a factorization algebra is.  The idea is that we can model observables in a field theory on each open subset of the spacetime (i.e. of $\RR^n$), along with the data of how to extend observables from a smaller open set to a larger one.  That is, we will describe a kind of precosheaf. Fix a symmetric monoidal category $\cC$, for instance the category of vector spaces.

\begin{definition}
A \emph{prefactorization algebra} on $\RR^n$ is a ``multiplicative'' precosheaf $\obs$ with values in $\cC$ on $\RR^n$.  That is, an assignment of $\obs(U)\in\cC$ to each open set, along with structure maps $\bigotimes_{i=1}^n \obs(U_i) \to \obs(V)$ for each pairwise disjoint collection of open subsets $U_1, \ldots, U_n \sub V$ satisfying the natural compatibility conditions.  
\end{definition}

\begin{definition}
A \emph{factorization algebra} on $\RR^n$ is a prefactorization algebra which satisfies descent for covers $\{U_i\}$ of open sets $U$ satisfying the condition that for any sequence of points $x_1, \ldots, x_n$ in $U$ there is an element $U_i$ of the cover so that $\{x_1, \ldots, x_n\} \sub U_i$.  Such covers are known as \emph{Weiss covers} (see \cite[Chapter 6, Section 1]{Book1} for a precise definition).
\end{definition}

The idea here is that the local sections $\obs(U)$ on an open set $U$ model the (classical or quantum) $\CC$-valued observables of a theory that depend only on measurements inside of $U$, so in particular there is a canonical extension (our structure maps) $\obs(U) \to \obs(V)$ for every pair of open sets $U \sub V$.  The sheaf condition is intended to convey the idea that all observables are determined by their values on arbitrarily small neighbourhoods of finite sets of points.

The observables in a classical field theory carry an additional structure: a degree 1 Poisson bracket.

\begin{definition}
A \emph{$\bb P_0$-algebra} is a commutative dg algebra $A$ together with a Poisson bracket of cohomological degree $1$: that is, a Lie bracket $\{-,-\}\colon A\otimes A\rightarrow A[1]$ which is a derivation for the product.
\end{definition}

Let $\alg_{\bb P_0}$ denote the symmetric monoidal category of $\bb P_0$-algebras with tensor product given by the tensor product of the underlying complexes. 

Let us briefly comment on the passage from classical to quantum observables which is given by BV quantization introduced in \cite{BatalinVilkovisky}. We will follow the modern account given in \cite[Chapter 2.4]{Book2}.

\begin{definition}
A \emph{$\BD_0$-algebra} is a complex of flat $\CC[[\hbar]]$-modules together with a commutative multiplication and a Poisson bracket of degree $1$ satisfying the relation
\[\d(ab) = \d(a) b + (-1)^{|a|} a \d(b) + \hbar\{a, b\}.\]
\end{definition}

Note that if $\tilde{A}$ is a $\BD_0$-algebra, $\tilde{A}/\hbar$ becomes a $\bb P_0$-algebra. Conversely, given a $\bb P_0$-algebra $A$, a \emph{deformation quantization} is a $\BD_0$-algebra $\tilde{A}$ together with an isomorphism $\tilde{A}/\hbar\cong A$ of $\bb P_0$-algebras. Similarly, given a prefactorization algebra $\obs$ representing the classical observables in a classical field theory, a quantization is a prefactorization algebra $\obs^q$ valued in $\BD_0$-algebras together with an isomorphism $\obs^q/\hbar\cong \obs$.  There is a comprehensive theory governing the process of quantizing factorization algebras of classical observables via renormalization \cite{Book2}.

Throughout the paper our results will be valid in all of the following settings:
\begin{itemize}
\item (Classical). We consider prefactorization algebras valued in $\bb P_0$-algebras.

\item (Quantum). We consider prefactorization algebras valued in complexes.

\item (Deformation quantization). We consider prefactorization algebras valued in $\BD_0$-algebras.
\end{itemize}

We will conclude this section with a comment on the dual point of view, where we model a field theory starting from its space of classical fields and action functional.  We represent classical field theory from this point of view using the Batalin-Vilkovisky formalism -- we refer to Costello \cite[Chapter 5]{CostelloBook} and Costello and Gwilliam \cite{Book1} for a more detailed account.

\begin{definition} \label{classical_field_theory_def}
A \emph{classical field theory} on $\RR^n$ is a sheaf $L$ of $L_\infty$-algebras on $\RR^n$ along with an invariant antisymmetric pairing of degree $-3$ (sometimes called the \emph{antibracket}):
\[\langle -,-\rangle \colon L \otimes L \to \dens_{\RR^n}[-3].\]  
In order to match up with the physics terminology we will refer to this sheaf of Lie algebras as the \emph{BV--BRST complex} of the theory.  
\end{definition}

These BV--BRST complexes can be constructed algorithmically from the more traditional physical data of a space of fields with a gauge action and an action functional; it models the tangent complex to the derived critical locus of the action functional.  In brief, one models the space of infinitesimal gauge transformations, with its action on the space of physical fields, by a sheaf $\Phi$ of $L_\infty$-algebras (the ``BRST complex'').  The BV-BRST complex is then constructed as the sum $\Phi \oplus \Phi^![-3]$ -- where $\Phi^!$ is a topological dual space to $\Phi$ tensored by the sheaf of densities on $\RR^n$ -- with $L_\infty$ structure deformed using the classical action functional.  For details, see \cite[Section 5.3]{Book2}.

Recall that the Chevalley-Eilenberg cochain complex of an $L_\infty$ algebra $\gg$ is the cochain complex $C^\bullet(\gg) = \sym(\gg[1])^*$ with differential coming from the $L_\infty$ structure on $\gg$. An invariant symmetric pairing of degree $-3$ on $\gg$ induces a $\bb P_0$-algebra structure on $C^\bullet(\gg)$.

\begin{definition}
Let $L$ be an $L_\infty$ algebra describing a classical field theory on $\RR^n$. Its algebra of \emph{classical observables} is the prefactorization algebra $\obs$ valued in $\alg_{\bb P_0}$ that assigns to an open set $V$ the $\bb P_0$-algebra
\[\obs_L(V) := C^\bullet(L(V)).\]
\end{definition}

In fact, in reasonable circumstances this prefactorization algebra is a factorization algebra (see \cite[Chapter 6, Theorem 6.0.1]{Book1} for the precise statement).

\subsection{Smooth Group Actions} \label{group_action_section}
There are several senses in which a field theory can be thought of as equivariant with respect to the action of a Lie group $G$.  We will mostly consider group actions in a strong sense, where we have a smooth action of the Lie group relating observables on different open sets along with a compatible infinitesimal action of its Lie algebra by derivations on each open set.

We will consider derivations of algebras as follows:
\begin{itemize}
\item If $A$ is a $\bb P_0$- or a $\BD_0$-algebra, a derivation $F\colon A\rightarrow A$ is an endomorphism which is a derivation of the multiplication and the bracket.

\item If $A$ is a complex, a derivation is merely an endomorphism.
\end{itemize}

\begin{definition}[{\cite[Chapter 4, Definition 8.1.2]{Book1}}]
A degree $k$ \emph{derivation} of a prefactorization algebra $\obs$ on $\RR^n$ is a cohomological degree $k$ derivation $F_U \colon \obs(U) \to \obs(U)$ for each open set $U$ in $\RR^n$ that collectively satisfy a Leibniz rule.  That is, if $U_1, \ldots, U_m \sub V$ are disjoint open subsets of an open set, and $m_{U_1, \ldots, U_m}^V$ is the associated factorization map, we require that
\begin{align*}
F_V \circ &m_{U_1, \ldots, U_m}^V(\OO_1, \ldots, \OO_m) \\
&= \sum_{i=1}^m (-1)^{k(|\OO_1| + \cdots + |\OO_{i-1}|)} m_{U_1, \ldots, U_m}^V(\OO_1, \ldots, F_{U_i}(\OO_i), \ldots, \OO_m).
\end{align*}
\end{definition}

The complex $\der(\obs)$ of derivations of arbitrary degree is naturally a dg Lie algebra, with the bracket defined on each open set.

\begin{definition} \label{action_of_lie_alg_def}
If $\gg$ is a dg Lie algebra and $\obs$ is a prefactorization algebra on $\RR^n$, an \emph{action} of $\gg$ on $\obs$ is a morphism $\gg \to \der(\obs)$ of dg Lie algebras.
\end{definition}

We may also consider \emph{inner} actions of dg Lie algebras on the observables:
\begin{itemize}
\item If $\obs$ is a prefactorization algebra valued in $\bb P_0$ or $\BD_0$-algebras, then there is a morphism of dg Lie algebras $\obs(\RR^n)[-1] \to \der^\bullet(\obs)$ defined by $\OO \mapsto \{\OO, -\}$.

\item If $\obs$ is a prefactorization algebra valued in complexes, we equip $\obs(\RR^n)[-1]$ with the trivial bracket and let $\obs(\RR^n)[-1]\to \der^\bullet(\obs)$ be the zero map.
\end{itemize}

\begin{definition}
A $\gg$-action $\gg \to \der^\bullet(\obs)$ is \emph{inner} if it factors through a morphism $\gg \to \obs(\RR^n)[-1]$ of dg Lie algebras.
\end{definition}

\begin{remark}
Note that by definition an inner action on a prefactorization algebra valued in complexes is trivial. Similarly, an inner action on a prefactorization algebra $\obs$ valued in $\BD_0$-algebras is homotopically trivial on $\obs[\hbar^{-1}]$ (see e.g. \cite[Lemma 4.15]{KatzarkovKontsevichPantev}).
\label{remark:inneractionhomotopytrivial}
\end{remark}

The following definition generalises the definition of a smooth translation action in \cite[Chapter 4 Definition 8.1.3]{Book1}.

\begin{definition}
Let $G$ be a Lie group with an action $\rho\colon G\rightarrow \mr{Diff}(\RR^n)$ on $\RR^n$ by diffeomorphisms.  We say $G$ acts \emph{smoothly} on a prefactorization algebra $\obs$ on $\RR^n$ if we have an action of the complexified Lie algebra $\gg$ on $\obs$ by derivations, which we denote by $v \mapsto \d / \d v$, along with isomorphisms $\alpha_g \colon \obs(U) \to \obs(\rho(g)(U))$ for each $g \in G$ and each open subset $U \sub \RR^n$ which satisfy the following conditions.
\begin{itemize}
 \item For each $g_1, g_2\in G$, $\alpha_{g_1} \circ \alpha_{g_2} = \alpha_{g_1 g_2}$.
 \item If $U_1,U_2\sub V$ are disjoint open subsets such that $\rho(g)(U_1)$ and $\rho(g)(U_2)$ are disjoint for $g\in G$, then the square
 \[\xymatrix{
  \obs(U_1) \otimes \obs(U_2) \ar[r]^-{\alpha_g} \ar[d] &\obs(\rho(g)(U_1)) \otimes \obs(\rho(g)(U_2)) \ar[d] \\
  \obs(V) \ar[r]^{\alpha_g} &\obs(\rho(g)(V))
 }\]
 commutes.
 \item For each collection $U_1, \ldots, U_k\sub V$ of disjoint open subsets, the composite map 
\begin{align*}
m_{g_1, \ldots, g_k} \colon \obs(U_1) \otimes \cdots \otimes \obs(U_k) &\to \obs(\rho(g_1)(U_1)) \otimes \cdots \otimes \obs(\rho(g_k)(U_k)) \\
&\to \obs(V)
\end{align*}
 varies smoothly in the space
 \[\{(g_1, \ldots, g_k) \in G^k: \text{the } \rho(g_k)(U_k) \text{ are disjoint subsets of } V\}.\]
 \item For each $v \in \gg$ we have the following compatibility between the Lie group and Lie algebra actions:
 \[ \left( \frac \partial {\partial v} \right)_i m_{g_1, \ldots, g_k}(\OO_1, \ldots, \OO_k) = m_{g_1, \ldots, g_k}(\OO_1, \ldots, \frac \partial {\partial v} \OO_i, \ldots, \OO_k)\]
 for every $i$, where on the left we are taking the directional derivative with respect to the element
 \[(0, \ldots, 0, L_{g_i}(v), 0, \ldots, 0) \in T_{g_1, \ldots, g_k}G^k\]
 with $v$ in the $i^{\text{th}}$ slot.
\end{itemize}
\end{definition}

\subsection{Translation-invariant Prefactorization Algebras} \label{translation_invariant_section}
Consider the action of $\RR^n$ on itself by translations.

\begin{definition}
A \emph{translation-invariant} prefactorization algebra is a prefactorization algebra $\obs$ on $\RR^n$ equipped with a smooth action of the translation group $\RR^n$.
\end{definition}

Let us give an operadic reformulation of this structure. Let $B_r(x)\sub \RR^n$ be the open $n$-disk of radius $r$ with center at $x\in\RR^n$ and let $P_r(x)\sub \CC^n$ be the open $n$-polydisk of radius $r$ with center at $x\in\CC^n$. Denote by $\overline{B}_r(x)$ and $\overline{P}_r(x)$ their closures. We consider the following four operads:
\begin{itemize}
\item Let $\disk^{\man}_n$ be the $\RR_{>0}$-coloured operad in smooth manifolds defined as follows.  Define
\[\disk^{\man}_n(r_1, \dots, r_k | R)\sub \RR^{nk}\]
to be the subset of vectors $x_1, \dots, x_k\in\RR^n$ such that the map $\overline{B}_{r_1}(x_1)\sqcup \dots \sqcup \overline{B}_{r_k}(x_k)\rightarrow \overline{B}_R(0)$ is injective. The composition in $\disk^{\man}_n$ is defined by embedding of disks.

\item Let $\pdisk^{\man}_n$ be the $\RR_{>0}$-coloured operad in complex manifolds where
\[\pdisk^{\man}_n(r_1, \dots, r_k | R)\sub \CC^{nk}\]
is the subset of vectors $x_1, \dots, x_k\in \CC^n$ such that the map $\overline{P}_{r_1}(x_1)\sqcup \dots \sqcup \overline{P}_{r_k}(x_k)\rightarrow \overline{P}_R(0)$ is injective.

\item Let $\disk^{\col}_n$ be the $\RR_{>0}$-coloured operad in simplicial sets given by taking the singular set of $\disk^{\man}_n$.

\item Let $\disk_n$ be the operad in simplicial sets defined as follows. $\disk_n(k)$ has a map to $\RR_{>0}^k$ whose fiber at $(r_1, \dots, r_k)\in\RR_{>0}^k$ is $\disk^{\col}_n(r_1, \dots, r_k | 1)$. The composition is given by rescaling and embedding the disks. $\disk_n$ is a model for the $\bb E_n$-operad of little disks (see e.g. \cite[Chapter 4.1]{MarklShniderStasheff}) and $\disk_n$-algebras are $\bb E_n$-algebras.
\end{itemize}

We have a morphism of coloured operads $\disk^{\col}_n\rightarrow \disk_n$ given by collapsing all colours so that the map
\[\disk^{\col}_n(r_1, \dots, r_k | R)\longrightarrow \disk_n(k)\]
is given by sending an embedding
\[B_{r_1}(x_1)\sqcup \dots \sqcup B_{r_k}(x_k)\rightarrow B_R(0)\]
to the rescaled embedding
\[B_{r_1/R}(x_1/R)\sqcup \dots \sqcup B_{r_k/R}(x_k/R)\rightarrow B_1(0).\]

\begin{remark}
The difference between the coloured operad $\disk^{\col}_n$ and the operad $\disk_n$ is that in the former we retain the information about the radii of the disks. As will be explained in Theorem \ref{E_n_theorem}, $\disk_n$-algebras are $\disk^{\col}_n$-algebras which are locally-constant in the radial direction.
\end{remark}

Let $\cO$ be an operad in manifolds. The functor $\C^\infty$ from manifolds to convenient vector spaces is symmetric monoidal, so $\C^\infty(\cO)$ forms a cooperad in convenient vector spaces. Similarly, $\Omega^\bullet(\cO)$ forms a cooperad in complexes of convenient vector spaces (see \cite[Chapter 5, Section 1.3]{Book1}). We may then talk about algebras over $\C^\infty(\cO)$ or $\Omega^\bullet(\cO)$ in differentiable chain complexes.

The singular chain complexes $\C_\bullet(\cO)$ form a dg operad and we may talk about algebras over $\C_\bullet(\cO)$ in chain complexes. For a manifold $M$ we have an integration map $\C_\bullet(M)\otimes \Omega^\bullet(M)\rightarrow \RR$, so an algebra in differentiable chain complexes over the cooperad $\Omega^\bullet(\cO)$ gives rise to an algebra in chain complexes over the operad $\C_\bullet(\cO)$.

For a prefactorization algebra $\obs$ on $\RR^n$ we denote by $\cF_r = \obs(B_r(0))$ the observables on the standard disk of radius $r$. The following statement is shown in \cite[Chapter 4, Section 8.2]{Book1}.

\begin{prop} \label{translation_invariance_prop}
A translation-invariant prefactorization algebra $\obs$ on $\RR^n$ gives rise to an algebra $\{\cF_r\}_{r\in\RR_{>0}}$ over the coloured cooperad $\C^\infty(\disk^{\man}_n)$ in the category of differentiable chain complexes.
\end{prop}

If we have additional structure on a translation-invariant prefactorization algebra, we may enhance this collection to an algebra structure over either the coloured operad $\pdisk^{\man}_{n/2}$ or the coloured operad $\disk^{\col}_n$.

We begin with the holomorphic case. Recall that if $\obs$ is a translation-invariant prefactorization algebra on $\RR^{2n}\cong \CC^n$ we get derivations $\frac{\partial}{\partial \overline{z}_i}$ of $\obs$ for every $i=1,\dots,n$.

\begin{definition} \label{hol_translation_invariance_def}
A \emph{holomorphic translation-invariant prefactorization algebra} is a translation-invariant prefactorization algebra on $\CC^n$ equipped with the following data: for every $v\in\CC^n$ we have derivations $\eta_i$ of $\obs$ of degree $-1$ satisfying the following equations:
\begin{align*}
\d \eta_i &= \frac{\partial}{\partial \overline{z}_i} \\
\left[\frac{\partial}{\partial \overline{z}_i}, \eta_j\right] &= 0 \\
[\eta_i, \eta_j] &= 0
\end{align*}
for every $i,j=1,\dots,n$.
\end{definition}

The following statement is shown in \cite[Chapter 5, Proposition 1.3.1]{Book1}.

\begin{prop} \label{hol_translation_invariance_prop}
A holomorphic translation-invariant prefactorization algebra $\obs$ on $\CC^n$ gives rise to an algebra $\{\obs(P_r(0))\}_{r\in\RR_{>0}}$ over the coloured cooperad $\Omega^{0, \bullet}(\pdisk^{\man}_n)$ in the category of differentiable chain complexes.
\end{prop}

\begin{remark}
An algebra over the cooperad $\Omega^{0, \bullet}(\pdisk^{\man}_1)$ is closely related to the notion of a vertex algebra. We refer to \cite{HuangGeomVOA} and \cite[Chapter 5, Section 2]{Book1} for more details.
\end{remark}

\subsection{De Rham Translation Invariance} \label{E_n_section}
In this section we explain how to produce $\bb E_n$ algebras from translation-invariant prefactorization algebras where the translation action is homotopically trivialized.  The arguments in this section (but only in this section and the subsequent Section \ref{Gstructure_section}) will use techniques from homotopical algebra.  At a first pass, in order to understand the results of this section as they will be applied in the rest of the paper it will be sufficient to read the initial definitions, up to Definition \ref{dR_translation_invariance_def} and then the statement of Corollary \ref{E_n_summary_cor}.

We will be interested in actions of the homotopy type of a Lie group on a prefactorization algebra. Such an action can be defined in the following way.

\begin{definition} \label{htpically_trivial_action_def}
Let $G$ be a Lie group which acts on $\RR^n$ with an ideal $\gg_0\sub \gg$ in its complexified Lie algebra which is stable under the $G$-action. Suppose $\obs$ is a prefactorization algebra on $\RR^n$ with a smooth $G$-action. We say the $\gg_0$-action is \emph{homotopically trivial} if it is equipped with a $G$-equivariant map $\eta\colon \gg_0\rightarrow \der(\obs)[-1]$ of graded vector spaces satisfying the following equations:
\begin{align*}
\d\eta(v) &= \frac{\partial}{\partial v} \\
[\eta(v), \eta(w)] &= 0.
\end{align*}
\end{definition}

\begin{remark}
Note that if $G$ is a connected Lie group, we may rephrase $G$-equivariance of $\eta\colon \gg_0\rightarrow \der(\obs)[-1]$ via the equation
\[\eta([v, w]) = \left[\eta(v), \frac{\partial}{\partial w}\right]\]
for any $v\in \gg_0$ and $w\in\gg$.
\end{remark}

\begin{example} \label{holo_translation_invt_example}
Consider $G=\RR^{2n}\cong \CC^n$ acting on $\RR^{2n}$ by translations. We have $\gg=\CC^n\otimes_\RR\CC$ and let
\[\gg_0 = \mathrm{span}\left\{\frac{\partial}{\partial \overline{z}_i}\right\}_i\sub \gg.\]
Then a translation-invariant prefactorization algebra on $\CC^n$ for which the $\gg_0$-action is homotopically trivial is exactly a holomorphic translation-invariant prefactorization algebra (see Definition \ref{hol_translation_invariance_def}).
\end{example}

\begin{definition} \label{dR_action_def}
Let $G$ be a Lie group which acts on $\RR^n$. A \emph{$G_{\dR}$-action} on a prefactorization algebra $\obs$ on $\RR^n$ is a smooth action of $G$ on $\obs$ for which the $\gg$-action is homotopically trivial.
\end{definition}

We can now give the following topological analogue of Definition \ref{hol_translation_invariance_def}.

\begin{definition} \label{dR_translation_invariance_def}
A translation-invariant prefactorization algebra $\obs$ on $\RR^n$ is \emph{de Rham translation-invariant} if the smooth $\RR^n$-action on $\obs$ is extended to a $\RR^n_{\dR}$-action.
\end{definition}

De Rham translation invariance allow us to enhance the action of $\disk^{\man}_n$ to a locally-constant action as follows.

\begin{theorem} \label{dR_translation_invariance_thm}
A de Rham translation-invariant prefactorization algebra $\obs$ on $\RR^n$ gives rise to an algebra $\{\cF_r\}_{r\in\RR_{>0}}$ over the coloured cooperad $\Omega^\bullet(\disk^{\man}_n)$ in the category of differentiable chain complexes.
\end{theorem}

Let us first state a preliminary result.

\begin{lemma}
Let $\gg$ be a finite-dimensional dg Lie algebra with a representation $f\colon \gg\rightarrow \eend(V)$ and a map $\eta\colon \gg\rightarrow \eend(V)[-1]$ satisfying the following equations:
\begin{align*}
f(v) &= \d \eta(v) + \eta(\d v) \\
\eta([v, w]) &= [\eta(v), f(w)] \\
[\eta(v), \eta(w)] &= 0
\end{align*}

Then the map $V\rightarrow \C^\bullet(\gg, V)$ given by
\[v\mapsto \exp\left(\sum_i e^i \eta(e_i)\right) v\]
is a chain map, where $\{e_i\}$ is a basis of $\mathfrak{g}$ and $\{e^i\}$ is the dual basis of $\mathfrak{g}^*$.
\label{lm:homotopysection}
\end{lemma}

\begin{proof}[Proof of Theorem \ref{dR_translation_invariance_thm}]
We will give a proof in a slightly more general setting so that it will be amenable to generalizations. Let $\widetilde{G} = \RR^n$ be the group of translations and let $\widetilde{\gg}=\RR^n$ be its Lie algebra.

By definition the space $\disk^{\man}_n(r_1, \dots, r_k | R)$ is an open subset of $\widetilde{G}^k$, so we can identify
\[\Omega^\bullet(\disk^{\man}_n(r_1, \dots, r_k | R)) \cong \C^\bullet(\widetilde{\gg}, C^\infty(\disk^{\man}_n(r_1, \dots, r_k | R))).\]

We have maps
\[\mu^0\colon \cF_{r_1}\otimes \dots \otimes \cF_{r_k}\longrightarrow C^\infty(\disk^{\man}_n(r_1, \dots, r_k | R), \cF_R)\]
provided by Proposition \ref{translation_invariance_prop} which are $\widetilde{\gg}^{\oplus k}$-equivariant. We define
\[\mu\in \Omega^\bullet(\disk^{\man}_n(r_1, \dots, r_k | R), \hom(\cF_{r_1}\otimes \dots\otimes \cF_{r_k}, \cF_R))\]
to be the composite
\begin{align*}
\cF_{r_1}\otimes \dots\otimes \cF_{r_k} &\rightarrow \C^\bullet(\widetilde{\gg}^{\oplus k}, \cF_{r_1}\otimes \dots\otimes \cF_{r_k}) \\
&\xrightarrow{\mu^0} \C^\bullet(\widetilde{\gg}^{\oplus k}, C^\infty(\disk^{\man}_n(r_1, \dots, r_k | R), \cF_R))
\end{align*}
where the map in the first line is given by Lemma \ref{lm:homotopysection}. By construction it is a chain map and it is straightforward to check the operadic identities for $\mu$ using those for $\mu^0$.
\end{proof}

Thus,  by the above proposition we obtain an algebra over the coloured operad $\C_\bullet(\disk^{\man}_n)$ in chain complexes. In the rest of this section we explain a precise relationship between such algebras over the coloured operad $\C_\bullet(\disk^{\man}_n)$ and $\disk_n$-algebras (i.e. $\bb E_n$-algebras).

Fix a presentable symmetric monoidal $\infty$-category $\cC$ and let $\ialg_{\disk^{\col}_n}(\cC)$ be the $\infty$-category of algebras over $\disk^{\col}_n$ considered as an $\infty$-operad (see \cite[Definition 2.1.1.23]{HA}).

\begin{remark}
if $\cC$ is the $\infty$-category of chain complexes, by \cite[Theorem 4.1.1]{Hinich} we may identify $\ialg_{\disk^{\col}_n}(\cC)$ as the localization of the category of $\C_\bullet(\disk^{\col}_n)$-algebras in chain complexes with respect to quasi-isomorphisms.
\end{remark}

Given a simplicial category $\cD$ we denote by $\pi_0(\cD)$ the category obtained by applying $\pi_0$ to the Hom-sets.

\begin{definition}
Let $F\colon \cD\rightarrow \cE$ be a functor of simplicial categories. It is a \emph{Dwyer--Kan equivalence} if $F\colon \pi_0(\cD)\rightarrow \pi_0(\cE)$ is essentially surjective and
\[\hom_\cD(x, y)\rightarrow \hom_\cE(F(x), F(y))\]
is a weak equivalence of simplicial sets for every $x,y\in\cD$.
\end{definition}

Given a simplicial coloured operad $\cO$, consdering only operations of arity $1$ we obtain a simplicial category $\cO_{\textbf{1}}$. Recall from \cite{Bergner} that there is a model structure on simplicial categories with weak equivalences given by Dwyer--Kan equivalences. Also recall from \cite{Robertson,CisinskiMoerdijk} that the category of simplicial coloured operads carries a model structure where $F\colon \cO\rightarrow \cP$ is a weak equivalence if $\cO_{\textbf{1}}\rightarrow \cP_{\textbf{1}}$ is a Dwyer--Kan equivalence and for every collection of colours, \[\cO(c_1, \dots, c_k | d)\rightarrow \cP(F(c_1), \dots, F(c_k) | F(d))\] is an equivalence of simplicial sets. We will only consider simplicial coloured operads with $\cO(0)=\ast$.

We may view $\RR_{>0}$ as a poset, hence as a category. In particular, we have the associated $\RR_{>0}$-coloured operad concentrated in arity 1 that we denote by the same symbol.

\begin{prop}
The diagram of simplicial coloured operads
\[
\xymatrix{
\disk_n & \disk^{\col}_n \ar[l] \\
1 \ar[u] & \RR_{>0} \ar[l] \ar[u]
}
\]
is homotopy coCartesian.
\label{prop:diskcoCartesian}
\end{prop}
\begin{proof}
By \cite[Theorem 8.7]{CisinskiMoerdijk}, the homotopy pushout $\disk^{\col}_n\sqcup^h_{\RR_{>0}} 1$ is equivalent to the strict pushout
\[\cP = \disk^{\col}_n \underset {\RR_{>0}} \sqcup I,\]
where $\RR_{>0}\rightarrow I$ is a cofibrant replacement of $\RR_{>0}\rightarrow 1$.

The inclusion of simplicial categories into simplicial coloured operads is left Quillen, so we may assume that $I$ is concentrated in arity $1$ and has the set of colours $\RR_{>0}$. Moreover, since $I\rightarrow 1$ is an acyclic fibration, $I(r|R)$ is contractible for any $r, R\in\RR_{>0}$.

By construction the set of colours of $\cP$ is $\RR_{>0}$ and we can identify
\[\cP(r_1, \dots, r_k | R)\cong \underset{a_1, \dots, a_k\rightarrow 0}{\colim} \underset{b\rightarrow\infty}{\colim}\ \disk^{\col}_n(a_1, \dots, a_k | b)\times I(b|R)\times \prod_i I(r_i|a_i).\]

The functors $\cP_{\textbf{1}}\rightarrow I_{\textbf{1}}\rightarrow \ast$ are both Dwyer--Kan equivalences, so $\cP_{\textbf{1}}\rightarrow (\disk_n)_{\textbf{1}}\cong \ast$ is a Dwyer--Kan equivalence as well.

Let $\Conf_k(\RR^n)\sub \RR^{nk}$ be the space of configurations of $k$ ordered points on $\RR^n$. Then we have maps
\[\cP(r_1, \dots, r_k | R)\longrightarrow \disk_n(k)\longrightarrow \Conf_k(\RR^n),\]
where $\disk_n(k)\rightarrow \Conf_k(\RR^n)$ is a weak equivalence. Since the simplicial sets $I(-|-)$ are contractible, the map
\[\cP(r_1, \dots, r_k | R)\longrightarrow \underset{a_1, \dots, a_k\rightarrow 0}{\colim} \underset{b\rightarrow\infty}{\colim}\ \disk^{\col}_n(a_1, \dots, a_k | b)\]
is a weak equivalence since weak equivalences in simplicial sets are stable under filtered colimits. Hence the map $\cP(r_1, \dots, r_k | R)\rightarrow \Conf_k(\RR^n)$ is also a weak equivalence. Therefore, the map $\cP(r_1, \dots, r_k | R)\rightarrow \disk_n(k)$ described above is a weak equivalence as well and we conclude that $\cP\rightarrow \disk_n$ is a weak equivalence of simplicial coloured operads.
\end{proof}

\begin{prop}
Let
\[
\xymatrix{
\cD \ar^{\tilde{F}}[r] \ar^{\tilde{G}}[d] & \cD_1 \ar^{G}[d] \\
\cD_2 \ar^{F}[r] & \cD_0
}
\]
be a Cartesian square of $\infty$-categories. Then:

\begin{enumerate}
\item The essential image of $\tilde{F}$ consists of objects $x\in\cD_1$ such that $G(x)$ is in the essential image of $F$.

\item If $F$ is fully faithful, so is $\tilde{F}$.
\end{enumerate}
\label{prop:Cartesiansquarecategories}
\end{prop}
\begin{proof}
The embedding $\cS\rightarrow \Cat_\infty$ of $\infty$-groupoids into $\infty$-categories has a right adjoint
\[(-)^{\sim}\colon \Cat_\infty\longrightarrow \cS\]
given by taking the maximal $\infty$-groupoid $\cD^\sim$ contained in $\cD\in\Cat_\infty$. Since it is a right adjoint, it preserves fiber products, so we obtain a Cartesian square of $\infty$-groupoids of objects
\[
\xymatrix{
\cD^\sim \ar[r] \ar[d] & \cD_1^\sim \ar[d] \\
\cD_2^\sim \ar[r] & \cD_0^\sim
}
\]
which proves the first claim.

We also have a Cartesian square of $\infty$-categories
\[
\xymatrix{
\Fun(\Delta^1, \cD) \ar[r] \ar[d] & \Fun(\Delta^1, \cD_1) \ar[d] \\
\Fun(\Delta^1, \cD_2) \ar[r] & \Fun(\Delta^1, \cD_0)
}
\]
and hence a Cartesian square of $\infty$-groupoids
\[
\xymatrix{
\Fun(\Delta^1, \cD)^\sim \ar[r] \ar[d] & \Fun(\Delta^1, \cD_1)^\sim \ar[d] \\
\Fun(\Delta^1, \cD_2)^\sim \ar[r] & \Fun(\Delta^1, \cD_0)^\sim
}
\]

For any two objects $x,y\in\cD$ taking the fiber of the natural map $\Fun(\Delta^1, \cD)^\sim\rightarrow \cD^\sim\times \cD^\sim$ in the above diagram we obtain a Cartesian square of $\infty$-groupoids
\[
\xymatrix{
\hom_\cD(x, y) \ar^{\tilde{F}}[r] \ar^{\tilde{G}}[d] & \hom_{\cD_1}(\tilde{F}(x), \tilde{F}(y)) \ar^{G}[d] \\
\hom_{\cD_2}(\tilde{G}(x), \tilde{G}(y)) \ar^{F}[r] & \hom_{\cD_0}(G\tilde{F}(x), G\tilde{F}(y))
}
\]

If $F$ is fully faithful, the bottom map is an equivalence. Since the square is Cartesian, the top map is also an equivalence. In other words, in this case $\tilde{F}\colon\cD\rightarrow \cD_1$ is fully faithful.
\end{proof}

\begin{theorem} \label{E_n_theorem}
The forgetful functor $\ialg_{\disk_n}(\cC)\rightarrow \ialg_{\disk^{\col}_n}(\cC)$ is fully faithful with essential image given by algebras $\{A(r)\}_{r\in\RR_{>0}}$ where the natural map $A(r)\rightarrow A(R)$ for $r \leq R$ is an equivalence.
\end{theorem}
\begin{proof}
From Proposition \ref{prop:diskcoCartesian} we obtain a homotopy coCartesian square of $\infty$-operads
\[
\xymatrix{
\disk_n & \disk^{\col}_n \ar[l]\\
1 \ar[u] & \RR_{>0} \ar[l] \ar[u]
}
\]

Therefore, we get a Cartesian square of $\infty$-categories of algebras
\[
\xymatrix{
\ialg_{\disk_n}(\cC) \ar^{\tilde{F}}[r] \ar^{\tilde{G}}[d] & \ialg_{\disk^{\col}_n}(\cC) \ar^{G}[d] \\
\cC \ar^{F}[r] & \ialg_{\RR_{>0}}(\cC)
}
\]

The functor $F\colon \cC\rightarrow \ialg_{\RR_{>0}}(\cC)$ sends an object $V$ to a sequence $\{V\}_{r\in\RR_{>0}}$ with all maps the identities. This functor has a left adjoint $\colim\colon \ialg_{\RR_{>0}}(\cC)\rightarrow \cC$ given by evaluating the colimit of a coloured collection over the poset $\RR_{>0}$. The unit of the adjunction $\colim\circ F\rightarrow \id$ is an equivalence, so $F$ is fully faithful. The essential image is contained in coloured collections $\{V(r)\}_{r\in\RR_{>0}}$ where $V(r)\rightarrow V(R)$ is an equivalence for every pair $r\leq R$. But the unit $\id\rightarrow F\circ \colim$ is an equivalence on such coloured collections, so the essential image consists of precisely such coloured collections.

Therefore, by Proposition \ref{prop:Cartesiansquarecategories} we can conclude that the forgetful functor
\[\tilde{F}\colon \ialg_{\disk_n}(\cC)\longrightarrow \ialg_{\disk^{\col}_n}(\cC)\]
is fully faithful with essential image given by those algebras $\{A(r)\}_{r\in\RR_{>0}}$ where the morphism $A(r)\rightarrow A(R)$ is an equivalence.
\end{proof}

We can summarise the results of this section in the following way: the following is just a restating of the combination of Theorem \ref{dR_translation_invariance_thm} and Theorem \ref{E_n_theorem}.
\begin{corollary} \label{E_n_summary_cor}
If $\obs$ is a de Rham translation-invariant prefactorization algebra on $\RR^n$ such that the factorization map $\obs(B_r(0)) \to \obs(B_R(0))$ is a quasi-isomorphism for all $r \le R$, then $\obs(B_1(0))$ can be canonically equipped with the structure of a $\disk_n$-algebra.
\end{corollary}

\begin{remark}
Suppose $\obs$ arises from a classical field theory, i.e. $\obs$ is a prefactorization algebra on $\RR^n$ valued in the category of $\bb P_0$-algebras. Then $\obs(B_1(0))$ becomes an $\bb E_n$-algebra in the $\infty$-category of $\bb P_0$-algebras. By \cite[Theorem 2.22]{SafronovAdditivity} we may identify $\bb E_n$-algebras in $\bb P_0$-algebras with $\bb P_n$-algebras, i.e. commutative dg algebras with a Poisson bracket of degree $1-n$.
\end{remark}

\subsection{$G$-Structures} \label{Gstructure_section}
Fix a finite-dimensional Lie group $G$ with a homomorphism $G\rightarrow \O(n)$.

\begin{definition}
A translation-invariant prefactorization algebra $\obs$ on $\RR^n$ has a \emph{$G$-structure} if the smooth action of $\RR^n$ is extended to a smooth action of the group $G \ltimes \RR^n$.
\end{definition}

Let us now explain operadic consequences of a $G$-structure on a prefactorization algebra. Recall from \cite{SalvatoreWahl} that if $\cO$ is a $G$-operad, then we can construct a semidirect product operad $G\ltimes \cO$ such that a $G\ltimes \cO$-algebra in chain complexes is the same as an $\cO$-algebra in chain complexes equipped with a $G$-action.

Observe that the (coloured) operads $\disk^{\man}_n, \disk^{\col}_n, \disk_n$ by construction carry a natural $\O(n)$-action. Thus, we may consider their $G$-equivariant versions:
\begin{itemize}
\item $\disk^{\man, G}_n$ is the $\RR_{>0}$-coloured operad in smooth manifolds $G\ltimes \disk^{\man}_n$.

\item $\disk^{\col, G}_n$ is the $\RR_{>0}$-coloured operad in simplicial sets $G\ltimes \disk^{\man}_n$ where we consider $G$ as a simplicial group by taking its singular complex.

\item $\disk^G_n$ is the operad in simplicial sets $G\ltimes \disk_n$ where we again treat $G$ as a simplicial group.
\end{itemize}

\begin{remark}
The operad $\disk^{\SO(n)}_n$ is the operad of framed little $n$-disks and so $\disk^{\SO(n)}_n$-algebras are framed $\bb E_n$-algebras.
\end{remark}

We have the following variant of Proposition \ref{translation_invariance_prop} with $G$-structures.

\begin{prop} \label{translation_invariance_G_prop}
A translation-invariant prefactorization algebra $\obs$ on $\RR^n$ with a $G$-structure gives rise to an algebra $\{\cF_r\}_{r\in\RR_{>0}}$ over the coloured cooperad $C^\infty(\disk^{\man, G}_n)$.
\end{prop}

We may also consider de Rham translation-invariant versions.

\begin{definition}
A prefactorization algebra $\obs$ on $\RR^n$ is \emph{de Rham translation-invariant with a $G$-structure} if it carries a smooth $G\ltimes \RR^n$-action for which the Lie algebra action $\RR^n\sub \gg\ltimes \RR^n$ is homotopically trivial.
\end{definition}

\begin{definition}
A prefactorization algebra $\obs$ on $\RR^n$ is \emph{de Rham translation-invariant with a $G_{\dR}$-structure} if it carries a $(G\ltimes \RR^n)_{\dR}$-action.
\end{definition}

For $G$ a Lie group we denote by $G_\delta$ the Lie group with the discrete topology. Then we have the following operadic interpretations of $G$- and $G_{\dR}$-actions.

\begin{theorem} \label{G-structured_operad_algebra_thm}
$ $
\begin{itemize}
\item A de Rham translation-invariant prefactorization algebra $\obs$ on $\RR^n$ with a $G$-structure gives rise to an algebra $\{\cF_r\}_{r\in\RR_{>0}}$ in differentiable chain complexes over the coloured cooperad $\Omega^\bullet(\disk^{\man, G_\delta}_n)$.

\item A de Rham translation-invariant prefactorization algebra $\obs$ on $\RR^n$ with a $G_{\dR}$-structure gives rise to an algebra $\{\cF_r\}_{r\in\RR_{>0}}$ in differentiable chain complexes over the coloured cooperad $\Omega^\bullet(\disk^{\man, G}_n)$.
\end{itemize}
\end{theorem}
\begin{proof}
The proof is identical to the proof of Theorem \ref{dR_translation_invariance_thm} where we perform the following modifications:
\begin{itemize}
\item If $\obs$ has a $G$-structure, we let $\widetilde{G} = G\ltimes \RR^n$ and $\widetilde{\gg} = \RR^n$.

\item If $\obs$ has a $G_{\dR}$-structure, we let $\widetilde{G} = G\ltimes \RR^n$ and $\widetilde{\gg} = \gg\ltimes \RR^n$.
\end{itemize}
\end{proof}

As before, an algebra in differentiable chain complexes over $\Omega^\bullet(\disk^{\man, G}_n)$ gives rise to an algebra in chain complexes over $\C_\bullet(\disk^{\col, G}_n)$.

Finally, we have an analogue of Theorem \ref{E_n_theorem} for prefactorization algebras with a $G_{\dR}$-structure.

\begin{theorem} \label{G-structured_E_n_theorem}
The forgetful functor $\ialg_{\disk^G_n}(\cC)\rightarrow \ialg_{\disk^{\col, G}_n}(\cC)$ is fully faithful with essential image given by algebras $\{A(r)\}_{r\in\RR_{>0}}$ where the natural map $A(r)\rightarrow A(R)$ for $r\leq R$ is a quasi-isomorphism.
\end{theorem}

We can summarize Theorem \ref{G-structured_operad_algebra_thm} and Theorem \ref{G-structured_E_n_theorem} in the following way.

\begin{corollary} \label{Efr_n_summary_cor}
If $\obs$ is a de Rham translation-invariant prefactorization algebra on $\RR^n$ with an $\SO(n)_{\dR}$-structure such that the factorization map $\obs(B_r(0)) \to \obs(B_R(0))$ is a quasi-isomorphism for all $r \le R$, then $\obs(B_1(0))$ can be canonically equipped with the structure of a framed $\bb E_n$-algebra.
\end{corollary}

\section{Supersymmetry and Topological Twists} \label{twist_section}
At this point we will change gears and begin to apply the results of the previous section to our main focus: the theory of twists of supersymmetric field theories.  We will begin by explaining what supersymmetric theories are in the factorization algebra context, giving the definition of twisting, then apply our results to explain the circumstances under which the observables in topologically twisted supersymmetric field theories admit $\bb E_n$-structures and either $G$- or $G_{\dR}$-structures.

\subsection{Supersymmetry Algebras} \label{susy_algebra_section}

We refer to \cite{DeligneSpinors} for some details on the material presented here.

Let $V_\RR=\RR^n$ endowed with a nondegenerate symmetric bilinear form. Denote by $V=V_\RR\otimes \CC$ its complexification and consider the complex Lie algebra $\so(V)$. Let us recall that if $n=2q+1$, $\so(V)$ has a distinguished fundamental representation $S$ called the \emph{spin representation}. If $n=2q$, $\so(V)$ has a pair of distinguished fundamental representations $S_+$ and $S_-$ called the \emph{semi-spin representations} which are exchanged under the outer automorphism group (recall that it is $\ZZ/2$ for $q\geq 5$).  Elements of $S$ (in odd dimensions) or $S_+ \oplus S_-$ (in even dimensions) are called \emph{Dirac spinors}, and elements of $S_+$ or $S_-$ (in even dimensions) are called \emph{Weyl spinors}.

\begin{definition}
A representation $\Sigma$ of $\so(V)$ is \emph{spinorial} if it splits as a sum of semi-spin (if $n=2q$) or spin (if $n=2q+1$) representations.
\end{definition}

Note that under the standard inclusion $\so(n-1)\sub \so(n)$ a spinorial representation restricts to a spinorial representation.

We will use certain canonical $\so(V)$-equivariant nondegenerate pairings $\Gamma$ from spin representations to $V$. The precise form of the pairing depends on $\dim(V)$ modulo 8:
\begin{itemize}
\item If $\dim(V)\equiv 0, 4\pmod 8$, we have a pairing $\Gamma\colon S_+\otimes S_-\rightarrow V$.

\item If $\dim(V)\equiv 1, 3\pmod 8$, we have a pairing $\Gamma\colon \sym^2(S)\rightarrow V$.

\item If $\dim(V)\equiv 2\pmod 8$, we have pairings $\Gamma_{\pm}\colon \sym^2(S_{\pm})\rightarrow V$.

\item If $\dim(V)\equiv 5, 7\pmod 8$, we have a pairing $\Gamma\colon \wedge^2(S)\rightarrow V$.

\item If $\dim(V)\equiv 6\pmod 8$, we have pairings $\Gamma_{\pm}\colon \wedge^2(S_{\pm})\rightarrow V$.
\end{itemize}

We then have the following proposition characterizing possible spinorial representations.

\begin{prop} \label{prop:spinorialreps}
Suppose $\Sigma$ is a spinorial representation of $\so(V)$ equipped with a nondegenerate $\so(V)$-equivariant pairing $\sym^2(\Sigma)\rightarrow V$. Then $\Sigma$ has the following form.
\begin{itemize}
\item Suppose $\dim(V)\equiv 0, 4\pmod 8$. There is a vector space $W$ such that
\[\Sigma \cong S_+\otimes W\oplus S_-\otimes W^*.\]

\item Suppose $\dim(V)\equiv 1,3\pmod 8$. There is a vector space $W$ equipped with a nondegenerate symmetric bilinear pairing such that
\[\Sigma \cong S\otimes W.\]

\item Suppose $\dim(V)\equiv 2 \pmod 8$. There is a pair of vector spaces $W_+$ and $W_-$ equipped with nondegenerate symmetric bilinear pairings such that
\[\Sigma \cong S_+\otimes W_+\oplus S_-\otimes W_-.\]

\item Suppose $\dim(V)\equiv 5, 7\pmod 8$. There is a symplectic vector space $W$ such that
\[\Sigma \cong S\otimes W.\]

\item Suppose $\dim(V)\equiv 6\pmod 8$. There is a pair of symplectic vector spaces $W_+$ and $W_-$ such that
\[\Sigma \cong S_+\otimes W_+\oplus S_-\otimes W_-.\]
\end{itemize}
\end{prop}

Using Proposition \ref{prop:spinorialreps} we can define spinorial representations using a single number or a pair of numbers as follows:
\begin{itemize}
\item If $\dim(V)\equiv 0, 1, 3, 4\pmod 8$, we let $\mc{N} = \dim(W)$.

\item If $\dim(V)\equiv 2 \pmod 8$, we let $\mc{N}_{\pm}=\dim(W_{\pm})$.

\item If $\dim(V)\equiv 5, 7\pmod 8$, we let $2\mc{N} = \dim(W)$.

\item If $\dim(V)\equiv 6\pmod 8$, we let $2\mc{N}_{\pm} = \dim(W_{\pm})$.
\end{itemize}

We are now ready to define supertranslation Lie algebras associated to spinorial representations $\Sigma$ equipped with a nondegenerate pairing $\Gamma \colon \sym^2(\Sigma)\rightarrow V$.

\begin{definition}
The \emph{supertranslation Lie algebra} associated to $\Sigma$ and $\Gamma$ is the $\so(V)$-equivariant super Lie algebra $T = V\oplus \Pi \Sigma$ with the only nontrivial bracket given by the pairing $\Gamma \colon \sym^2(\Sigma)\rightarrow V$.
\end{definition}

We can also define super Poincar\'e algebra as follows. Fix a nondegenerate symmetric bilinear pairing on $V_\RR$. Recall that the \emph{Poincar\'e group} is $\ISO(V_\RR) = \Spin(V_\RR)\ltimes V_\RR$, so that its complexified Lie algebra, the \emph{Poincar\'{e} algebra}, is $\mf{iso}(V) = \so(V)\ltimes V$.

\begin{definition}
The \emph{super Poincar\'e algebra} associated to $\Sigma$ and $\Gamma$ is the super Lie algebra $\so(V)\ltimes T$ where $T$ is the supertranslation Lie algebra.
\end{definition}

Note that the super Poincar\'e algebra carries a large outer automorphism group corresponding to automorphisms of the auxiliary space $W$. Many field theories of interest carry an action of a subgroup of this outer automorphism group. So, we will define a supersymmetry algebra to incorporate those symmetries as well.

Fix the following data:
\begin{enumerate}
 \item A spinorial representation $\Sigma$ of $\so(V)$.

 \item An $\so(V)$-equivariant nondegenerate pairing $\Gamma\colon \sym^2(\Sigma)\to V$.

 \item A complex Lie group $G_R$, the group of \emph{R-symmetries}, which is a subgroup of the group of outer automorphisms of the super Poincar\'e algebra that act trivially on the even part.
\end{enumerate}

\begin{definition}
The \emph{supersymmetry algebra} associated to a choice of $\Sigma, \Gamma$ and $G_R$ is the super Lie algebra
\[\mf A = (\mf{iso}(V) \oplus \gg_R) \ltimes \Pi \Sigma.\]
\end{definition}

With this discussion in hand, we can define what it means for a prefactorization algebra to be supersymmetric. From now on all prefactorization algebras are $\ZZ/2$-graded in addition to the original cohomological grading. In the symmetric monoidal structure on $\ZZ/2$-graded complexes that we consider, there is a Koszul sign coming from both gradings.

\begin{definition}
We say a prefactorization algebra $\obs$ on $\RR^n$ is \emph{$\mf A$-supersymmetric} for the supersymmetry algebra $\mf A$ if it admits a smooth $\ISO(p,q) \times G_R$-action, where $\ISO(p,q)$ acts on $\RR^n$ by isometries and $G_R$ acts trivially, so that the action of $\mf{iso}(n) \oplus \gg_R$ extends to an action $\nu\colon \mf A\to\der(\obs)$ of $\mf A$.
\end{definition}

\subsection{Topological Twists of Supersymmetric Theories} \label{twist_def_section}

We now proceed to give a definition of the topological twist of an $\mf A$-supersymmetric factorization algebra $\obs$ on $\RR^n$ with respect to an odd element $Q \in \mf A$. We refer to \cite[Section 15]{CostelloSUSY} and \cite{ElliottYoo} for more details on these ideas.

\begin{definition}
A \emph{square-zero supercharge} is an odd element $Q\in \mf A$ satisfying $[Q,Q]=0$.  We say $Q$ is \emph{topological} if the map $[Q,-] \colon \Sigma \to \RR^n$ is surjective.  We say $Q$ is \emph{holomorphic} if $n$ is even and the image of the map $[Q,-]$ has dimension $n/2$.
\end{definition}

Note that all supercharges are at least holomorphic, i.e. the image of $[Q,-]$ has dimension at least $n/2$ (see Proposition \ref{at_least_holo_prop}).

\begin{definition}
A \emph{twisting datum} is a pair $(\alpha, Q)$ where $Q$ is a square-zero supercharge and $\alpha\colon U(1)\rightarrow G_R$ is a homomorphism so that $Q$ has $\alpha$-weight 1.
\label{def:twistingdatum}
\end{definition}

Given a twisting datum, we can construct the twisted theory in the following way.

\begin{definition}
A \emph{graded mixed complex} is a cochain complex $M$ equipped with an additional $(\ZZ \times \ZZ/2)$-grading and a square-zero endomorphism $Q$ of $(\ZZ \times \ZZ/2)$-degree $(1, 1)$ and cohomological degree $0$.
\end{definition}

Given a graded mixed complex $M$, we construct a $\ZZ/2$-graded complex $M^Q$ as
\[M^Q = \underset{m\rightarrow\infty}{\colim}\prod_{n\geq -m} (\Pi^n M(n)[-n])\]
with the differential $\d_Q = \d + Q$, where $\d$ is the original differential on $M$ and $M(n)$ is the weight $n$ (i.e. degree $n$ with respect to the auxiliary $\ZZ$-grading) part of $M$. The shifts have been arranged so that $\d_Q$ has cohomological degree $1$ and even $\ZZ/2$-grading. This defines a lax symmetric monoidal functor from graded mixed complexes to $\ZZ/2$-complexes. This functor is analogous to the Tate realization functor introduced in \cite[Section 1.5]{CPTVV}.

Given a twisting datum $(\alpha, Q)$, $\obs(U)$ becomes a graded mixed complex for any open subset $U\sub \RR^n$ where the extra grading is given by $\alpha$ and the mixed structure by $\nu(Q)$. Thus, we may construct the twisted complex $\obs(U)^Q$. Since the functor $(-)^Q$ is lax symmetric monoidal, it will send prefactorization algebras to prefactorization algebras and will preserve $\bb P_0$- or $\BD_0$-structures. We call $\obs^Q$ the \emph{twisted prefactorization algebra}.

\begin{remark}
Suppose we did not choose the homomorphism $\alpha$, so that $\obs$ did not have an extra grading. Then we may still construct $\obs^Q$, but the cohomological $\ZZ\times \ZZ/2$-grading will collapse to a cohomological $\ZZ/2$-grading.
\label{rmk:twistedgrading}
\end{remark}

\begin{remark}[Twisting the BV--BRST Complex]
If the factorization algebra $\obs$ is constructed as the Chevalley-Eilenberg complex of a classical field theory $L$ as in Definition \ref{classical_field_theory_def} then we could instead define the twisted factorization algebra by twisting the sheaf $L$ of $L_\infty$-algebras with respect to the twisting datum $(\alpha, Q)$, and then taking the Chevalley-Eilenberg complex of the result.  This is the point of view taken in \cite{CostelloSUSY} and \cite{ElliottYoo}.  
\end{remark}

\begin{prop} \label{topologically_twisted_theories_are_dR_translation_invt_prop}
Suppose a prefactorization algebra $\obs$ on $\RR^n$ is $\mf{A}$-supersymmetric Let $(\alpha, Q)$ be a twisting datum where $Q$ is topological and fix an abelian section $\mf{a}\sub \mf A$, i.e. an odd abelian subalgebra such that $[Q, -]\colon \mf{a}\rightarrow V$ is an isomorphism. Then $\obs^Q$ is de Rham translation-invariant.  In particular, if $\obs^Q(B_r(0))\rightarrow \obs^Q(B_R(0))$ is a quasi-isomorphism for $R\geq r$ then the collection $\{\obs^Q(B_r(0))\}$ forms an $\bb E_n$-algebra.
\end{prop}

\begin{proof}
Recall from Definition \ref{htpically_trivial_action_def} that in order to extend the $V$-action on $\obs^Q$ to a $V_\dR$-action we need to specify a linear map $\eta \colon V \to \der(\obs)[-1]$ so that $[\eta(v), \eta(w)] = 0$, $[\nu(v), \eta(w)]=0$ for all $v, w \in V$ and $\d_Q \eta(v) = \frac {\partial}{\partial v}$.  

We construct $\eta$ using the action of $\mf{a}\sub \mf A$.  Since $Q$ has $\alpha$-weight 1 and $V$ has $\alpha$-weight 0, the subalgebra $\mf a$ has $\alpha$-weight $-1$, therefore the action $\nu \colon \mf a \to \der(\obs)$ becomes a map $\mf a \to \der(\obs^Q)[-1]$ of the desired degree.  We define $\eta \colon V \to \der(\obs^Q)[-1]$ by composing this map with the inverse to the isomorphism $[Q, -]\colon \mf{a}\rightarrow V$.  The fact that $\nu$ was a Lie map implies that $[\eta(v), \eta(w)] = 0$ and $[\nu(v), \eta(w)] = 0$ for all $v, w \in V$. Moreover, we have
\begin{align*}
\d_Q \eta(v) &= \d \eta(v) + [\nu(Q), \eta(v)] \\
&= 0 + \nu(v) = \frac{\partial}{\partial v}
\end{align*}
since $\eta(v)$ is a cocycle in $\der(\obs)$.  The final sentence of the statement is a direct application of Corollary \ref{E_n_summary_cor}.
\end{proof}

We can make an analogous statement in the case where $Q$ is a holomorphic supercharge, where instead of a de Rham translation invariant factorization algebra the twist becomes holomorphic translation invariant as in Example \ref{holo_translation_invt_example}.
\begin{prop}
Suppose a prefactorization algebra $\obs$ on $\RR^{2n}$ is $\mf{A}$-supersymmetric. Let $(\alpha, Q)$ be a twisting datum where $Q$ isholomorphic, and let $V_{\mr{hol}} \sub V$ be the $n$-dimensional image of the map $[Q,-]$.  Fix $\mf{a}\sub \mf A$, an odd abelian subalgebra such that $[Q, -]\colon \mf{a}\rightarrow V_{\mr{hol}}$ is an isomorphism. Then $\obs^Q$ is holomorphic translation-invariant.
\end{prop}

The proof is identical to that of Proposition \ref{topologically_twisted_theories_are_dR_translation_invt_prop}.

Let us now show that we can find the abelian section $\mf{a}\sub \mf{A}$ in the examples of interest to us.

\begin{prop} \label{commuting_potential_prop}
Let $Q\in\mf{A}$ be a square-zero supercharge and let $S_Q \sub \CC^n$ be the image of the bracket $[Q,-] \colon \Sigma \to \CC^n$.  If either
\begin{enumerate}
 \item $n \equiv 5, 6, 7 \pmod 8$,
 \item $n \equiv 0, 4 \pmod 8$ and $Q$ lies in $S_+\otimes W\sub \Sigma$ or $S_-\otimes W^*\sub \Sigma$, 
 \item $n \equiv 1, 3 \pmod 8$ and $\mc N \geq 2$, or
 \item $n \equiv 2 \pmod 8$ and $\mc N_+$ or $\mc N_- \geq 2$,
\end{enumerate}
then there is a section $\sigma\colon S_Q\rightarrow \Sigma$ of $[Q, -]\colon \Sigma\rightarrow S_Q$ with an abelian image.
\end{prop}

\begin{proof}
Let $\Sigma$ be a spinorial representation of $\so(n)$, and let $\Gamma_n \colon \sym^2(\Sigma) \to \CC^n$ be a nondegenerate $\so(n)$-equivariant pairing. If we consider the subalgebra $\so(n-1)\sub \so(n)$, we obtain a nondegenerate $\so(n-1)$-equivariant pairing
\[\Gamma_{n-1}\colon \sym^2(\Sigma)\xrightarrow{\Gamma_n} \CC^n\rightarrow \CC^{n-1}.\]
In fact, every such $\so(n-1)$-equivariant $\CC^{n-1}$-valued pairing on $\Sigma$ arises in this way. As such, if we choose $Q$ in $\Sigma$ and an abelian section $\sigma\colon S_Q\rightarrow \Sigma$, in dimension $n$, we obtain an abelian section in dimension $n-1$. In general, $\Gamma_n(Q, Q)\neq 0$, so we will prove the statement for all supercharges which will not necessarily square to zero. We can now build abelian subalgebras on a case-by-case basis.

\begin{itemize}
\item Suppose $n \equiv 0, 4 \pmod 8$, so $\Sigma \iso S_+\otimes W \oplus S_-\otimes W^*$.  Suppose $Q \in S_+\otimes W$. Then any section $\sigma\colon S_Q\rightarrow S_-\otimes W^*$ has abelian image.

\item Suppose $n \equiv 7 \pmod 8$. Given a spinorial representation $\Sigma$ of $\so(n)$, we can write it as a restriction of the spinorial representation $S_+\otimes W\oplus S_-\otimes W^*$ of $\so(n+1)$ where both $S_+$ and $S_-$ restrict to the spin representation of $\so(n)$. We may always arrange $Q\in \Sigma$ to lie in $S_+\otimes W$, so by the previous point we get an abelian section.

The same analysis applies to the case $n\equiv 3\pmod 8$ with $\mc{N} \geq 2$.

\item Suppose $n \equiv 2, 6 \pmod 8$. In this case $\Sigma = S_+\otimes W_+\oplus S_-\otimes W_-$. Suppose $\dim(W_+)\geq \dim(W_-)$ and consider the $\so(n)$-equivariant pairing $\Gamma_n$ on $S_+\otimes W_+\oplus S_-\otimes W_+$. It arises by restriction from an $\so(n+1)$-equivariant pairing $\Gamma_{n+1}$ on $(S_+\oplus S_-)\otimes W_+$, where $S_+\oplus S_-$ is the spin representation of $\so(n+1)$. We can choose a surjective map $W_+\rightarrow W_-$ compatible with the pairings and lift $Q$ to an element $\widetilde{Q}\in (S_+\oplus S_-)\otimes W_+$. If $n\equiv 6 \pmod 8$ or $n\equiv 2\pmod 8$ and $\mc{N}_+\geq 2$, then by the previous results we may find an abelian section $S_{\widetilde{Q}}\rightarrow (S_+\oplus S_-)\otimes W_+$. Post-composing it with the projection $W_+\rightarrow W_-$, we obtain an abelian section for $Q$.

\item Suppose $n \equiv 1, 5 \pmod 8$. In this case $\Sigma = S\otimes W$ which arises as a restriction of the spinorial representation $S_+\otimes W$ with its pairing in dimension $n+1$. Thus, if $n\equiv 1\pmod 8$ and $\mc{N}\geq 2$ or $n\equiv 5\pmod 8$, we obtain an abelian section.
\end{itemize}
\end{proof}

\begin{remark}
This argument includes all the square-zero supercharges in dimensions less than 8 with the exception of supercharges in dimension 4 which contain components in both $S_+$ and $S_-$. In the latter case it is easy to construct an abelian section by hand. All the topological supercharges we will discuss in Section \ref{twisted_theories_section} are included in this result. The only square-zero supercharges that are not included are non-topological supercharges in dimensions 8, 9 and 10.  We expect that the hypotheses in our statement can be weakened; we are not aware of any counter-examples to the existence of an abelian section to $[Q,-]$ in any dimension.
\end{remark}

\subsection{Twisting Homomorphisms} \label{twisting_hom_section}

Supersymmetric theories in particular include an action of the spin group $\Spin(V_\RR)$.  However, this spin action does not survive the twisting procedure.  An action of a Lie group $G$ on the observables of a supersymmetric field theory is only well-defined after twisting by $Q$ only if it preserves $Q$, i.e. we have $[\nu(Q), \nu(X)] = 0$ for every $X \in \gg$. We can see the following.

\begin{prop} \label{twisting_hom_gives_G_action_prop}
If $\obs$ is an $\mf A$-supersymmetric factorization algebra on $\RR^n$ and $Q$ is a square-zero supercharge, the twisted factorization algebra $\obs^Q$ admits a smooth action of $\stab(Q) \sub \Spin(V_\RR) \times G_R$.
\end{prop}

Since the representation $\Sigma$ does not contain a trivial subrepresentation, a non-zero $Q$ is never preserved under $\Spin(V_\RR)$. Nevertheless, one can define a $\Spin(V_\RR)$-action on $\obs^Q$ by means of a \emph{twisting homomorphism}. Consider a Lie group $G$ equipped with a homomorphism $G\rightarrow \Spin(V_\RR)$. In particular, $G$ acts on supersymmetric field theories on $\RR^n$.

\begin{definition}
Let $\obs$ be an $\mf A$-supersymmetric factorization algebra on $\RR^n$. A \emph{twisting homomorphism} is a homomorphism $\phi\colon G\to G_R$. A square-zero supercharge $Q$ is \emph{compatible} with the twisting homomorphism $\phi$ if the action of the subgroup $(\id,\phi) \colon G \to G \times G_R$ on $Q$ is trivial.
\end{definition}

If a supercharge $Q$ is compatible with a twisting homomorphism, then $\phi$ defines a smooth $G$-action on the twisted factorization algebra $\obs^Q$.

To put this in context, in the physics literature a topological twist is often defined to be a \emph{pair} consisting of a twisting homomorphism $\phi$ and a compatible topological supercharge $Q$ (as in Witten's original construction \cite{WittenTQFT}).  In fact, this is sometimes taken as the definition of a topological supercharge. This is justified by the following result. 

\begin{theorem} \label{compatible_twisting_hom_thm}
Let $Q$ be a square-zero supercharge in a supersymmetry algebra in dimension $n\geq 3$ which is compatible with a twisting homomorphism $\phi\colon \Spin(V_\RR)\to G_R$. Then $Q$ is a topological supercharge.
\end{theorem}

\begin{proof}
By construction $\Sigma$ carries an action of $\so(n)\oplus \gg_R$. In particular, we may consider a modified $\so(n)$-action on $\Sigma$ via the twisting homomorphism $(\id, \phi)\colon \so(n)\to\so(n)\oplus \gg_R$. We denote this subalgebra by $\so(n)'$.

By assumption $\CC Q\sub \Sigma$ is a trivial $\so(n)'$-subrepresentation. Therefore, $(\CC Q)\otimes \Sigma\sub \Sigma\otimes \Sigma$ is also a subrepresentation. The map $\Gamma\colon\Sigma\otimes \Sigma\to \CC^n$ is $\so(n)'$-equivariant since $G_R$ acts by automorphisms of the Lie structure. Therefore, the image of $[Q, -]\colon \Sigma\rightarrow \CC^n$ is an $\so(n)$-subrepresentation. By nondegeneracy of the pairing $\Gamma$, the image of $[Q, -]$ is nonzero. But for $n > 2$ the representation $\CC^n$ of $\so(n)$ is irreducible, so $[Q, -]$ must be surjective.
\end{proof}

\begin{remark}
The theorem is false in dimension 2.  For instance, choose the spinorial representation $\Sigma = S_+^{\oplus 2}$ and the R-symmetry group $G_R = \SO(2)$.  If one takes the twisting homomorphism $\phi \colon \Spin(2) \to \SO(2)$ to be the standard cover, then there is a compatible supercharge.  However in this supersymmetry algebra no supercharges are topological.  We will discuss this supersymmetry algebra more completely in Section \ref{twisted_theories_section}.
\end{remark}

The formalism of twisting homomorphisms allows us to apply Corollary \ref{Efr_n_summary_cor} in the context of topologically twisted theories.
\begin{theorem} \label{twisted_theory_with_twisting_hom_thm}
Let $\obs$ be an $\mf{A}$-supersymmetric factorization algebra on $\RR^n$, $Q$ a square-zero supercharge compatible with a twisting homomorphism $\phi\colon \SO(n)\rightarrow G_R$. Suppose that the following two conditions hold.
\begin{itemize}
\item The factorization map $\obs^Q(B_r(0))\rightarrow \obs^Q(B_R(0))$ is a quasi-isomorphism.
\item The $\phi$-twisted $\SO(n)$-action on $\obs^Q$ extends to a de Rham $\SO(n)$-action.
\end{itemize}
Then the collection $\{\obs^Q(B_r(0))\}$ forms a framed $\bb E_n$-algebra.
\end{theorem}

\begin{proof}
This is a direct application of Corollary \ref{Efr_n_summary_cor}.  By Theorem \ref{compatible_twisting_hom_thm} $Q$ is a topological supercharge.  Then by Proposition \ref{topologically_twisted_theories_are_dR_translation_invt_prop} and Proposition \ref{commuting_potential_prop} the $Q$-twisted factorization algebra is de Rham translation invariant.  The hypotheses in the statement of the theorem ensure that we can apply Corollary \ref{Efr_n_summary_cor}.
\end{proof}

\subsection{Pure Spinors} \label{pure_spinor_section}
The subspace of square-zero supercharges in any dimension is closely related to the space of \emph{pure spinors} introduced by Cartan and Chevalley.  We will make reference to this theory when we discuss $\mc N=1$ or $\mc N=(1,0)$ square-zero supercharges in high dimensions.  We refer to \cite[Chapter III]{Chevalley}, \cite{BudinichTrautman} and \cite{Meinrenken} for a reference on pure spinors.

Let $S$ be an irreducible representation of the Clifford algebra $\mr{Cl}(V)$. Here $S$ coincides with the Dirac spinor representation of $\so(V)$: that is, in odd dimensions $S$ is the spin representation of $\so(V)$ and in even dimensions $S\cong S_+\oplus S_-$ is the sum of the semi-spin representations of $\so(V)$.

Let $L\sub V$ be a maximal isotropic subspace and let $\SO(V)_L\sub \SO(V)$ be the subgroup preserving $L\sub V$. We can identify $\SO(V)_L\cong \GL(L)\ltimes G_L$ where $G_L$ has the following description:
\begin{itemize}
\item If $n$ is even, $G_L=\wedge^2 L$.
\item If $n$ is odd, $G_L$ is a central extension $0\rightarrow \wedge^2 L\rightarrow G_L\rightarrow L\rightarrow 0$. Explicitly, $G_L$ is the group of pairs $(v, \omega)\in (L\oplus \wedge^2 L)$ with the multiplication rule
\[(v_1, \omega_1)(v_2, \omega_2) = (v_1 + v_2, \omega_1 + \omega_2 + v_1\wedge v_2).\]
\end{itemize}

The preimage $\Spin(V)_L\sub \Spin(V)$ of $\SO(V)_L\sub \SO(V)$ under $\Spin(V)\rightarrow \SO(V)$ is identified with $\ML(L)\ltimes G_L$, where $\ML(L)$ is the metalinear group. Then we can identify the restriction of the spinor representation to the subgroup $\Spin(V)_L$ as
\begin{equation}
S\cong \wedge^\bullet L\otimes \det(L)^{-1/2}
\label{eq:spinormodule}
\end{equation}
where $v\in L$ acts by multiplication by $1+v$ and $\omega\in\wedge^2 L$ acts by multiplication by $e^\omega$.

Pick a nondegenerate pairing $\langle -, -\rangle\colon S\otimes S\rightarrow \CC$ such that $\langle \rho(v)Q_1, \rho(v)Q_2 \rangle = \langle Q_1, Q_2 \rangle$ for $v$ a unit vector, which is unique up to rescaling.  Here we write $\rho(v)$ to denote the Clifford multiplication by the vector $v \in V$. Then we define a nondegenerate pairing $\Gamma\colon S\otimes S\rightarrow V$ by
\[(v, \Gamma(Q_1, Q_2)) = \langle \rho(v) Q_1, Q_2\rangle\]
for $v\in V$ and $Q_1,Q_2\in S$, where $(-, -)$ denotes the inner product on $V$.  These pairings coincide with the symmetric and anti-symmetric pairings we described at the beginning of Section \ref{susy_algebra_section}.

\begin{definition}
For a spinor $Q \in S$, its \emph{nullspace} $T_Q\sub V$ is the space of vectors $v \in V$ such that $\rho(v) Q = 0$.
\end{definition}

It is easy to see that for any nonzero spinor $Q\in S$, the nullspace $T_Q\sub V$ is isotropic. This motivates the following definition.

\begin{definition}
The spinor is called \emph{pure} if $T_Q$ is a maximal isotropic subspace, i.e. $\dim T_Q = \lfloor \frac n2 \rfloor$.
\end{definition}

From the description \eqref{eq:spinormodule} of the spinor module $S$ as a $\Spin(V)_L$-representation we get the following statement (see also \cite[Lemma 1]{Igusa}).

\begin{prop} \label{prop:purespinorstabilizer}
Let $L\sub V$ be a maximal isotropic subspace.
\begin{enumerate}
\item The subspace $\{Q\in S\colon \rho(L)Q=0\}\sub S$ is one-dimensional and isomorphic to $\det(L)^{1/2}$ as a $\Spin(V)_L$-representation.

\item Suppose $Q\in S$ is a pure spinor such that $\rho(L) Q = 0$. Then its stabilizer in $\Spin(V)$ is isomorphic to $(\SL(L)\times \ZZ/2\ZZ)\ltimes G_L$.
\end{enumerate}
\end{prop}

Now, let us investigate the image of the $\Gamma$ pairing with a pure spinor, i.e. the subspace
\[S_Q = \mr{Im}(\Gamma(Q,-)) \sub V.\]

\begin{prop} \label{purespinorimage}
Let $Q$ be a spinor. Then $S_Q = T_Q^{\perp}$. In particular, if $Q$ is a pure spinor,
\begin{itemize}
 \item If $n$ is even, then $S_Q = T_Q$.
 \item If $n$ is odd, then $T_Q \sub S_Q = T_Q^\perp$ is a codimension 1 subspace.
\end{itemize}
\end{prop}

\begin{proof}
Let $v\in S_Q^\perp$, i.e. $(v, w) = 0$ for all $w \in S_Q$. That is for every $Q'\in S$ we have
\[
0 = (v, \Gamma(Q, Q')) = \langle\rho(v) Q, Q'\rangle.
\]

Since this is true for every $Q'$, we get that $\rho(v) Q = 0$, i.e. $S_Q^\perp=T_Q$. Taking the orthogonal complement again, we get $S_Q = T_Q^\perp$.

If $Q$ is a pure spinor, $T_Q$ is maximal isotropic, so in even dimensions $S_Q = T_Q$ and in odd dimensions $T_Q\sub S_Q$ has codimension 1.
\end{proof}

Let us return to prove the statement we claimed in Section \ref{twist_section}: that all supercharges are at least holomorphic.

\begin{prop} \label{at_least_holo_prop}
Let $Q \in \Sigma$ be an element of an \emph{arbitrary} spinorial representation in dimension $n$ equipped with a non-degenerate pairing.  Then the image space $S_Q$ has dimension at least $n/2$.
\end{prop}

\begin{proof}
First, if $Q\in S$, then by Proposition \ref{purespinorimage} $S_Q$ has dimension $\ge n/2$, since $T_Q$ has dimension $\le n/2$.  If $\Sigma = S \otimes W$, then if we choose a basis $\{w_1, \ldots, w_k\}$ for $W$, we can write $Q = Q_1 \otimes w_1 + \cdots + Q_k \otimes w_k$, and then $S_Q = \mr{span}(S_{Q_1},\ldots, S_{Q_k})$ which then also has dimension $\ge n/2$.  

More generally, suppose $\Sigma = S_+ \otimes W_+ \oplus S_- \oplus W_-$ where there is a strict inclusion $W_- \subsetneq W_+$. Then $S_+$ and $S_-$ are paired with themselves. For an element $Q_1 \otimes w_1 \in S_+ \otimes W_+$ we can identify the image of $\Gamma(Q_1 \otimes w_1, -)$ in the spaces $S_+ \otimes W_+ \oplus S_- \oplus W_-$ and in $(S_+ \oplus S_-) \otimes W_+$.  So we still have that $S_{Q_1} \sub S_Q$, so $S_Q$ has dimension $\ge n/2$.
\end{proof}

Now, what is the relationship between pure spinors and those satisfying $\Gamma(Q,Q)=0$ (i.e. that square to zero in the supersymmetry algebra)?  To answer this question, we will use an interesting decomposition.  Every spinor $Q$ induces a matrix element
\[Q \otimes Q \in S \otimes S\cong S\otimes S^*  \iso \eend(S)\] using the invariant pairing on $S$.  Since $S$ is irreducible, we can identify endomorphisms with elements of the Clifford algebra, and thus decompose the spinor bilinear as
\[Q \otimes Q = \sum_{p=0}^n F_p \in \eend(S)\cong \mr{Cl}(V)\cong \wedge^\bullet(V),\]
where $F_p \in \wedge^p V$.  We can explicitly identify these elements $F_p$ as, in index notation, the Clifford algebra element $Q_i \gamma^{a_1 \cdots a_p}_{ij} Q_j$.  There is a nice characterisation of pure spinors in this language.

\begin{theorem}[Chevalley] \label{ChevalleyCriterion}
A non-zero spinor $Q$ is pure if and only if it is Weyl (in even dimensions), and $F_p = 0$ for all $p < \lfloor n/2 \rfloor$. 
\end{theorem}

For instance, if $n\geq 4$ and $Q$ is a pure spinor, we have $F_1 = \Gamma(Q, Q) = 0$. In other words, any pure spinor in these dimensions is square-zero.

We have the following observation (see \cite{BudinichTrautman}).

\begin{lemma}\label{WeylCriterion}
If $Q$ is a Weyl spinor in dimension $n=2m$, the tensor $F_p$ is zero unless $p\equiv m\pmod 4$.
\end{lemma}

Finally, we can characterize pure spinors in odd dimensions in terms of those in even dimensions. Indeed, suppose $n$ is odd and $S$ is the spinor module. If we denote by $S_+$ the semi-spin representation of $\Spin(V\oplus \C)$, then we can identify $S\cong S_+$ as $\Spin(V)$-representations.

 \begin{prop} [{\cite[Chapter III.8]{Chevalley}}]\label{prop:oddimensionpurespinor}
 A spinor $Q\in S$ in an odd dimension $n$ is pure if and only if it is pure as a Weyl spinor $Q\in S_+$ in dimension $n+1$.
 \end{prop}

Now, let us use these facts to investigate pure spinors in various dimensions.
\begin{itemize}
 \item If $n=2, 4, 6$ every non-zero Weyl spinor is pure.
 \item If $n=1, 3, 5$ every non-zero spinor is pure.
 \item If $n=8$ we have a single constraint for a Weyl spinor to be pure, that $F_0 = 0$. That is, the spinor satisfies $\langle Q,Q \rangle = 0$.  All Weyl spinors satisfy $\Gamma(Q,Q) = 0$ however, so in this dimension squaring to zero is strictly weaker than being pure.  As we will see in Section \ref{dim8_section} all impure square-zero spinors in dimension 8 are in fact topological.
 \item If $n=7$ we can use Proposition \ref{prop:oddimensionpurespinor} to conclude that a spinor $Q$ is pure if and only if $\langle Q, Q\rangle_8=0$, where $\langle -, -\rangle_8$ is the scalar-valued spinor pairing in 8 dimensions.
 \item If $n=10$ the constraint for a Weyl spinor to be pure is that $F_1 = 0$. That is, a Weyl spinor $Q$ in dimension 10 is pure if and only if $\Gamma(Q,Q) = 0$.
 \item If $n=9$ we can again use Proposition \ref{prop:oddimensionpurespinor} to conclude that a spinor $Q$ is pure if and only if $\Gamma_{10}(Q, Q) = 0$.
 \end{itemize}
 
To summarise, there are non-trivial pure spinor constraints only in dimensions at least 7.  In the classification of twists we will refer to purity only in those large dimensions where it is a non-trivial condition.

\subsection{The Stress-Energy Tensor} \label{stressenergy_section}

Often topologically twisted theories carry a natural action of all vector fields extending the action of the Poincar\'{e} algebra. In this section we explain how this action can be constructed. The results in this section will not be used in the rest of the paper, but instead help to motivate the concept of the topological twist by means of the following structure: topologically twisted field theories are often automatically \emph{generally covariant}, meaning they carry an action of the group of diffeomorphisms of spacetime.

\begin{definition}
Let $g_0$ be the standard metric on $\RR^n$ of signature $(p,q)$ with constant coefficients.  The local Lie algebra of \emph{Killing vector fields} is the sheaf of dg Lie algebras given by
\[L_{\mr{Kill}} = \big(T_{\RR^n} \to \sym^2 T_{\RR^n}\big)\]
placed in degrees 0 and 1, where the differential sends a vector field $X$ to the Lie derivative $\LL_X(g_0)$, and where the Lie bracket is likewise given by the Lie derivative.
\end{definition}

Note that the Poincar\'{e} algebra $\mf{iso}(n)$ acts by Killing vector fields on $\RR^n$, so we have a natural morphism $\mf{iso}(n) \to L_{\mr{Kill}}(U)$ for each open set $U\sub \RR^n$.

\begin{definition}
A prefactorization algebra on $\RR^n$ is \emph{isometry-invariant} if carries a smooth $\ISO(p, q)$-action and the infinitesimal Poincar\'e algebra action extends to an inner action $\nu\colon L_{\mr{Kill}}\rightarrow \obs(\RR^n)[-1]$ of $L_{\mr{Kill}}$.  The \emph{gravitational stress-energy tensor} $T_{\mr{grav}}$ of an isometry-invariant factorization algebra $\obs$ is the composite
\[T_{\mr{grav}}\colon \sym^2 T_{\RR_n}\longrightarrow L_{\mr{Kill}}[1]\longrightarrow \obs(\RR^n).\]
\end{definition}

\begin{remark}
The gravitational stress-energy tensor of a field theory models infinitesimal deformations of the background metric $g_0$. Note that this is not the same as the ``canonical'' stress-energy tensor modelling the conserved currents of the global translation action.  These two tensors are related by a correction term involving the conserved currents for rotations: this is known as the \emph{Belinfante-Rosenfeld formula} (see for instance \cite{FernandezGarciaRodrigo} for a rigorous account of this relationship).
\end{remark}

\begin{remark}
Given a quantum field theory arising from a classical Lagrangian description where the fields are tensors, i.e. sections of $T_{\RR^n}^{\otimes a}\otimes  (T^*_{\RR^n})^{\otimes b}$, we have a natural action of the Lie algebra of vector fields on $\RR^n$ on the space of fields. This gives a isometry invariant structure on the factorization algebra of observables $\obs$ where the gravitational stress-energy tensor is
\[T_{\mr{grav}}(\delta g) = \int \left.\frac{\delta S}{\delta (g^{\mu\nu}(x))}\right|_{g=g_0} \delta g^{\mu\nu}(x),\]
the variation of the action with respect to the metric.

Typical supersymmetric field theories involve fields which are sections of spinor bundles, so they are not locally isometry invariant in this sense. However, consider a twisting homomorphism $\phi\colon \SO(p, q)\rightarrow G_R$ so that under the induced $\SO(p, q)$-action the fields have integer spins (i.e. are tensors). Then we obtain that the prefactorization algebra of observables $\obs$ is locally isometry invariant with respect to the twisted $\SO(p, q)$ action. If we moreover assume that the supercharge $Q$ is compatible with the twisting homomorphism $\phi$, the twisted prefactorization algebra $\obs^Q$ is also locally isometry invariant.
\end{remark}

Let us fix a supersymmetric factorization algebra $\obs$ on $\RR^n$ with a twisting homomorphism $\phi\colon \SO(p, q)\rightarrow G_R$ and a compatible supercharge $Q$. Moreover, suppose that the twisted factorization algebra $\obs^Q$ is locally isometry invariant.

\begin{lemma} \label{exact_SE_tensor_lemma}
Suppose $\obs$ is equipped with a map
\[\Lambda\colon \Gamma(\RR^n, \sym^2 T_{\RR^n})\rightarrow \obs(\RR^n)[-1]\]
satisfying the following equations:
\begin{enumerate}
\item $\d_Q \Lambda(\delta g) = T_{\mr{grav}}(\delta g)$.

\item $\Lambda(\LL_{[v, w]} g_0) = [\nu(v), \Lambda(\LL_w(g_0))] + [\Lambda(\LL_v(g_0)), \nu(w)] - [\Lambda(\LL_v g_0), \Lambda(\LL_w g_0)]$.
\end{enumerate}

Then the Poincar\'e action on $\obs^Q$ extends to an infinitesimal action of the tangent sheaf $T_{\RR^n}$ as a sheaf of Lie algebras.
\end{lemma}
\begin{proof}
Define a map
\[\nu_Q\colon \Gamma(\RR^n, T_{\RR^n})\rightarrow \obs^Q(\RR^n)[-1]\]
by
\[\nu_Q(v) = \nu(v) - \Lambda(\LL_v(g_0)).\]
The fact that it is a morphism of dg Lie algebras follows from the equations on $\Lambda$.
\end{proof}

\begin{remark}
The Lie algebra of vector fields on $\RR^n$ is a Lie algebra of the group $\mr{Diff}(\RR^n)$ of diffeomorphisms, so we have shown that topologically twisted theories with null-homotopic stress-energy tensor are \emph{generally covariant} (at least infinitesimally -- the infinitesimal $T_{\RR^n}$-action need not necessarily exponentiate).  Witten discusses topological twisting from this perspective in his original paper on topological quantum field theories \cite[Section 6]{WittenTQFT}.
\end{remark}

\begin{remark}
In physical examples where one has a Lagrangian description for a classical field theory with an explicit metric dependence one can check that the action functional consists of a $Q$-exact term plus a metric independent term.  This implies that the variation of the action with respect to the metric, the gravitational stress-energy tensor, is $Q$-exact.
\end{remark}

\subsection{Dilation Actions} \label{dilation_section}
When can we apply the result of Section \ref{E_n_section} to topologically twisted field theories?  Since topological twists are automatically de Rham translation-invariant, according to Corollary \ref{E_n_summary_cor} the factorization algebra $\obs^Q$ of twisted observables admits a canonical $\bb E_n$-algebra structure whenever the factorization map $\mr{fact}_{r,R} \colon \obs^Q(B_r(0)) \to \obs^Q(B_R(0))$ associated to concentric balls is a quasi-isomorphism.  When we have a direct description of the local twisted observables, we can check this condition by hand.  Alternatively, we can abstractly describe quasi-inverses to the factorization maps on concentric balls by means of a \emph{dilation action}.

\begin{definition}
A prefactorization algebra $\obs$ on $\RR^n$ is \emph{dilation-invariant} if it admits a smooth action of the group $\RR_{>0}$ acting on $\RR^n$ by dilations.
\end{definition}

Such a smooth dilation action defines a rescaling map $\alpha_{R,r} \colon \obs(B_R(0)) \to \obs(B_r(0))$ for any pair of radii $r, R$. This rescaling map makes a de Rham translation-invariant factorization algebra into an $\bb E_n$-algebra as long as the dilation action is \emph{itself} homotopically trivial.

\begin{definition} \label{dR_dilation_def}
A dilation-invariant prefactorization algebra $\obs$ on $\RR^n$ is \emph{de Rham dilation-invariant} if the smooth $\RR_{>0}$-action extends to an action of $(\RR_{>0})_{\dR}$.
\end{definition}

\begin{prop} \label{dR_dilation_prop}
Let $\obs$ be a de Rham dilation-invariant prefactorization algebra on $\RR^n$. The rescaling map $\alpha_{R,r}$ provides a quasi-inverse to the factorization map $\mr{fact}_{r,R}$.
\end{prop}

\begin{proof}
Fix $r>0$ and consider the family of endomorphisms $m_s = \mr{fact}_{sr,r} \circ \alpha_{r,sr}$ of the local observables $\obs(B_r(0))$ parameterised by rescalings $0<s\le1$. By the definition of a smooth action, $m_s$ is a smooth family of endomorphisms of $\obs(B_r(0))$ which is the identity at $s=1$. Thus, to show that it is homotopic to the identity, it will be enough to show that $\frac{\d}{\d s} m_s$ is null-homotopic for any $s$.

If we denote by
\[\epsilon = \d \eta + \eta \d\]
the derivation generating the dilation action, then we have
\begin{align*}
\frac{\d}{\d s} m_s(\OO) &= m_s(\epsilon(\OO)) \\
&= \d m_s(\eta(\OO)) + m_s(\eta(\d\OO)).
\end{align*}

Thus, $m_s\circ\eta$ provides a nullhomotopy for $\frac{\d}{\d s} m_s$.
\end{proof}

\begin{remark}
Suppose $\obs$ is a prefactorization algebra on $\RR^n$ valued in $\BD_0$-algebras. Then the factorization maps $\mr{fact}_{r, R}$ are quasi-isomorphisms if and only if they are quasi-isomorphisms at $\hbar = 0$. This can be proved by a spectral sequence argument similar to \cite[Proposition 9.6.1.1]{Book2}. There are de Rham dilation-invariant field theories on the classical level where the dilation action is anomalous on the quantum level; this argument tells us that $\mr{fact}_{r, R}$ is nevertheless still a quasi-isomorphism on the quantum level.
\label{remark:classicaldilationaction}
\end{remark}

Combining Proposition \ref{dR_dilation_prop} and Corollary \ref{E_n_summary_cor} we get the following.

\begin{corollary}
Let $\obs$ be a de Rham dilation-invariant prefactorization algebra on $\RR^n$. Then $\obs(B_1(0))$ can be canonically equipped with the structure of a $\disk_n$-algebra.
\end{corollary}

\subsection{Superconformal Theories} \label{superconformal_section}

We can build twisted theories which are de Rham dilation invariant as well as de Rham translation-invariant starting from not just a supersymmetry action, but a \emph{superconformal} action.  We refer to Kac's survey article \cite{KacSurvey} for a comprehensive account of the theory of Lie superalgebras and superconformal algebras.

\begin{definition}
Fix an inner product on $\RR^n$ of signature $(p,q)$. The \emph{conformal group} $\mr{CO}(p,q)$ of signature $(p,q)$ is the group of oriented conformal transformations of $\RR^n$.
\end{definition}

\begin{remark}
We always have an inclusion $\SO(p+1, q+1)\sub \mr{CO}(p, q)$. If $(p, q) = (1, 1), (1, 0)$ or $(0, 1)$, the conformal group is infinite-dimensional. In the other cases this inclusion is an isomorphism.
\end{remark}

Let us denote by $\mf{co}(p, q)$ the complexified Lie algebra of $\mr{CO}(p, q)$.

\begin{definition}
A \emph{superconformal algebra} is a simple super Lie algebra whose bosonic part has form $\mf{co}(p,q) \oplus \gg_R$ for some Lie algebra $\gg_R$, and where the action of $\mf{co}(p,q)$ (if $n\leq 2$, the action of its subalgebra isomorphic to $\so(p+1,q+1)$) on the fermionic part defines a spinorial representation.
\end{definition}

The possible superconformal algebras in dimensions $\ge 3$ were classified by Nahm \cite{Nahm} and Shnider \cite{Shnider}, who showed that they only exist in dimensions up to 6 \footnote{Shnider's definition of a superconformal algebra is slightly different to the one presented here: he only assumes that the bosonic part contains $\mf{co}(p,q)$, but he then requires the existence of an infinitesimal action on super Minkowski space.  However \cite[Lemma 2]{Shnider} ensures that the two definitions coincide.}. 

\begin{definition}
A factorization algebra on $\RR^n$ is \emph{superconformal} if it admits a smooth action of $\mr{CO}(p, q)\times G_R$, so that the action of $\mf{co}(p, q)\oplus \gg_R$ extends to an action of a superconformal algebra $\mf{A}$.
\end{definition}

We can explicitly describe the possible superconformal algebras starting from dimension 6 and working down using Kac's classification of finite-dimensional semisimple super Lie algebras \cite{Kac}. That is, identifying those super Lie algebras where the bosonic part is equivalent to $\so(n+2) \oplus \gg_R$ and where the action on the fermionic part is spinorial.

Let $\mf{A}$ be a superconformal algebra in dimension $n$. Then its even part is
\[\mf{A}_{\mr{even}} = \so(n)\oplus V\oplus V\oplus \CC\cdot D\oplus \gg_R,\]
where the two copies of $V$ correspond to translations and special conformal transformations respectively and where $D$ is the generator of dilations. Its odd part is
\[\mf{A}_{\mr{odd}} = \Sigma\oplus \Sigma^*\]
where $\Sigma$ is a spinorial representation of $\so(n)$. The superconformal algebra $\mf{A}$ contains a subalgebra spanned by $\Sigma$ and translations; it also contains an isomorphic subalgebra spanned by $\Sigma^*$ and special conformal transformations. We also have a bracket
\[\Sigma\otimes \Sigma^*\longrightarrow \so(n)\oplus \CC\cdot D\oplus \gg_R\]
whose $D$ component corresponds to the natural pairing.  This bracket has the property that if $\Sigma = S \otimes W$ the projection $\Sigma \otimes \Sigma^* \to \so(n) \oplus \CC \cdot D$ factors through the natural evaluation pairing on $W$:
\[\Sigma\otimes \Sigma^* \to S \otimes S^* \to \so(n)\oplus \CC\cdot D.\]
Likewise if $\Sigma = S_+ \otimes W_+ \oplus S_- \otimes W_-$ it factors through 
\[\Sigma\otimes \Sigma^* \to S_+ \otimes S_+^* \oplus S_- \otimes S_-^* \to \so(n)\oplus \CC\cdot D.\]
We refer to \cite{Minwalla} for explicit formulas for brackets.

\begin{itemize}
 \item In dimension 6 there are superconformal algebras with $\mc N=(k,0)$ supersymmetry for all $k$ isomorphic to $\osp(8|k)$.
 \item In dimension 5 there is a unique superconformal algebra: the exceptional super Lie algebra $\mf f(4)$ with bosonic part $\so(7) \oplus \sl(2)$. This superconformal algebra contains an $\mc N=1$ supersymmetry algebra.
 \item In dimension 4 there is a superconformal algebra for all $\mc N$ given by $\sl(4|\mc N)$ unless $\mc N=4$ in which case we should quotient by the center to get $\mf{psl}(4|4)$.
 \item In dimension 3 there is again a superconformal algebra for all $\mc N$ isomorphic to $\osp(\mc N|4)$. 
 \item In dimension 2 the superconformal algebras are infinite-dimensional.  For a classification of the possible superconformal algebras see Fattori-Kac \cite{FattoriKac}.  The most notable 2-dimensional superconformal algebras are the supersymmetric extensions of the Virasoro algebra such as the Ramond, Neveu--Schwarz and $\mc N=2$ superconformal algebras. 
 \item In dimension 1 there are many possible superconformal algebras, where we restrict attention to supersymmetric extensions of the small conformal algebra $\so(3) \iso \sl(2)$. We refer to \cite[Section 3.1.3]{Okazaki} for more details.
\end{itemize}

Since superconformal factorization algebras are in particular supersymmetric it makes sense to form their twist by a square-zero supercharge $Q$.  These twisted theories admit a smooth action of the subgroup of the conformal group stabilizing $Q$ just like we discussed in Section \ref{twisting_hom_section}.  As such, we can make a $Q$-twisted factorization algebra $\obs^Q$ dilation invariant whenever we have a twisting homomorphism for the dilation action of the form
\[\psi \colon \RR_{>0} \to G_R\]
compatible with $Q$. Note that since $\RR_{>0}$ is contractible, it is enough to specify the twisting homomorphism on the level of Lie algebras. In fact, this is always possible.

\begin{lemma} \label{dilation_potential_lemma}
If $Q$ is a square-zero supercharge in a superconformal algebra $\mf A$ as follows:
\begin{itemize}
 \item The $\mc N=(2,2)$ superconformal algebra in dimension 2 with $Q$ a topological supercharge.
 \item The $\mc N=k$ superconformal algebra in dimension 3 with $Q$ of rank 2.
 \item The $\mc N=k$ superconformal algebra in dimension 4 with $Q$ of rank at least 2.
 \item The $\mc N=1$ superconformal algebra in dimension 5 with $Q$ of rank 2.
 \item The $\mc N=(k,0)$ superconformal algebra in dimension 6 with $Q$ of rank 4.
\end{itemize}
Then there exists a twisting homomorphism $\psi$ for the dilation group so that $Q$ is invariant for the twisted dilation action.  Furthermore, this twisted dilation action is de Rham.
\end{lemma}

\begin{proof}
We can characterize the data of a dilation twisting homomorphism as an element $J$ of the R-symmetry algebra $\mf g_R$ so that $[J, Q] = [D, Q]$ in the superconformal algebra, where $D$ is the infinitesimal dilation. This twisted dilation action actually makes the twisted theory de Rham dilation invariant if $J - D$ is actually $Q$-exact, i.e. if there exists an odd element $Z$ of the superconformal algebra so that $[Z, Q] = J - D$.  That is, it suffices to find $Z$ so that the $\so(n)$ component of $[Z,Q]$ vanishes but the $D$ component of $[Z,Q]$ does not vanish.  We can check that such elements $Z$ and therefore $J$ exist by considering the classification of superconformal algebras.  That is, we'll consider $Q \otimes Z$ as an element of $\Sigma \otimes \Sigma^*$, and choose $Z$ so that $Q \otimes Z$ vanishes when we project onto the irreducible summands of the $\so(n)$-module $\Sigma \otimes \Sigma^*$ isomorphic to the adjoint representation, but does not vanish when we project those summands corresponding to the trivial representation.

Decompose $Q = \sum s_i \otimes w_i$ as a sum of pure tensors in $\Sigma = S\otimes W$ or $S_+ \otimes W_+ \oplus S_- \otimes W_-$.  Recall that the odd part of the superconformal algebra splits as $\Sigma \oplus \Sigma^*$, and recall that the $\so(n) \oplus \CC \cdot D$ part of the bracket between $\Sigma$ and $\Sigma^*$ factors through the evaluation pairing on $W$, or the separate evaluation pairings on $W_+$ and $W_-$.  We will construct $Z \in \Sigma^*$ in each dimension in turn.
\begin{itemize}
 \item In dimension 2 the $\mc{N}=2$ Neveu--Schwarz algebra is spanned by the Virasoro generators $\{L_n\}_{n\in\ZZ}$, odd elements $\{G^+_i,G^-_i\}_{i\in \frac 12+\ZZ}$, $R$-symmetry generators $\{J_n\}_{n\in\ZZ}$ and the central element $\textbf{1}$ with brackets of the form
\[[G^+_i, G^-_j] = L_{i+j} + \frac{1}{2} (i-j) J_{i+j} + \delta_{i,-j} \left(\frac{i^2}6 - \frac 1{24}\right) c \textbf{1},\]
where $c$ is the central charge.

We consider the $\mc{N}=(2, 2)$ superconformal algebra which is given by the sum of two $\mc{N}=2$ Neveu--Schwarz algebras with generators we denote by $\{G^+_i,G^-_i,L_n,J_n, \textbf{1}\}$ and $\{\overline{G}^+_i, \overline{G}^-_i, \overline{L}_n, \overline{J}_n, \overline{\textbf{1}}\}$.

In particular, for the natural $\mc N=(2,2)$ topological supercharge $Q = G^-_{-1/2} + \overline{G}^-_{-1/2}$ we have
\[[Q, G^+_{1/2} + \overline{G}^+_{1/2}] = L_0 + \overline{L}_0 + \frac{1}{2} (J_0 + \overline{J}_0),\]
where $L_0 + \overline{L}_0$ is the generator of dilations. In particular, there is a twisting homomorphism making $Q$ dilation-invariant.
 \item In dimension 3 $S \iso S^*$ and the tensor product of representations decomposes as $S \otimes S^* \iso \sym^2(S) \oplus \wedge^2(S) \iso \so(3) \oplus \CC$.  Since the bracket between $\Sigma$ and $\Sigma^*$ is $\so(3)$-equivariant, the $\so(3) \oplus \CC \cdot D$ part of the bracket in the superconformal algebra is uniquely determined up to rescaling.  Choose $f_1, f_2 \in W^*$ so that $f_i(w_j) = \delta_{ij}$ for $j=1,\cdots,4$.  Then if $Z = s_2 \otimes f_1 - s_1 \otimes f_2 \in S \otimes W^*$ then, since $Q$ had rank at least 2, the image of $Q \otimes Z$ in $\sym^2(S)$ vanishes, and so $[Q,Z] = D + J$ for some $J$ as required.
 \item In dimension 4 $S_\pm \iso S^*_\pm$ and the tensor product of representations decomposes as $S_+ \otimes S_+^* \oplus S_- \otimes S_-^* \iso \sym^2(S_+) \oplus \sym^2(S_-) \oplus \wedge^2(S_+) \oplus \wedge^2(S_-) \iso \so(4) \oplus \CC \oplus \CC$.  Therefore as in dimension 3 $\so(4)$-equivariant of the bracket, along with the $\CC \cdot D$ term being symmetrical in $+$ and $-$, means the $\so(4) \oplus \CC \cdot D$ part of the bracket in the superconformal algebra is uniquely determined up to rescaling.  Choose $f_1, f_2 \in W_+^* \oplus W_-^*$ so that $f_i(w_j) = \delta_{ij}$ for $j=1,\cdots,4$.  Then if $Z = s_2 \otimes f_1 - s_1 \otimes f_2 \in \Sigma^*$ then, since $Q$ had rank at least 2, the image of $Q \otimes Z$ in $\sym^2(S_+) \oplus \sym^2(S_-)$ vanishes, and so $[Q,Z] = D + J$ for some $J$.
 \item In dimension 5 $S \iso S^*$ and the tensor product of representations decomposes as $S \otimes S^* \iso \sym^2(S) \oplus \wedge^2(S) \iso \so(5) \oplus (V \oplus \CC)$.  Therefore as above the $\so(5) \oplus \CC \cdot D$ part of the bracket in the superconformal algebra is uniquely determined up to rescaling.  Here $W$ has a symplectic structure which induces an isomorphism $W \otimes W^*$, so let $Z = Q = s_1 \otimes w_1 + s_2 \otimes w_2$.  Then, since $Q$ had rank at least 2, $[Z,Q] = \omega(w_1,w_2)(s_1 \otimes s_2 - s_2\otimes s_1)$, which vanishes when projected to $\sym^2(S)$.  It also vanishes when projected to $V$ because $Q$ squares to zero in the supersymmetry algebra.  However since $s_1 \ne s_2$ it does not vanish when projected to $\CC \cdot D$, so $[Q,Z] = \lambda D + J$ for some $J$ and some non zero $\lambda \in \CC$.
 \item In dimension 6 $S_+ \otimes S_+^*$ decomposes into irreducible summands as $\so(6) \oplus \CC$, where we view $S$ as the defining representation of $\sl(4)$, so $S_+ \otimes S_+^* \iso \gl(4)$ and the projection onto $\CC$ is given by the trace.  Again this uniquely determined the $\so(6) \oplus \CC \cdot D$ part of the bracket in the superconformal algebra is uniquely determined up to rescaling.  Choose $Z \in \Sigma^*$ of rank 4 so that the image of $Q \otimes Z$ under the evaluation pairing $\Sigma \otimes \Sigma \to S_+ \otimes S_+^*$ is represented by the matrix $\lambda \cdot \mr{id}_4$ for $\lambda \ne 0$.  This vanishes under the projection to $\so(6) \iso \sl(4)$, but not under the projection to $\CC \cdot D$.  Therefore $[Q,Z] = \lambda D + J$ for some $J$ and some non zero $\lambda \in \CC$.
\end{itemize}
\end{proof}

\begin{remark}
In dimension 3 and larger this result is best possible: if $Q$ has lower rank then there is never a twisting homomorphism for the dilation group coming from the superconformal algebra.  However as we will see explicitly in Section \ref{twisted_theories_section} all topological supercharges have rank at least 2, and in dimension 6 $\mc N=(k,0)$ all topological supercharges have rank at least 4.  The Lemma therefore includes all topological twists of superconformal theories.
\end{remark}

Altogether, this analysis along with our main Theorem \ref{E_n_theorem} tells us the following.
\begin{theorem} \label{superconformal_theorem}
Let $\obs$ be a superconformal prefactorization algebra where $\mf A$ is one of the superconformal algebras occuring in Lemma \ref{dilation_potential_lemma}.  Let $Q$ be a topological supercharge acting on $\obs$. Then the complex of twisted observables $\obs^Q(B_1(0))$ has the structure of a $\disk_n$-algebra.
\end{theorem}
\begin{proof}
By our analysis above, since $\obs$ is superconformal there is a twisted dilation action on $\obs$ that descends to make $\obs^Q$ de Rham dilation-invariant.  According to Proposition \ref{dR_dilation_prop} that means the factorization map associated to concentric balls is a quasi-isomorphism.

By the classification of topological supercharges in these dimensions performed in Section \ref{twisted_theories_section} and Proposition \ref{commuting_potential_prop}, we may find an odd abelian subalgebra $\mf{a}\sub \mf{A}$ such that $[Q, -]\colon \mf{a}\rightarrow \CC^n$ is an isomorphism. Therefore, by Proposition \ref{topologically_twisted_theories_are_dR_translation_invt_prop} $\obs^Q$ is de Rham translation-invariant.

Combining these two statements, by Corollary \ref{E_n_summary_cor} we get that $\obs^Q(B_1(0))$ has the structure of a $\disk_n$-algebra.
\end{proof}

This is not the only thing we will use superconformal symmetry for.  In several 2, 3 and 4-dimensional examples in Section \ref{twisted_theories_section} we will verify that using the action of the superconformal algebra it is possible to find potentials for the infinitesimal $\so(n)$-action that endow the twisted factorization algebra with an $\SO(n)_{\dR}$-structure. Let us write out explicitly the conditions for an $(\SO(n)\ltimes \RR^n)_{\dR}$-structure.

\begin{definition}
Let $\mf{A}$ be a superconformal algebra and $Q\in\mf{A}$ a topological supercharge compatible with the twisting homomorphism $\phi\colon \so(n)\rightarrow \gg_R$. A \emph{potential for affine transformations} is a linear map $\psi\colon \so(n)\oplus V\rightarrow \mf{A}_{\mr{odd}}$ with an abelian image which satisfies the following equations:
\begin{align*}
[Q, \psi(v)] &= v,& &\forall v\in V \\
[Q, \psi(r)] &= r + \phi(r),& &\forall r\in\so(n) \\
[\psi(r), v] &= \psi([r, v]),& &\forall r\in \so(n), v\in V \\
[\psi(v), r + \phi(r)] &= -\psi([r, v]),& &\forall r\in \so(n), v\in V \\
[\psi(v), w] &= 0,& &\forall v,w \in V \\
[\psi(r_1), r_2 + \phi(r_2)] &= \psi([r_1, r_2]),& &\forall r_1,r_2\in \so(n).
\end{align*}
\label{def:affinepotential}
\end{definition}

Combining Proposition \ref{dR_dilation_prop} and Corollary \ref{Efr_n_summary_cor} we get the following statement.

\begin{prop}
Suppose $\obs$ is a superconformal prefactorization algebra where $\mf A$ is one of the superconformal algebras occuring in Lemma \ref{dilation_potential_lemma}. Let $Q$ be a topological supercharge compatible with a twisting homomorphism $\phi\colon \so(n)\rightarrow \gg_R$ and choose a potential for affine transformations. Then $\obs^Q(B_1(0))$ can be equipped with a structure of a framed $\bb E_n$-algebra.
\end{prop}

\section{A Catalogue of Twisted Theories in Various Dimensions} \label{twisted_theories_section}

\setcounter{subsection}{-1}
\subsection{Conventions}
In this section we will describe the possible twists for each supersymmetry algebra in dimensions 1 to 10, and give references to discussions in the literature of the corresponding twisted field theories.  In general we will restrict attention to theories where the fields have spins at most one; in terms of supersymmetry algebras this means we only consider algebras with at most 16 supercharges (this is a condition on the representation theory of super Poincar\'e algebras, see e.g. the original classification of Nahm \cite{Nahm}).  In low dimensions however (below dimension five) our analysis will include all supersymmetry algebras without needing to impose this restriction.  

In each dimension we will first describe the Lorentz group and supertranslation Lie algebra.  By the ``Lorentz'' group we will actually mean the spin group (which we will write in Euclidean signature). We will also mention what happens to the supertranslation algebra when we compactify it to one dimension lower.

We will then classify square-zero supercharges, giving a complete classification of orbits of square-zero supercharges under the action of the complexified Lorentz group, R-symmetry and scaling.  We also explain ``how topological'' they are, i.e. the dimension of their image in the Lie algebra $V$ of translations.  Recall that if the image consists of all translations, we say the supercharge is \emph{topological}. If it is half-dimensional, we say it is \emph{holomorphic}. In the intermediate cases we say $Q$ has \emph{$n$ invariant directions} if the dimension of the image of $[Q, -]$ in $V$ is $n$.  We will often refer to the following invariant.

\begin{definition} \label{rank_definition}
Given a supercharge $Q\in S\otimes W$ in odd dimensions, it defines a morphism $S^*\rightarrow W$.  The \emph{rank} $\rk(Q)$ of $Q$ is the dimension of the image of this map.  We will denote this image by $W_Q$. Similarly, in even dimensions $Q_+ + Q_-\in S_+\otimes W_+\oplus S_-\otimes W_-$ which defines two maps $S_+^*\rightarrow W_+$ and $S_-^*\rightarrow W_-$.  The \emph{rank} $\rk(Q)$ of $Q$ is a pair of integers given by the dimensions of the images $W_{Q_+}, W_{Q_-}$ of these two maps.
\end{definition}

We continue by describing the twisting homomorphisms from the Lorentz group to the group of R-symmetries, and the space of topological supercharges compatible with each twisting homomorphism (i.e. twisted Lorentz-invariant). In some cases we will only describe a twisting homomorphism from a subgroup $G$ of the Lorentz group in which case we get a $G$-structured theory. In many cases we can additionally find a $\mr U(1)$ subgroup in the R-symmetry group making a correspondingly twisted theory $\ZZ$-graded rather than just $\ZZ/2$-graded (see Remark \ref{rmk:twistedgrading}).  Recall that such a $\mr U(1)$ subgroup chosen so that $Q$ has weight one is called a choice of \emph{twisting datum}.

\begin{remark}
In some cases we can find both a twisting homomorphism $\phi \colon G \to G_R$ from a subgroup $G$ of the Lorentz group \emph{and} a $\mr U(1)$ subgroup $\alpha \colon \mr U(1) \to G_R$ in the R-symmetry group making the twisted theory $\ZZ$-graded.  However these two structures are not always compatible: sometimes $\alpha$ cannot be chosen to commute with $\phi$.  In this case the $\ZZ/2$-graded theory is $G$-structured but the $\ZZ$-graded theory is only framed.  Concretely, the twisting homomorphism $\phi$ induces a decomposition of the auxiliary space $W$ into irreducible $G$-representations, and one must check whether there exists a $\mr U(1)$-subgroup of $G_R$ that acts by a scalar on each summand.  
\end{remark}

\subsection{Dimension 1}
The Lorentz group is $\Spin(1) \iso \ZZ/2$.  The Dirac spinor representation is the sign representation $S = \CC$ with trivial pairing $S \otimes S \to \CC$.

Pick a vector space $W$ equipped with a symmetric non-degenerate pairing. The supertranslation algebra is
\[\CC \oplus \Pi(S \otimes W)\]
with $\dim W = \mc N$.  The R-symmetry group is $\mr O(\mc N)$.

A supercharge $Q \in W \otimes S$ squares to zero if the image $W_Q$ of $S^*\rightarrow W$ -- as in Definition \ref{rank_definition} -- is isotropic. These square-zero supercharges are always topological. The group of R-symmetries $\O(W)$ acts transitively on the space of isotropic lines $W_Q\sub W$, so the orbits of square-zero supercharges are classified by the rank.

This supercharge can be promoted to twisting data, i.e. we can find a copy of $\mr U(1) \sub \mr O(\mc N)$ acting on $Q$ with weight one.

\begin{example}[Topological quantum mechanics]
The theory of $\mc N=2$ supersymmetric quantum mechanics with target space $X$, a Riemannian spin manifold, admits a topological twist equivalent to topological quantum mechanics on $X$.  This theory has a $Q$-exact stress-energy tensor, and has partition function computing the $\hat A$-genus of $X$ \cite{WittenMorse,AlvarezGaume, GradyGwilliam}.
\end{example}

\subsection{Dimension 2}
The Lorentz group is $\Spin(2) \iso \mr{U}(1)$. The semi-spin representations are $S_+=\CC$ and $S_-=\CC$ with weights $\frac{1}{2}$ and $-\frac{1}{2}$ respectively. The vector representation is $V=\CC_1\oplus \CC_{-1}$. We have identifications
\[\sym^2(S_+)\cong S_+^{\otimes 2}\cong \CC_1\]
and
\[\sym^2(S_-)\cong S_-^{\otimes 2}\cong \CC_{-1}.\]

We denote the bases of $S_+$ and $S_-$ as $s_+$ and $s_-$ respectively.

The subrepresentations $\CC_1,\CC_{-1}\sub V$ are only real in the split signature in which case they correspond to the light-cone coordinates $\partial_+ = \partial_x - \partial_t$ and $\partial_- = \partial_x + \partial_t$.

Pick complex vector spaces $W_+, W_-$ equipped with symmetric nondegenerate pairings. The supertranslation algebra is
\[V\oplus \Pi(S_+\otimes W_+\oplus S_-\otimes W_-),\]
where $\dim W_+ = \mc{N}_+$ and $\dim W_- = \mc{N}_-$. The R-symmetry group is $\O(\mc{N}_+)\times \O(\mc{N}_-)$.

A supercharge $Q=Q_++Q_-$ defines two subspaces $W_{Q_+}\sub W_+$ and $W_{Q_-}\sub W_-$ as in Definition \ref{rank_definition}.

\begin{lemma}
A supercharge $Q$ is square-zero if and only if the subspaces $W_{Q_{\pm}}\subset W_\pm$ are totally isotropic.
\end{lemma}
\begin{proof}
$Q^2 = 0$ if and only if $Q_+^2 = 0$ and $Q_-^2 = 0$. Each of these conditions is equivalent to the condition that the span of $Q_{\pm}$ in $W_{\pm}$ is isotropic.
\end{proof}

The group $G_R$ of R-symmetries $\O(W_+)\times \O(W_-)$ acts transitively on the set of totally isotropic subspaces, so the orbits of square-zero supercharges are classified by their rank.

The maximal supersymmetry possible for theories with fields of spin $\leq 1$ is $\mc N=(8, 8)$.

\paragraph{$\underline{\boldsymbol{\mc N=(1, 0)$, $\mc N=(1, 1)}}$:}

There are no square-zero supercharges.

\paragraph{$\underline{\boldsymbol{\mc N=(k, 0)}}$:}

The supercharge $Q_+=a\otimes s_+$ squares to zero if and only if $|a|=0$. Assume $a\neq 0$. Since the bilinear pairing is non-degenerate, there exists a $b\in W_+$ such that $(a, b)\neq 0$. Therefore, the image of $[Q, -]$ consists of $\partial_+$, i.e. this is a holomorphic supercharge.

These holomorphic supercharges can always be promoted to twisting data.

\begin{example}
Witten discussed the holomorphic twist of the $\mc N=(2,0)$ supersymmetric sigma model in \cite{Wittenmirror}: this theory is defined on a general Riemann surface and is sometimes called the ``half-twist'' (since it admits further twists to the A- and B-twists).  Kapustin described this theory in terms of the chiral de Rham complex \cite{KapustinChiraldR}, and more recently the factorization algebra modelling this twisted theory on a general Riemann surface and with general complex target was computed to be the sheaf of chiral differential operators on the target \cite{GorbounovGwilliamWilliams}.
\end{example}

\paragraph{$\underline{\boldsymbol{\mc N=(k, \ell)}}$:}

The non-zero R-symmetry orbits in the space of square-zero supercharges up to rescaling are classified as follows.

\textit{Square-zero supercharges}:
\begin{enumerate}
\item Rank (0, 1): we recover the case of $\mc N=(0, \ell)$ supersymmetry.  These supercharges are holomorphic.
\item Rank (1, 0): we recover the case of $\mc N=(k, 0)$ supersymmetry.  These supercharges are likewise holomorphic.
\item Rank (1, 1): the supercharge is topological.
\end{enumerate}

\textit{Twisting homomorphisms}:
Twisting homomorphisms in dimension 2 are atypical.  In general -- according to Theorem \ref{compatible_twisting_hom_thm} -- supercharges that are compatible with a twisting homomorphism are automatically topological.  In dimension 2 however this is false because the complex vector representation of $\SO(2)$ is reducible.  We can nevertheless investigate twisting homomorphisms and their compatible supercharges and check when they are topological.

\begin{enumerate}
 \item If $\mc N = (2,0)$ twisting homomorphisms $\phi_i \colon \SO(2) \to \SO(2)$ are indexed by winding numbers, i.e. by integers.  There is a compatible holomorphic supercharge if and only if $\phi$ is the identity $\phi_1$ or the inversion map $\phi_{-1}$.
 \item If $\mc N = (2,2)$ twisting homomorphisms $\phi_{i,j} \colon \SO(2) \to \SO(2) \times \SO(2)$ are indexed by pairs of integers.  There is a compatible holomorphic supercharge if and only if at least one of the integers is equal to 1 or $-1$, and a unique compatible topological supercharge if and only if both integers are equal to 1 or $-1$ (that is, in the four cases $\phi_{1,1}, \phi_{1,-1}, \phi_{-1,1}$ and $\phi_{-1,-1}$).
 \item More generally, for $\mc N = (k,\ell)$ twisting homomorphisms $\phi \colon \SO(2) \to \SO(k) \times \SO(\ell)$ are the same as pairs of coweights $(\lambda,\mu)$ for $\so(k)$ and $\so(\ell)$.  There is a compatible topological supercharge if and only if $\lambda$ and $\mu$ are both dual to weights of the defining representation. 
\end{enumerate}

\textit{$\ZZ$-gradings}:
All non-zero supercharges can be promoted to twisting data for a suitable choice of $\SO(2)$ inside the R-symmetry group.  This twisting data can be chosen compatibly with any twisting homomorphism from $\SO(2)$, i.e. to respect the decomposition into $\SO(2)$-weight spaces.

For $\mc N=(2,2)$ superconformal theories we may provide potentials for affine transformations (Definition \ref{def:affinepotential}).

\begin{prop} \label{2d_superconformal_potential}
Let $Q$ be a topological supercharge in the $\mc N=(2, 2)$ superconformal algebra. Then there exists a potential for affine transofrmations.
\end{prop}
\begin{proof}
Let $\alpha, \beta=\pm 1$ be choices of signs. By the classification of twisting homomorphisms, only $\phi_{\alpha, \beta}$ are compatible with a topological supercharge.

Let us denote by $\{G_i^+, G_i^-, L_n, J_n, \textbf{1}, \overline{G}_i^+, \overline{G}_i^-, \overline{L}_n, \overline{J}_n, \overline{\textbf{1}}\}$ the generators of the $\mc N=(2, 2)$ superconformal algebra as in the 2d case of Lemma  \ref{dilation_potential_lemma}. Then:
\begin{itemize}
\item The twisting homomorphism $\phi_{\alpha, \beta}$ sends the generator of rotations to $\frac{1}{2} (\alpha J_0 + \beta \overline{J}_0)$.

\item The compatible topological supercharge is $Q = G^{-\alpha}_{-1/2} + \overline{G}^{\beta}_{-1/2}$.

\item The algebra of translations is spanned by $L_{-1}$ and $\overline{L}_{-1}$.

\item The generator of rotations is $L_0-\overline{L}_0$.
\end{itemize}

We have
\begin{align*}
[Q, G^{\alpha}_{-1/2}] &= L_{-1} \\
[Q, \overline{G}^{-\beta}_{-1/2}] &= \overline{L}_{-1} \\
[Q, G^{\alpha}_{1/2} - \overline{G}^{-\beta}_{1/2}] &= L_0 + \frac{\alpha}{2}J_0 - \overline{L}_0 + \frac{\beta}{2} \overline{J}_0.
\end{align*}

Thus, we obtain potentials for translations and $\phi_{\alpha, \beta}$-twisted rotations. It is straightforward to check that these satisfy the equations for a potential of affine transformations.
\end{proof}

As a consequence of the previous Proposition and Theorem \ref{superconformal_theorem}, observables in topological twists of $\mc N=(2,2)$ superconformal field theories form algebras over the operad of framed little disks.

\begin{example}
An essential early example of twisting occurs for the 2d $\mc N=(2,2)$ supersymmetric sigma model associated to a K\"{a}hler manifold $X$ \cite{EguchiYang,Wittenmirror}. On the classical level the theory admits an $\SO(2)\times \SO(2)$ $R$-symmetry group. Let us denote the diagonal embedding by $\SO(2)_V$ and the antidiagonal embedding by $\SO(2)_A$. If $X$ is not Calabi-Yau, only a discrete subgroup of $\SO(2)_A$ acts on the quantum level.

The space of topological supercharges up to rescaling and Lorentz symmetry consists of four points we denote by $Q_A, Q^\dagger_A, Q_B, Q^\dagger_B$ (see e.g. \cite[Chapter 13]{MirrorSymmetry1}). The supercharges $Q_B, Q^\dagger_B$ are compatible with the twisting homomorphism $\SO(2)\rightarrow \SO(2)_A$ (i.e. $\phi_{1, -1}$) and the supercharges $Q_A, Q^\dagger_A$ are compatible with the twisting homomorphism $\SO(2)\rightarrow \SO(2)_V$ (i.e. $\phi_{1, 1}$). Moreover, the supercharges $Q_B, Q^\dagger_B$ extend to a twisting datum using $\SO(2)_V$ and $Q_A, Q^\dagger_A$ extend to a twisting datum using $\SO(2)_A$.

Thus, we see that if $X$ is not Calabi-Yau, the $B$-twists give rise to a $\ZZ$-graded theory which is merely framed while the $A$-twists give rise to a $\ZZ/2$-graded theory which is oriented.

A calculation of these twists in the language of this paper in terms of $\bb E_2$-algebras appearing as twists of the chiral de Rham complex was done by Ben-Zvi, Heluani and Szczesny \cite{BenZviHeluaniSzczesny}.
\end{example}

\begin{example}
One can also consider twists of $\mc N = (2,2)$ supersymmetric Yang--Mills in two dimensions.  Calculations for this topologically twisted theory on a general oriented surface appear for instance in \cite{MatsuuraMisumiOhta}, who show that the the stress-energy tensor is $Q$-exact (equation 2.9 in loc. cit.).
\end{example}

\subsection{Dimension 3}
The Lorentz group is $\Spin(3) \iso \SU(2)$. The vector representation is $V=\CC^3$ and the spinor representation is $S=\CC^2$. We have an isomorphism of representations
\[\sym^2(S) \iso V.\]

Pick a vector space $W$ equipped with a non-degenerate symmetric bilinear pairing. Then the supertranslation algebra is
\[T = V\oplus \Pi(S\otimes W)\]
where $\dim(W) = \mc N$. The R-symmetry group $G_R$ in this case is $\O(\mc N)$.

\begin{lemma}
A supercharge $Q\in S\otimes W$ is square-zero if and only if the the image of $Q\colon S^*\rightarrow W$ is a totally isotropic subspace $W_Q\subset W$.
\end{lemma}
\begin{proof}
The image of $Q\colon S^*\rightarrow W$ is totally isotropic if and only if the composite
\[\sym^2(S^*)\xrightarrow{\sym^2(Q)} \sym^2(W)\longrightarrow \CC\]
is zero, where the second map is given by applying the pairing on $W$.

Similarly, the supercharge $Q$ is square-zero if and only if the above composite considered as an element of $\sym^2(S)$ gives the zero element under the projection $\sym^2(S)\rightarrow V$, but the latter map is an isomorphism, so the two conditions are equivalent.
\end{proof}

The R-symmetry group acts transitively on the set of totally isotropic subspaces of a given rank while the complexified Lorentz group acts transitively on the set of subspaces in $S^*$. So, the orbits of square-zero supercharges are classified by their rank.

The maximal supersymmetry possible for theories with fields of spin $\leq 1$ is $\mc N=8$.

\paragraph{$\underline{\boldsymbol{\mc N=1}}$:}
There are no square-zero supercharges.

\paragraph{$\underline{\boldsymbol{\mc N=2}}$:} $Q$ has to have rank 1. These supercharges have 2 invariant directions.

\textit{Twisting homomorphisms}:
There are no twisting homomorphisms for the full group $\Spin(3)$.  If we choose a subgroup $\Spin(2) \sub \Spin(3)$ acting on a subspace of codimension 1 then there is a twisting homomorphism for each degree $d \in \ZZ$.  There is a unique holomorphic-topological square-zero supercharge compatible with the two twisting homomorphisms of degree $\pm 1$.

\textit{$\ZZ$-gradings}:
All supercharges can automatically be promoted to twisting data using the inclusion $\SO(2) \to \O(2)$.  Since $\O(2)$ is abelian, the twisting data commutes with twisting homomorphisms.

\begin{example}
One can obtain Chern--Simons theory as a deformation of the holomorphic twist of the pure $\mc N=2$ super Yang--Mills theory in dimension 3, see \cite{ACMV} which also discuss the theory coupled to a chiral multiplet.
\end{example}

\paragraph{$\underline{\boldsymbol{\mc N\geq 3}}$:} $ $

\textit{Square-zero supercharges}:
\begin{enumerate}
\item Rank 1. We are in the situation of the holomorphic-topological $\mc N=2$ twist. These supercharges have 2 invariant directions.

\item Rank 2 which is possible only for $\mc N\geq 4$. These supercharges are topological.
\end{enumerate}

\begin{remark}
If $\mc{N} = 4$, an isotropic embedding $S^*\hookrightarrow W$ is Lagrangian. A Lagrangian embedding defines a volume form on $W$. If we consider the R-symmetry group $\SO(4)\sub \O(4)$, it also defines a volume form on $W$ and we can compare the signs of the two volume forms. In particular, the space of Lagrangians in $W$ has two orbits under the $\SO(4)$-action, so we get two inequivalent supercharges.
\label{rmk:3dN4supercharges}
\end{remark}

\textit{Twisting homomorphisms}:
\begin{enumerate}
 \item $\mc N=3$: There is a unique twisting homomorphism given by the projection $\Spin(3) \to \SO(3)$, but there are no supercharges compatible with this twist.

 \item $\mc N=4$: The inclusions $i_1,i_2 \colon \SU(2) \to \Spin(4) \iso \SU(2)\times \SU(2) \to \SO(4)$ into the first and second factor are twisting homomorphisms. There is a unique topological supercharge of rank 2 compatible with each $i_1$ and $i_2$. These twisting homomorphisms are obtained by restricting an $\mc N=2$ twisting homomorphism in four dimensions down to three dimensions.
 
Note that these twisting homomorphisms are equivalent if we consider the R-symmetry group $\O(4)$ instead of $\SO(4)$ (see also Remark \ref{rmk:3dN4supercharges}). Usually only $\SO(4)$ acts on a supersymmetric field theory, so these twisting homomorphisms will typically yield different twisted theories, as we will see in Example \ref{RW_example} below.

 \item $\mc N>4$: According to the Jacobson--Morozov theorem twisting homomorphisms $\phi \colon \SU(2) \to \Spin(\mc N)$ are determined by a choice of nilpotent element $e$ in the Lie algebra $\so(\mc N)$ (uniquely up to conjugation by the centralizer of $e$).  Up to conjugation such nilpotent elements are determined by partitions of $\mc N$ \footnote{More specifically, nilpotent orbits are in bijection with partitions of $\mc N$ where even numbers occur with even multiplicity, with the complication that partitions consisting \emph{entirely} of even elements correspond to two nilpotent orbits \cite[Chapter 5]{CollingwoodMcGovern}.}, which determine in turn how the defining representation $W$ of $\so(\mc N)$ splits as an $\sl(2)$-representation.  There is a compatible topological supercharge if under this splitting there is at least one 2-dimensional irreducible summand, i.e. if the partition of $\mc N$ includes at least one 2.  For instance for $\mc N=8$ there are four possible twisting homomorphisms with this property up to conjugation (see \cite[Example 5.3.7]{CollingwoodMcGovern}).  
\end{enumerate}

\textit{$\ZZ$-gradings}:
It is always possible to choose a homomorphism $\mr U(1) \to \mr O(\mc N)$ so that a given supercharge has weight one, i.e. all supercharges can be promoted to twisting data.  If $\mc N=4$, then these twisting data can be made compatible with the twisting homomorphisms by extending $i \colon \SU(2) \to \SO(4)$ to $\SU(2) \times \mr U(1) \to \SO(4)$ using a maximal torus in the other copy of $\SU(2)$.  The $\SU(2)$-fundamental summand in $W$ then has $\alpha$-weight one.  

As in dimension 2, in the superconformal situation we can find potentials for the twisted rotation action.

\begin{prop} \label{3d_superconformal_potential}
Let $Q$ be the $\mc N=4$ topological supercharge compatible with the twisting homomorphism $\phi=i_1$. Then there exists a potential for affine transformations in the 3d $\mc N=4$ superconformal algebra.
\end{prop}

\begin{proof}
Recall that the odd part of the superconformal algebra can always be written as $\Sigma \oplus \Sigma^*$.  In dimension 3 explicitly, as an $\so(3) \oplus \gg_R$-representation we can write
\[\Sigma \oplus \Sigma^* \iso (S \otimes W) \oplus (S \otimes W).\]
We can restrict the $\so(3) \oplus \gg_R$-action in the $\mc N=4$ case to an action of the subalgebra $\so(3)_{\mr{tw}} \oplus \so(3)_2$, where $\gg_R$ splits as $\so(3)_1 \oplus \so(3)_2$ and $\so(3)_{\mr{tw}}$ is embedded in $\so(3) \oplus \gg_R$ using the twisting homomorphism, i.e. using $(1,i_1)$.  As an $\so(3)_{\mr{tw}} \oplus \so(3)_2$-representation we can decompose $\Sigma \oplus \Sigma^*$ as
\[\Sigma \oplus \Sigma^* \iso ((V \oplus \CC) \otimes S_2) \oplus ((V \oplus \CC) \otimes S_2)\]
where $V$ is the vector representation of $\so(3)_{\mr{tw}}$ and $S_2$ is the spin representation of $\so(3)_2$.

With respect to this decomposition $Q$ is a non-zero element of $\CC \otimes S_2 \sub \Sigma$.  Choose a non-zero spinor $s\in S_R$ which is not parallel to $Q$.  We can then define a potential for the translation action to be the subalgebra $\mf a = V \otimes s \sub \Sigma$, and a potential for the rotation action to be the subalgebra $\wt {\mf a} = V \otimes s \sub \Sigma^*$.  These two subalgebras are abelian, we just need to verify that they satisfy the conditions for a potential for affine transformations as in Definition \ref{def:affinepotential}.  By construction the potentials transform correctly under the action of $\so(3)_{\mr{tw}}$.

Recall that the bracket
\[[\Sigma, \Sigma^*]\sub \CC\cdot D\]
is antisymmetric and the bracket
\[[\Sigma, \Sigma^*]\sub \so(3)\oplus \gg_R\]
is symmetric. Therefore $[\mf a, \wt {\mf a}] = 0$.

It remains to check that $Q$ carries $\mf a$ and $\wt {\mf a}$ isomorphically onto the algebras of translations and rotations respectively.  This is clear for $\mf a$.  Since $Q$ is scalar with respect to $\so(3)_2$, the map
\[[Q, -]\colon \wt {\mf a} \cong \so(3)\longrightarrow \so(3)\oplus \gg_R\]
takes the form $(\id, \alpha \iota_1)$ for some constant $\alpha$. But $[Q, [Q, \wt {\mf a}]] = 0$ from which we deduce that $\alpha = 1$, i.e. $\wt {\mf a}$ is the primitive of $\iota_1$-twisted rotations.
\end{proof}

\begin{example} \label{RW_example}
The 3d $\mathcal{N}=4$ supersymmetric $\sigma$-model with target a hyperK\"ahler manifold admits the $R$-symmetry group $\SO(3)\subset \SO(4)$ obtained by quotienting by $\ZZ/2$ the embedding $i_1\colon \SU(2)\rightarrow \Spin(4)$. Let us explain the classification of topological supercharges in this case. Recall that rank 2 supercharges correspond to Lagrangians $S^*\rightarrow W$. The space of Lagrangians in $W$ is isomorphic to $\bb{CP}^1\sqcup \bb{CP}^1$. The action of $\SO(3)\subset \SO(4)$ on the first $\bb{CP}^1$ is trivial and on the second $\bb{CP}^1\cong S^2$ is given by rotations. Thus, the space of orbits is $\bb{CP}^1\sqcup \pt$.

The $\bb{CP}^1$ of topological twists is compatible with the twisting homomorphism $i_1$ and admits a description in terms of the Rozansky--Witten theory \cite{RozanskyWitten}. Note that it does not extend to a twisting datum, so this theory is only naturally $\ZZ/2$-graded. The extra point corresponds to the ``twisted'' Rozansky--Witten theory which is $\ZZ$-graded, but is incompatible with any twisting homomorphism. This twist in certain cases admits a description in terms of the 3d A-model as defined by Kapustin and Vyas \cite{KapustinVyas}.
\end{example}

\begin{example} 
More generally, one can construct three-dimensional $\mc N=4$ supersymmetric \emph{gauged} sigma models with target a hyperk\"ahler manifold \cite{GaiottoWittenJanus}.  If the target is flat, then this theory is superconformal.  Kapustin and Saulina \cite{KapustinSaulina} compute the twist of this theory with respect to the twisting homomorphism $i_1$.  This theory is defined on general oriented 3-manifolds, and is equivalent to a gauged version of Rozansky--Witten theory.  Boundary conditions for this twist in various examples, along with its 3d mirror -- the other twist via $i_2$ -- are discussed in \cite{Gaiotto3d}.
\end{example}

\begin{example}
We can also discuss twists of $\mc N=4$ supersymmetric Yang--Mills theories in three dimensions.  As in the case of supersymmetric sigma models, there are two mirror possibilities, both defined on arbitrary oriented three-manifolds.  These twists have been studied by Blau and Thompson.  The twist by the twisting homomorphism $i_1$ was studied in \cite{BlauThompson1}.  It arises via dimensional reduction from the Donaldson--Witten twist of $\mc N=2$ gauge theory in dimension 4 (Example \ref{DW_example} below), has $Q$-exact stress-energy tensor and an action functional of BF theory type.  This twist can be used to compute the Casson invariant of a 3-manifold.  The other twist, using $i_2$, has been analysed in \cite{BlauThompson2}, where it was shown that its stress-energy tensor is $Q$-exact (see \cite[equation 4.45]{BlauThompson2}.

Blau and Thompson also discussed a topological twist of $\mc N=8$ three-dimensional super Yang--Mills.  Again the theory is defined on arbitrary oriented three-manifolds and has exact stress-energy tensor.  The twisting homomorphism they consider is one of the two where $W$ splits as four copies of the fundamental representation.
\end{example}

\begin{example}
Bullimore, Dimofte, Gaiotto, Hilburn and Kim \cite{BDGH1, BDGH2} have analysed the 3d mirror symmetry relating the Higgs and Coulomb branches of general 3d $\mc N=4$ theories.  They obtain these branches by forming successive twists, first by a holomorphic supercharge $Q$, then by a further square-zero supercharge $Q'$ so that $Q + Q'$ is topological.  There are two complications in this story \footnote{These subtleties were explained to us by Philsang Yoo.}.
\begin{enumerate}
 \item Twisting by $Q$ then by $Q'$ does not give the same result as twisting by $Q+Q'$ (i.e. the Rozansky-Witten and twisted Rozansky-Witten twists of the 3d $\mc N=4$ theory).  Instead there is a spectral sequence from the iterated twist that converges to the twist by $Q+Q'$.
 \item One can also consider mass deformations of the 3d $\mc N=4$ theory.  After performing such deformations only a central extension of the supersymmetry algebra acts, and in this extension the supercharges $Q$ and $Q'$ no longer commute \cite[section 2.2]{BDGH2}.  One can however take the twist by $Q'$ after considering appropriate invariants for the $\mr U(1)$-action generated by $(Q')^2$ (as in the twisted $\Omega$-background).
\end{enumerate}
\end{example}

\subsection{Dimension 4}
The Lorentz group is $\Spin(4) \iso \SU(2)\times \SU(2)$.  We denote its irreducible representations by pairs $[n, m]$ of non-negative integers.  The dimension of such a representation is $(n+1)(m+1)$.

The vector representation $V=\CC^4$ is the representation $[1, 1]$. The spinor representations are $S_+=\CC^2$ of type $[1, 0]$ and $S_- = \CC^2$ of type $[0, 1]$. We have an isomorphism
\[\Gamma \colon S_+\otimes S_-\stackrel{\sim}\rightarrow V.\]

Pick a vector space $W$. Then the supertranslation algebra is
\[T = V\oplus \Pi(S_+\otimes W\oplus S_-\otimes W^*),\]
where $\dim W = \mc N$. The R-symmetry group $G_R$ in this case is $\GL(\mc N;\CC)$.

\begin{lemma}
A supercharge $Q$ squares to zero if and only if $\langle W_{Q_-}, W_{Q_+}\rangle = 0$, where the subspaces $W_{Q_\pm}$ are as in Definition \ref{rank_definition}, and where $\langle -, -\rangle$ is the natural pairing between $W$ and $W^*$.
\end{lemma}
\begin{proof}
Indeed, both conditions are equivalent to the condition that the composite
\[S_+^*\otimes S_-^*\xrightarrow{Q_+\otimes Q_-} W\otimes W^*\longrightarrow \CC,\]
where the second map is given by the natural pairing.
\end{proof}

Thus, the data of a square-zero supercharge $Q$ is given by specifying a flag $W_{Q_+}\sub W_{Q_-}^{\perp}\sub W$ together with surjective maps $S_+^*\rightarrow W_{Q_+}$ and $S_-^*\rightarrow W_{Q_-}$. The R-symmetry group $\GL(W)$ acts transitively on the set of flags $W_{Q_+}\sub W_{Q_-}^{\perp}\sub W$. Moreover, the complexified Lorentz group $\SL(2; \CC)\times \SL(2; \CC)$ acts transitively on the set of subspaces of $S_+^*$ and $S_-^*$. So, the orbits of square-zero supercharges are classified by their rank.

For any pairs of automorphisms of $W_{Q_+}$ and $W_{Q_-}$ we can find an element of $\GL(W)$ which stabilizes this flag and acts on $W_{Q_+}$ and $W_{Q_-}$ by the given automorphisms, so the orbits of square-zero supercharges are classified by their rank.

Compactification to $d=3$ corresponds to the isomorphism $S^3\cong S_+^4\cong S_-^4$ of $\so(3)$-representations. The auxiliary spaces are identified as $W^3=W^4\oplus (W^4)^*$ with the standard symmetric nondegenerate pairing.

\paragraph{$\underline{\boldsymbol{\mc N=1}}$:}
A supercharge is square-zero if and only if it has rank (1, 0) or (0, 1). Such a square-zero supercharge is holomorphic.  The stabiliser of such a supercharge is isomorphic to $\SU(2)$.  These supercharges can be extended to twisting data using the entire R-symmetry group, which is isomorphic to $\GL(1;\CC) = \CC^\times$.

\begin{example} \label{Costello_4d_twist}
The holomorphic twist of $\mc N=1$ super Yang--Mills has been studied by Johansen \cite{Johansen} in a physical context, and later by Costello \cite{CostelloYangian} in a mathematical context.  Johansen refers to the twist using different but equivalent terminology to that which we have used in this paper: one can consider super Yang-Mills in an external supergravity background then define the twist by giving a definite non-zero value to the ghost for the gravitino.  The twisted theory is equivalent to holomorphic BF theory and defined on an arbitrary complex surface.  This twisted theory can be deformed by the addition of a Chern--Simons term to the action functional of a holomorphic-topological 4d theory.  Costello proves that this deformed theory admits a perturbative quantization whose algebra of local observables is Koszul dual to the Yangian for the gauge group.

One can additionally couple the holomorphically twisted deformed $\mc N=1$ theory to a chiral multiplet.  This twisted coupled theory is discussed in \cite[Section 5]{ACMV}.
\end{example}

\begin{example}
Witten \cite{WittenYMDonaldson} studied the holomorphic twist of mass deformations of $\mc N=2$ Yang-Mills theories, i.e. $\mc N=1$ Yang-Mills theories with a massive adjoint chiral multiplet.  He shows that these twisted theories are defined on compact K\"ahler surfaces, and uses them to compute the Donaldson invariants of such surfaces.
\end{example}

\paragraph{$\underline{\boldsymbol{\mc N=2}}$:}$ $

\textit{Square-zero supercharges}:
\begin{enumerate}
 \item Rank $(1,0)$ and $(0,1)$: these supercharges are holomorphic.
 \item Rank $(2,0)$ and $(0,2)$: these supercharges are topological.
 \item Rank $(1,1)$: square-zero supercharges of this type have 3 invariant directions.
\end{enumerate}

Now let us consider twisting homomorphisms.  These are homomorphisms $\SU(2) \times \SU(2) \to \GL(2;\CC)$.  All such homomorphisms factor through the projection to one of the factors, say the first factor (the other projection is equivalent under the outer automorphism of $\Spin(4)$, although we should note this does not guarantee that the associated twisted theories are equivalent without a further calculation).  There is a unique such map up to post-composition with an automorphism of $\GL(2;\CC)$.

\textit{Twisting homomorphisms}:
\begin{enumerate}
 \item The projection $\phi = \pi_1 \colon \SU(2) \times \SU(2) \to \GL(2;\CC)$ is a twisting homomorphism.  There is a unique compatible topological supercharge of rank $(2,0)$.
 \item Likewise the projection $\phi = \pi_2 \colon \SU(2) \times \SU(2) \to \GL(2;\CC)$ is a twisting homomorphism.  There is a unique compatible topological supercharge of rank $(0,2)$.
\end{enumerate}

\textit{$\ZZ$-gradings}:
We can associate twisting data in $\GL(2;\CC)$ to any square-zero supercharge $(a,b)$.  This twisting datum can be chosen to factor through $\SL(2;\CC)$ for supercharges of type $(1,0)$ and $(1,1)$, but not those of type $(2,0)$, where we are forced to choose the diagonal $\CC^\times$ in $\GL(2;\CC)$.  Only the diagonal twisting datum is compatible with either twisting homomorphism, so only supersymmetric field theories with an action of the full R-symmetry group $\GL(2;\CC)$ have the potential to twist to a $\spin(4)$-structured $\ZZ$-graded theory.

\begin{example} \label{DW_example}
There is a famous example of a topologically twisted $\mc N=2$ theory in four-dimensions: that of \emph{Donaldson--Witten theory} as a twist of pure $\mc N=2$ super Yang--Mills theory.  This was one of the original motivating examples for the topological twisting procedure as introduced by Witten \cite{WittenTQFT}.  One twists the $\mc N=2$ gauge theory with respect to a supercharge of type $(2,0)$ and the compatible twisting homomorphism and obtains a theory defined on aribtrary oriented 4-manifolds. Its moduli space of classical solutions models the Donaldson theory moduli space of instantons, and Witten used the twisted theory to interpret Donaldson invariants in terms of correlation functions.

One can twist such an $\mc N=2$ theory coupled to a hypermultiplet valued in some representation $R \oplus R^*$ by the same supercharge and twisting homomorphism \cite{AlvarezLabastida, HyunParkPark}.  In particular, for $G = \SU(2)$ and the fundamental hypermultiplet the twisted theory is related to the theory of Seiberg--Witten invariants (this example was also constructed by Labastida and Mari\~no \cite{LabastidaMarino}). These twisted theories are still defined on arbitrary oriented four-manifolds, and the moduli spaces of solutions now model solutions to the generalized monopole equations.
\end{example}

\begin{example}
The holomorphic twist of $\mc N=2$ super Yang--Mills was recently discussed by Cautis and Williams \cite{CautisWilliams}.  Specifically, they discuss its category of line operators, which are modelled by a coherent version of the Satake category for the gauge group $G$ (for instance using the same analysis as we will discuss in Example \ref{Kapustin_Witten_example}).
\end{example}

\begin{example} \label{Kapustin_twist}
Kapustin \cite{KapustinHolo} studied a holomorphic-topological twist of $\mc N=2$ super Yang--Mills theory coupled to a hypermultiplet, defined on the product of two Riemann surfaces.  This is the twist by a supercharge of type $(1,1)$ with respect to a twisting homomorphism from $\mr U(1) \times \mr U(1) \to \GL(2)$.
\end{example}

We could go on to talk about the $\mc N=3$ supersymmetry, but we will not include it.  In fact, every representation of the $\mc N=3$ supersymmetry algebra with spins at most one arises by restriction from a representation of the $\mc N=4$ supersymmetry algebra, so if one is interested in renormalizable supersymmetric quantum field theories in dimension 4, one does not obtain any new physics by considering $\mc N=3$ theories.  We refer to the discussion in \cite[Section 25.4]{WeinbergIII}.

\paragraph{$\underline{\boldsymbol{\mc N=4}}$:} $ $

\textit{Square-zero supercharges}:
\begin{enumerate}
 \item Rank $(1,0)$ and $(0,1)$: these supercharges are holomorphic.
 \item Rank $(2,0)$ and $(0,2)$: these supercharges are topological.
 \item Rank $(1,1)$: square-zero supercharges of this type have 3 invariant directions.
 \item Rank $(2,1)$ and $(1,2)$: square-zero supercharges of this type are topological.
 \item Rank $(2,2)$: square-zero supercharges of this type are topological.
\end{enumerate}

\begin{remark}
This classification of square-zero supercharges also applies for $\mc N > 4$.
\end{remark}

\begin{remark}
A rank $(2, 2)$ supercharge defines embeddings $S_+^*\sub W$ and $S_-^*\sub W^*$. The square-zero condition implies that we have an exact sequence
\[0\rightarrow S_+^*\rightarrow W\rightarrow S_-\rightarrow 0.\]
If we consider the full R-symmetry group $\GL(4; \CC)$, there is a single orbit of rank $(2, 2)$ square-zero supercharges. However, if we consider only $\SL(4; \CC)$, we obtain a volume form on $W$. We also have canonical volume forms on $S_-$ and $S_+$. Therefore, we may compare the volume forms using the above exact sequence and this gives an invariant of such supercharges lying in $\CC^\times$. Therefore, the $\SL(4; \CC)$-orbits of rank (2, 2) square-zero supercharges are parametrized by $\CC^\times$.
\end{remark}

We can also classify the twisting homomorphisms.  There are three twisting homomorphisms from $\Spin(4)$ to $\GL(4;\CC)$ (all of which factor through $\SL(4;\CC)$) which admit compatible supercharges, up to automorphisms of the source and target.  See Lozano \cite{Lozano} for a discussion of these twists.

\textit{Twisting homomorphisms}:
\begin{enumerate}
 \item Kapustin--Witten twist \cite{KapustinWitten}: $\phi_{\mr{KW}} \colon \SU(2) \times \SU(2) \inj \GL(4)$ defined by $(A,B) \mapsto \mr{diag}(A,B)$.
       This twist is compatible with a $\bb{CP}^1$ of topological supercharges.  The poles have rank $(2,0)$ and $(0,2)$ and the remaining $\CC^\times$ worth have type $(2,2)$. This twist was previously studied by Marcus \cite{Marcus}.
 \item Vafa--Witten twist \cite{VafaWitten} (see also \cite[Section 3]{Yamron}): $\phi_{\mr{VW}} \colon \SU(2) \times \SU(2) \to \GL(4)$ defined by $(A,B) \mapsto \mr{diag}(A,A)$.
       This twist is also compatible with a $\bb{CP}^1$ of topological supercharges, all of which have rank $(2,0)$.  The outer automorphism of $\Spin(4)$ turns this into $(A,B) \mapsto \mr{diag}(B,B)$ where now the compatible topological supercharges have rank $(0,2)$.
 \item Half twist \cite[Section 2]{Yamron}: $\phi_{\mr{half}} \colon \SU(2) \times \SU(2) \to \GL(4)$ defined by $(A,B) \mapsto \mr{diag}(A,1,1)$.
       This twist is compatible with a single topological supercharge of type $(2,0)$.  The outer automorphism of $\Spin(4)$ turns this into $(A,B) \mapsto \mr{diag}(B,1,1)$ where now the compatible topological supercharge has rank $(0,2)$.
\end{enumerate}

\begin{remark}
More generally, for $\mc N \ge 4$, twisting homomorphisms from $\Spin(4) \to \SL(\mc N; \CC)$ are given by $\mc N$-dimensional representations of $\SU(2) \times \SU(2)$ with at least one irreducible summand isomorphic to $S_+$ or $S_-$.  To count the twisting homomorphisms one must enumerate the partitions of $\mc N-2$ of the form $\sum_{i=1}^k (m_i+1)(n_i+1)$, where $m_i$ and $n_i$ are non-negative integers.
\end{remark}

\textit{$\ZZ$-gradings}:
We can associate twisting data in $\SL(4;\CC)$ to any square-zero supercharge $Q = a+b$. We obtain maps $S_a\hookrightarrow W$ and $W\twoheadrightarrow S_b^*$. Moreover, the square-zero condition implies that the composite is zero. Choose a $\mr U(1)$-subgroup of the R-symmetry group $\SL(4;\CC)$ where $S_a$ has weight 1 and $S_b^*$ has weight $-1$.  If the supercharge $Q$ is compatible with one of the above twisting homomorphisms $\phi$, then this $\mr U(1)$ automatically commutes with the image of $\phi$.

In the $\mc N=4$ superconformal case, we can find a potential for affine transformations \emph{only} for the Kapustin--Witten twist, and generic compatible topological supercharges of rank $(2,2)$.

\begin{prop} \label{4d_superconformal_potential}
Let $Q$ be an the $\mc N=4$ topological supercharge of rank $(2,2)$ compatible with the twisting homomorphism $\phi=\phi_{\mr{KW}}$.  There is a strict potential for the generator of $\phi$-twisted rotations in the 4d $\mc N=4$ superconformal algebra.
\end{prop}

\begin{proof}
The proof proceeds in the same manner as Proposition \ref{3d_superconformal_potential}.  Recall that the odd part of the superconformal algebra can always be written as $\Sigma \oplus \Sigma^*$.  In dimension 4 explicitly, as an $\so(4) \oplus \gg_R$-representation we can write
\[\Sigma \oplus \Sigma^* \iso (S_+ \otimes W \oplus S_- \otimes W^*) \oplus (S_+ \otimes W^* \oplus S_- \otimes W).\]
We can restrict the $\so(4) \oplus \gg_R$-action in the $\mc N=4$ case to an action of the subalgebra $\so(4)_{\mr{tw}}$ defined by the twisting homomorphism $\phi_{\mr{KW}}$.  As an $\so(4)_{\mr{tw}}$-representation we can decompose the summands $\Sigma$ and  $\Sigma^*$ as
\[\Sigma \iso \Sigma^* \iso (\sym^2(S_+) \oplus V \oplus \CC) \oplus (\sym^2(S_-) \oplus V \oplus \CC)\]
where $S_\pm$ are the semispin representations of $\so(4)_{\mr{tw}}$ and $V$ is its vector representation.

With respect to this decomposition $Q$ is a linear combination of elements in the two copies of $\CC$ which is non-zero under projection onto either factor.  We can then define a potential for the translation action to be the subalgebra $\mf a = V\sub \Sigma$ (say, the first copy), and a potential for the rotation action to be the subalgebra $\wt{\mf a} = \sym^2(S_+) \oplus \sym^2(S_-)\sub \Sigma^*$.  These two subalgebras are abelian, we just need to verify that they satisfy the conditions for a potential for affine transformations as in Definition \ref{def:affinepotential}.  By construction the potentials transform correctly under the action of $\so(4)_{\mr{tw}}$.

Just as in the three-dimensional case, consider the bracket $[\Sigma, \Sigma^*]$ as a pairing on $\Sigma$ viewed as an $\so(4)_{\mr{tw}}$-module.  Recall that the pairing
\[[\Sigma, \Sigma^*]\sub \CC\cdot D\]
pairs irreducible summands with themselves and the pairing
\[[\Sigma, \Sigma^*]\sub \so(4)\oplus \gg_R\]
pairs the summands $S_\pm \otimes S_\pm$ and $S_\pm \otimes S_\mp$ with themselves. Therefore $[\mf a, \wt {\mf a}] = 0$. 

It remains to check that $Q$ carries $\mf a$ and $\wt {\mf a}$ isomorphically onto the algebras of translations and rotations respectively.  This is clear for $\mf a$.  Consider the map
\[[Q, -]\colon \wt {\mf a} \cong \so(4)\longrightarrow \so(4)\oplus \gg_R\]
and denote it by $(f,\phi)$. Here $f$ is an adjoint-equivariant map $\so(4) \to \so(4)$.  The two projections $\pi_{\pm} \circ f \colon \so(4) \to \so(3)$ can be seen to be non-zero from the explicit formulas for the brackets (see \cite{Minwalla}) which implies $f = \mr{id}$.  Finally $[Q, [Q, \wt {\mf a}]] = 0$ from which we deduce that $\wt{\mf a}$ is compatible with the twisting homomorphism $\phi$, which means $\phi = \phi_{\mr{KW}}$, i.e. $\wt {\mf a}$ is the primitive of $\phi_{\mr{KW}}$-twisted rotations.
\end{proof}

\begin{remark}
If we fix coordinates, we can write the potential for rotations explicitly.  Choose a basis $Q^+_{ai}, Q^-_{bj}$ for the space $\Sigma$ of supercharges and $S^+_{ai}, S^-_{bj}$ for the space $\Sigma^*$ of special supercharges, where $a=1,\ldots, 4$ ranges over a basis for $W$ and the dual basis for $W^*$ (orthonormal for the metric induced by $\phi_{\mr{KW}}$) and $i=1,2$ over a basis for $S_\pm$.  In such a basis we can write the $\bb{CP}^1$ family of topological supercharges as $Q_{(\lambda:\mu)} = \lambda(Q^+_{11} + Q^+_{22}) + \mu(Q^-_{31} + Q^-_{42})$.  We are choosing a generic point in this family, so for notational simplicity set $\lambda=\mu=1$.  To fix a potential, write
\begin{align*}
U_+ &= \langle S^+_{11} - S^+_{22}, S^+_{12} - S^-_{32}, S^+_{12} - S^-_{41}  \rangle \\
\text{and } U_- &= \langle S^+_{21} - S^-_{32}, S^+_{21} - S^-_{41}, S^-_{31} - S^-_{42}  \rangle.
\end{align*}
Then in these coordinates we have $\wt {\mf a} = U_+ \oplus U_-$.
\end{remark}

\begin{remark}
Topological supercharges of type $(2,0)$ or $(0,2)$ automatically commute with all but 4 special supercharges in the $\mc N=4$ superconformal algebra, so there cannot be potentials for the full algebra $\so(4)$ of rotations. However, one can hope for potentials for the subalgebras $\so(3)\sub \so(4)$ or $\so(2)\oplus\so(2)\sub \so(4)$, so that twisted theories yield $\bb E_4^{\SO(3)}-$ or $\bb E_4^{\SO(2)\times \SO(2)}$-algebras.  For instance for the Kapustin--Witten twist and the compatible topological supercharges of either type $(2,0)$ or $(0,2)$ either the space $U_+$ or the space $U_-$ above will generate a subalgebra of the form $[Q,U_\pm] = \so(3;\CC) \sub \so(4;\CC)$.
\end{remark}

\begin{example} \label{Kapustin_twist_2}
Kapustin's Example \ref{Kapustin_twist} defines a holomorphic-topological twist of $\mc N=4$ super Yang--Mills theory when one considers the example of adjoint matter.  This is the twist by a supercharge of type $(1,1)$ with respect to the Kapustin--Witten twisting homomorphism, defined on a product of two Riemann surfaces.
\end{example}

\begin{example}
The Vafa--Witten twist of $\mc N=4$ super Yang--Mills was analysed in \cite{VafaWitten}.  They analyse the twisted theory on general oriented 4-manifolds.  Vafa and Witten showed that the partition function of the twisted theory on $M^4$ computes the Euler characteristic of the moduli space of instantons on $M$.
\end{example}

\begin{example} \label{Kapustin_Witten_example}
Kapustin and Witten computed a $\bb{CP}^1$ family of topological twists of $\mc N=4$ super Yang--Mills theory using the Kapustin--Witten twisting homomorphism in \cite{KapustinWitten}, defined -- generically -- on arbitrary oriented 4-manifolds.  More generally they also discuss a larger, $\bb{CP}^1 \times \bb{CP}^1$-family of twists which are generally only compatible with a twisting homomorphism from $\Spin(2) \times \Spin(2)$: the relevant supercharges are generically topological of type $(2,2)$, but upon flowing to $0$ or $\infty$ in one or both factors the supercharges degenerate.  There are two $\CC^\times$ strata of type $(1,2)$ and two of type $(2,1)$, and four special points: one of type $(2,0)$, one of type $(0,2)$ and two holomorphic-topological points of type $(1,1)$ as in Kapustin's Example \ref{Kapustin_twist_2} above (these dimensionally reduce to topological theories in dimension 2).

In \cite{ElliottYoo} a $\bb{CP}^1$ family of such twists was analysed from the point of view of derived geometry, keeping track of algebraic structures.  This $\bb{CP}^1$ arises as the family of topological twists extending a fixed holomorphic twist.  Generically these twists used supercharges of type $(2,2)$ with poles of type $(2,0)$ and $(1,2)$.  As such, these theories can be defined generically only on products of Riemann surfaces, not on all oriented 4-manifolds.  The holomorphic twist was shown to recover the derived moduli stack of $G$-Higgs bundles, and the family of further topological twists recovers the derived moduli stack of flat $G$-bundles and the de Rham stack of holomorphic $G$-bundles at the two points 0 and $\infty$ in $\bb{CP}^1$.  These are the expected moduli stacks from the point of view of the geometric Langlands program.
\end{example}

\begin{example}
From our point of view we can explain -- without needing to do any calculations -- the result of Setter \cite{Setter} that the Vafa--Witten twist and the Kapustin--Witten twist at parameter $t = 0 \in \bb{CP}^1$ become equivalent after dimensional reduction on a circle.  Indeed, there is a topological supercharge of type $(2,0)$ compatible with both the Vafa--Witten \emph{and} Kapustin--Witten twisting homomorphisms.  Furthermore precomposing these two twisting homomorphisms with the inclusion $\Spin(3) \to \Spin(4)$ they become equivalent.  Therefore, as $\Spin(3)$-structured 4d theories, or as oriented theories after reduction to three dimensions, these two twists are necessarily equivalent.
\end{example}

\subsection{Dimension 5}
The Lorentz group is $\Spin(5)\cong \mathrm{USp}(4)$. Let $V=\CC^5$ be the vector representation and $S=\CC^4$ the spinor representation. We have an equivariant isomorphism
\[\wedge^2(S) \iso \CC\oplus V,\]
where the summand $\CC\rightarrow \wedge^2(S)$ corresponds to the given symplectic structure on $S$ (that is, it is spanned by the associated bivector).

Let $W$ be a symplectic vector space. The supertranslation algebra is
\[T= V\oplus \Pi(S\otimes W)\]
where $\dim W = 2\mc N$. The R-symmetry group $G_R$ in this case is $\Sp(2\mc N, \CC)$.

\begin{lemma}
A supercharge $Q\in S\otimes W$ is square-zero if and only if the map $Q\colon W^*\rightarrow S$ satisfies $Q_* \pi_{W^*} = \lambda \pi_S$ for some $\lambda\in\CC$, where $\pi_{W^*}$ and $\pi_S$ are the Poisson bivectors on $W^*$ and $S$ respectively.
\end{lemma}
\begin{proof}
The square-zero condition is equivalent to the composite
\[\CC\xrightarrow{\pi_{W^*}} \sym^2(W^*)\xrightarrow{\sym^2(Q)} \sym^2(S)\longrightarrow V\]
being zero. In other words, $(Q_*\pi_{W^*})\colon \CC \rightarrow \sym^2(S)$ factors through $\pi_S\colon \CC\rightarrow \sym^2(S)$, i.e. through the kernel of the pairing.
\end{proof}

Since we only consider supercharges up to scale, we may assume $\lambda = 1$.

Compactification to $d=4$ corresponds to the isomorphism $S^5\cong S^4_+\oplus S^4_-$ as $\so(4)$-representations. The auxiliary spaces are identified: $W^4=W^5$.

\paragraph{$\underline{\boldsymbol{\mc N=1}}$:}$ $

\textit{Square-zero supercharges}: $Q$ has to have rank 1 in which case it automatically squares to zero. These supercharges have 3 invariant directions. Such supercharges are determined by specifying a embedded lines $W_Q\sub W$ and $W_Q^*\sub S$ as in Definition \ref{rank_definition}. The symplectic group acts transitively on the space of lines, so there is a single orbit of this type.

\begin{remark}
Suppose $Q$ has rank 2. Then we get an embedding $\CC\cong \wedge^2 W^*\hookrightarrow \wedge^2 S$. Since $\pi_S$ has rank 2, this embedding is never proportional to $\pi_S$ and hence there are no square-zero supercharges of this type.
\end{remark}

\textit{Twisting homomorphisms}:
There are no twisting homomorphisms for the full R-symmetry group.  There is however a unique twisting homomorphism from $\Spin(3) \sub \Spin(5)$ for each subgroup induced from a three-dimensional subspace of $V=\CC^5$.  For each such twisting homomorphism there are two compatible rank 1 (therefore square-zero) holomorphic-topological supercharges which are interchanged by the action of $\Spin(5)$.

\textit{$\ZZ$-gradings}:
Every square-zero supercharge can automatically be promoted to compatible twisting data using a suitable maximal torus $\mr U(1) \sub \Sp(2)$.  These twisting data are not compatible with the twisting homomorphism from $\Spin(3)$ These twisting data are not compatible with the twisting homomorphism from $\Spin(3)$ because there is no embedding $\mr U(1) \inj \Sp(2)$ acting on the defining representation with weight one.

\begin{example}
The holomorphic-topological twist of $\mc N=1$ super Yang--Mills in five dimensions is discussed by K\"all\'en and Zabzine \cite{KallenZabzine} (see in particular Section 3.5 of loc. cit.).  They define the twisted theory on arbitrary contact five-manifolds, and in particular the classical solutions to their equations of motion include 5d contact instantons.  This example fits into a more general construction of holomorphic-topological twists of odd-dimensional gauge theories due to Baulieu, Losev and Nekrasov \cite{BaulieuLosevNekrasov} which we will also mention in examples in seven and nine dimensions.
\end{example}

\paragraph{$\underline{\boldsymbol{\mc N=2}}$:}

Pick $Q\in S\otimes W$.

\textit{Square-zero supercharges}:
\begin{enumerate}
\item Rank 1.  As in the $\mc N=1$ case, these supercharges all square to zero. They have 3 invariant directions.

\item Rank 2. Then $Q$ squares to zero if and only if the image of the induced map $S^*\rightarrow W$ is Lagrangian. We have the following two subcases:
\begin{itemize}
\item If the image of the dual map $W^*\rightarrow S$ is symplectic, then the supercharge is topological.

\item If the image of the dual map $W^*\rightarrow S$ is Lagrangian, the supercharge has 4 invariant directions.
\end{itemize}

Since the symplectic group acts transitively on the space of Lagrangian or symplectic subspaces of a given dimension, there is a single orbit of square-zero supercharges in each case.

\item Rank 3. These supercharges never square to zero.

\item Rank 4. Then $Q$ defines an isomorphism $W^*\rightarrow S$. The symplectic group acts transitively on the space of symplectomorphisms, so there is a single orbit of supercharges of this type. This twist is topological.
\end{enumerate}

\textit{Twisting homomorphisms}:
\begin{enumerate}
\item There is a unique twisting homomorphism given by the embedding of the maximal compact subgroup $\Spin(5)\cong \mathrm{USp}(4) \to \Sp(4;\CC)$.  As mentioned above, the tensor square of the spin representation $S$ splits as $\sym^2 S \oplus V \oplus \CC$, so there is a unique compatible topological twist corresponding to the rank 4 supercharge given above.

\item There are additional interesting partial twisting homomorphisms $\Spin(4) \sub \Spin(5)$.  We can use any of the three 4d $\mc N=4$ twisting homomorphisms since they all land in $\Sp(4;\CC) \sub \GL(4;\CC)$.  Their compatible supercharges are those rank 2 and rank 4 supercharges discussed in the section on 4d $\mc N=4$ superalgebras.
\end{enumerate}

\textit{$\ZZ$-gradings}:
For rank 1 and 2 supercharges one can choose an embedding $\mr U(1) \to \Sp(4; \CC)$ so that they have weight 1, i.e. promote the supercharges to twisting data.  These twisting data cannot be chosen compatibly with the partial twisting homomorphisms from $\Spin(4)$.  Indeed under the three 4d $\mc N=4$ twisting homomorphisms $W$ decomposes as either $S_+ \oplus S_-, S_+ \oplus S_+$ or $S_+ \oplus \CC^2$ as a $\Spin(4)$-module.  In each case the two summands are forced to be symplectic subspaces, but a twisting datum $\alpha \colon \mr U(1) \to \Sp(4)$ acting on a compatible topological supercharge with weight one must act on the first summand with weight one in each case, which is impossible.

For rank 4 supercharges there is no compatible twisting datum even independent of twisting homomorphisms, so these theories are only naturally $\ZZ/2$ graded.

\begin{remark}
This is the last dimension where there exists a non-zero twisting homomorphism from the full spin group.
\end{remark}

\begin{remark}[5d Superconformal Theories] \label{5d_superconformal_rmk}
There is an $\mc N=1$ superconformal algebra in dimension 5, but no $\mc N=2$ superconformal algebra.  We might however try to construct potentials for twisted rotations starting from $\mc N=(2,0)$ superconformal theories in dimension 6 and reducing on a circle.  The only topological supercharge and compatible twisting homomorphism is the rank 4 supercharge discussed above, but even in the 6d $\mc N=(2,0)$ superconformal algebra there is \emph{no} potential for the twisted $\so(5)$-action.  Indeed, under this twisting homomorphism we can decompose $S \otimes W$ as $\sym^2(S) \oplus V \oplus \CC$ with $\sym^2(S)\cong \so(5)$ the adjoint representation. Therefore, a potential for the twisted copy of $\so(5)$ would have to be spanned by this summand inside the space of special supersymmetries.  However, this summand is not abelian: the bracket in the supertranslation algebra does not vanish on $\sym^2(S)$. Indeed, for generic elements $Q, Q'\in S$ the bracket $[Q\otimes Q, Q'\otimes Q'] = \Gamma(Q, Q') \omega(Q, Q')$ is nonzero.

Nevertheless, our calculations in 4d $\mc N=4$ allow one to construct a potential for $\so(4)$ under the Kapustin-Witten twisting homomorphism with suitable compatible rank 4 supercharges, or even a potential for $\so(3)$ under a $\Spin(4)$ twisting homomorphism and a compatible rank 2 supercharge.
\end{remark}

\begin{example}
The behaviour of $\mc N=2$ super Yang--Mills theory under the $\Spin(4)$-twisting homomorphism of Donaldson--Witten type was analysed by Anderson \cite{Anderson}, without taking the cohomology by a square-zero supercharge.
\end{example}

\begin{example}
A holomorphically twisted $\mc N=2$ 5d Super Yang--Mills theory using the twisting homomorphism of Anderson was discussed by Qiu and Zabzine \cite{QiuZabzine}.  As in the $\mc N=1$ theory of K\"all\'en and Zabzine discussed above this theory is defined on contact five-manifolds.  The solutions to the equations of motion are related to the Haydys--Witten equations for connections on a contact five-manifold.  The stress-energy tensor is shown to be $Q$-exact.
\end{example}

\begin{example}
The topological twist with respect to both the full twisting homomorphism and a $\Spin(4)$-twisting homomorphism of Kapustin--Witten type, in both cases with the rank 4 topological supercharge, has been studied by Geyer and M\"{u}lsch \cite{GeyerMuelsch} and Bak and Gustavsson \cite{BakGustavsson1, BakGustavsson2}.  In the case of the full twisting homomorphism they describe the twisted theory as ``5d fermionic Chern--Simons theory''.  In other words the action functional in the twisted theory is of Chern--Simons type but where the fundamental field is not a connection but a spinor-valued 2-form.
\end{example}

\subsection{Dimension 6}
The Lorentz group is $\Spin(6)\cong \SU(4)$. Let $V=\CC^6$ be the vector representation. The spinor representations are given by $S_+=\CC^4$ and $S_-=S_+^*$. We have equivariant isomorphisms
\[\wedge^2(S_+)\cong V \text{ and } \wedge^2(S_-)\cong V.\]

Pick symplectic vector spaces $W_+$ and $W_-$. Then the supertranslation algebra is
\[T= V\oplus \Pi(S_+\otimes W_+\oplus S_-\otimes W_-),\]
where $\dim W_+ = 2\mc{N}_+$ and $\dim W_- = 2\mc{N}_-$. The R-symmetry group $G_R$ is $\Sp(2\mc{N}_+, \CC)\times \Sp(2\mc{N}_-, \CC)$.


Compactification to $d=5$ corresponds to the isomorphisms $S^6_+\cong S^6_-\cong S^5$ as $\so(5)$-representations. The 5d auxiliary space is $W^5=W^6_+\oplus W^6_-$.

\paragraph{$\underline{\boldsymbol{\mc N=(1, 0)}}$:}$ $

\textit{Square-zero supercharges}: $Q$ has to have rank 1 in which case it automatically squares to zero. These supercharges are holomorphic.

\textit{Twisting homomorphisms}:
While there are no longer non-trivial twisting homomorphisms from the full spin group, there are twisting homomorphisms from the subgroup $\Spin(4) \iso \Spin(3) \times \Spin(3)$ given by projection onto either the first or second factor.  In either case there is a unique compatible holomorphic supercharge.

\textit{$\ZZ$-gradings}:
All non-zero supercharges can be promoted to twisting data using a maximal torus inside the R-symmetry group.  These twisting data are not compatible with the twisting homomorphisms above because there is no embedding $\mr U(1) \inj \Sp(2)$ acting on the defining representation with weight one.

\begin{example}
A calculation of a holomorphic twist of $\mc N=(1,0)$ super Yang-Mills theories in six dimensions is performed in \cite{BaulieuBossard}.
\end{example}

\paragraph{$\underline{\boldsymbol{\mc N=(1, 1)}}$:} $ $

\textit{Square-zero supercharges}:
\begin{enumerate}
\item Rank $(1, 0)$ and $(0, 1)$: we are in the setting of $\mc N=(1, 0)$ or $\mc N=(0, 1)$ supercharges. These twists are holomorphic.

\item Rank $(1, 1)$: such supercharges again automatically square to zero. They define lines $W^*_{Q_\pm}\sub S_\pm$ and we have the following orbits:
\begin{itemize}
\item If $\langle W^*_{Q_-}, W^*_{Q_+}\rangle = 0$, the supercharge has 4 invariant directions.

\item If $\langle W^*_{Q_-}, W^*_{Q_+}\rangle \neq 0$, the supercharge is topological.
\end{itemize}

\item Rank $(2, 2)$. In this case we have rank 2 subspaces $W^*_{Q_\pm}\sub S_{\pm}$.

\begin{lemma}
The supercharge $Q=Q_+ + Q_-$ is square-zero if and only if we have an exact sequence
\[0\longrightarrow W^*_{Q_+}\longrightarrow S_+\longrightarrow W_{Q_-}\longrightarrow 0\]
compatible with the volume forms on $W_{Q_\pm}$ and $S_+$.
\end{lemma}
\begin{proof}
As before, $Q_{\pm}^2$ is equal to the image of the Poisson bivector on $W^*_{Q_\pm}$ under the map $\wedge^2 W^*_{Q_{\pm}}\rightarrow \wedge^2 S_{\pm}\cong V$. Therefore, $Q$ is square-zero if and only if the two elements of $V$ are opposite.

Let $W_{Q_-}^\perp\hookrightarrow S_+$ be the kernel of $S_+\cong S_-^*\rightarrow W_{Q_-}$ which inherits a volume form from $S_+$ and $W_{Q_-}$. Then we may identify the image of $\wedge^2 W^*_-\rightarrow \wedge^2 (S_+)^*\cong \wedge^2 S_+$ with the image of $\wedge^2 W_{Q_-}^\perp\rightarrow \wedge^2 S_+$. So, the square-zero condition is equivalent to the equality $W_{Q_-}^\perp=W_{Q_+}^*$ equipped with their natural volume forms.
\end{proof}

This supercharge has 5 invariant directions: the image of $Q$ is given by the orthogonal complement to the isotropic line $\wedge^2 W^*_{Q_+}\sub \wedge^2 S_+\cong V$.
\end{enumerate}

\textit{Twisting homomorphisms}:
Again, there are non-trivial twisting homomorphism from $\Spin(4)\cong \Spin(3) \times \Spin(3)$.  Such twisting homomorphisms are equivalent to homomorphisms $\SU(2) \times \SU(2) \to \SL(2;\CC) \times \SL(2;\CC)$.
\begin{enumerate}
 \item The identity twisting homomorphism admits a $\bb{CP}^1$ family of compatible supercharges.  These supercharges are generically topological, with the exception of the points 0 and $\infty$ which are holomorphic supercharges of rank $(1,0)$ and $(0,1)$ respectively.
 \item The projection onto the first factor admits a unique compatible holomorphic supercharge of rank $(1,0)$.
 \item The projection onto the second factor admits a unique compatible holomorphic supercharge of rank $(0,1)$.
\end{enumerate}
Additionally, any supercharge is compatible with the trivial twisting homomorphism from its stabilizer.  Maximally, one can choose a topological supercharge with stabilizer isomorphic to $\mr U(3) \sub \SU(4)$ to obtain twisted theories which are $\mr U(3)$-structured.

\textit{$\ZZ$-gradings}:
All supercharges except for the one of rank $(2, 2)$ can be promoted to twisting data using the diagonal embedding $\mr U(1) \to \Sp(2; \CC) \times \Sp(2; \CC)$.  This twisting datum is only compatible with trivial twisting homomorphisms.

\begin{example}
There is a topological twist of $\mc N=(1,1)$ super Yang--Mills theory in six dimensions arising via dimensional reduction of an eight-dimensional topological twist (see Example \ref{8d_example} below) \cite{AcharyaOLoughlinSpence, BaulieuLosevNekrasov}.  This theory is defined on arbitrary K\"ahler threefolds.
\end{example}

\paragraph{$\underline{\boldsymbol{\mc N=(2, 0)}}$:}$ $

\textit{Square-zero supercharges}:$ $
\begin{enumerate}
\item Rank $(1,0)$: we get the holomorphic $\mc N=(1, 0)$ twist.

\item Rank $(2, 0)$: $Q$ squares to zero if and only if the image of $S_+^*\rightarrow W_+$ is Lagrangian. These supercharges have 5 invariant directions.
\end{enumerate}

\textit{Twisting homomorphisms}:
\begin{enumerate}
\item There is now a unique twisting homomorphism from $\Spin(5) \sub \Spin(6)$ given by the embedding of the maximal compact subgroup $\mathrm{USp}(4) \to \Sp(4; \CC)$, paralleling the case of 5d $\mc N=2$.  There is, however, no compatible square-zero supercharge.
\item There is a unique injective twisting homomorphism from $\Spin(4) \iso \SU(2) \times \SU(2) \to \Sp(4; \CC)$.  There is a $\bb{CP}^1$ family of compatible supercharges.  Again they are generically of rank 2, except for the points 0 and $\infty$ which have rank 1.

\item For any twisting homomorphism from $\Spin(4)$ that factors through projection onto the first or second factor there is a unique compatible rank 1 holomorphic supercharge.  
\end{enumerate}

\textit{$\ZZ$-gradings}:
Again all supercharges can be promoted to twisting data by suitably embedding $\SO(2)$ into the R-symmetry group.  This twisting datum cannot, however, be chosen compatibly with the $\Spin(4)$ twisting homomorphisms above for the reasons discussed in the example of 5d $\mc N=2$.

\begin{example}
There is a famous $\mc N=(2,0)$ superconformal field theory in six dimensions which does not admit a Lagrangian description, but which dimensionally reduces to $\mc N=2$ supersymmetric gauge theory in five dimensions \cite{WittenStringDynamics}.  Twists of this theory are discussed in the abelian case -- where there is an action functional -- by Anderson--Linander \cite{AndersonLinander} and Gran--Linander--Nilsson \cite{GranLinanderNilsson} in the case of the non-full-rank $\Spin(4)$ twisting homomorphism, and by Bak and Gustavsson \cite{BakGustavsson1} with respect to the $\Spin(5)$ and full rank $\Spin(4)$ twisting homomorphisms. 
\end{example}

\begin{remark}[Superconformal Theories]
In dimension 6, as we saw in Remark \ref{5d_superconformal_rmk} in dimension 5, there is no way of building a potential for a twisted $\so(5)$-action using the superconformal action.  However, as we noted in the remark above one will be able to construct potentials for twisted $\so(3)$ and $\so(4)$ actions with respect to compatible holomorphic-topological supercharges of rank 2.

We might also note that since there are no topological supercharges in the $\mc N=(1,0)$ and $\mc N=(2,0)$ supersymmetry algebras, there are no examples in our range to which we can apply Theorem \ref{superconformal_theorem}.  However, there \emph{are} topological supercharges in the $\mc N=(k,0)$ supersymmetry algebra for $k>2$, so if one considers superconformal theories with fields of spin greater than 1 there will exist topological twists which automatically define $\bb E_6$-algebras.  There do exist $\mc N=(4,0)$ superconformal theories of gravitational type in dimension 6: the linear such theories were constructed by Hull \cite{Hull1, Hull2} and there is some speculation in the literature regarding non-linear versions analogous to the $(2,0)$ theory \cite{ChiodaroliGunaydinRoiban, Borsten}.
\end{remark}

\subsection{Dimension 7}
For the first time the Lorentz group $\Spin(7)$ does not admit an alternative description via an exceptional isomorphism. Let $V=\CC^7$ be the vector representation and let $S=\CC^8$ be the spin representation. Then we can decompose the exterior square of $S$ as
\[\wedge^2(S)\cong V\oplus \mathfrak{so}(7).\]

Pick a symplectic vector space $W$. The supertranslation algebra is
\[T=V\oplus \Pi(S\otimes W),\]
where $\dim W = 2\mc N$. The R-symmetry group $G_R$ is $\Sp(2\mc{N};\CC)$.

Compactification to $d=6$ corresponds to the isomorphism $S^7\cong S^6_+\oplus S^6_-$ of $\so(6)$-representations. The auxiliary space is $W^6_+ = W^6_- = W^7$.

Let us recall the following description of $\Spin(7)$. There is a triple cross product on the octonions which we denote by $x\times y\times z$ for $x,y,z\in\bb O$. Then $\Spin(7)\sub \SO(\bb O)$ can be described as the subgroup preserving a 4-form
\[\Phi(x,y,z,w) = (x, y\times z\times w)\]
on $\bb O$. The morphism $\Spin(7)\rightarrow \SO(\bb O)$ gives the 8-dimensional spin representation $S=\bb O\otimes_\RR \CC$. The 7-dimensional vector representation $V_\RR$ can be identified with the imaginary quaternions and the morphism
\[\wedge^2(\bb O)\rightarrow \mr{Im}(\bb O)\cong V_\RR\]
is the cross product, defined as the imaginary part of the octonionic product.

In what follows we will write $\ol a$ for the octonionic conjugate of a (complexified) octonion $a$, $\tr(a) = a + \ol a$ for the octonionic trace, and $N(a) = a \cdot \ol a$ for the octonionic norm. Note that a complexified octonion $a$ is a zero-divisor if and only if $N(a) = 0$. The spinor pairing $S\otimes S\rightarrow \CC$ is the bilinear pairing associated to the quadratic form $N(-)$.

In dimensions 7 to 10 we will understand the orbits in the space of square-zero supercharges using Igusa's classification of spinors in dimensions up to 12 \cite{Igusa} -- recall that dimension 7 is the minimal dimension where there are non-trivial pure spinor constraints.  In all cases, there is a unique orbit in $S$ (in odd dimensions) or $S_+$ (in even dimensions) consisting of pure spinors, and this orbit has minimal dimension \cite{Chevalley}.  Igusa classifies the other orbits and their stabilizers.

\paragraph{$\underline{\boldsymbol{\mc N=1}}$:}$ $

\textit{Square-zero supercharges}: $ $

\begin{enumerate}
\item Rank 1. Any $Q = s \otimes w \in S\otimes W$ square to zero. According to Igusa \cite[Proposition 4]{Igusa} there are two orbits consisting of pure and impure spinors respectively.
\begin{itemize}
\item Suppose $N(s) = 0$. Then $Q$ satisfies the constraint defining a pure spinor, so it has 4 invariant directions.

\item Suppose $N(s)\neq 0$. Then $Q$ is a topological supercharge.
\end{itemize}

\item Rank 2. Let $Q = s_1\otimes w_1 + s_2\otimes w_2$ where $\{w_1, w_2\}$ is a symplectic basis of $W$. Then $Q^2 = 0$ if and only if $s_1\times s_2 = 0$. In particular, $N(s_1) = N(s_2) = 0$, i.e. both $s_1$ and $s_2$ are pure spinors. Let $L\sub V$ be the nullspace of $s_1$. Then we can identify
\[S \cong (\CC\oplus L\oplus \wedge^2 L\oplus \wedge^3 L)\otimes \det(L)^{-1/2}\]
as representations of $\Spin(V)_L\sub \Spin(V)$, where $\Spin(V)_L\cong \ML(L)\ltimes G_L$ is the stabilizer of $L\sub V$. The vector-valued spinor pairing $\wedge^2(S)\rightarrow V\cong L\oplus L^*\oplus \CC$ has the following components:
\begin{itemize}
\item The $L$-component is given by pairing $L\otimes \det(L)^{-1/2}$ and $\wedge^3 L\otimes \det(L)^{-1/2}$.

\item The $L^*$-component is given by pairing $\CC\otimes \det(L)^{-1/2}$ and
\[\wedge^2 L\otimes \det(L)^{-1/2}\cong L^*\otimes \det(L)^{1/2}.\]

\item The $\CC$-component is given by pairing $\CC\otimes \det(L)^{-1/2}$ and $\wedge^3 L\otimes \det(L)^{-1/2}$.
\end{itemize}

Thus, $s_2\in S$ such that $s_1\times s_2 = 0$ lie in the subspace
\[(L^*\oplus \CC)\otimes \det(L)^{1/2}\cong (\wedge^2 L\oplus \wedge^3 L)\otimes \det(L)^{-1/2}\sub S.\]
Using the $\Sp(W)$-action on $Q$ we can arrange that $s_2\in L^*\otimes \det(L)^{1/2}\sub S$. Since we assume $s_2$ is nonzero, the space of such is permuted by the subgroup $\SL(L)\sub \Spin(V)$ which fixes $s_1$. These supercharges have 5 invariant directions given by the span of $L\oplus \CC\sub V$ and the line in $L^*\sub V$ determined by $s_2\in L^*\otimes \det(L)^{1/2}$.
\end{enumerate}

\textit{Twisting homomorphisms}:
Supercharges are automatically compatible with the zero twisting homomorphism from their stabilizer, which for a topological rank 1 supercharge is isomorphic to $\mr G_2$. There are two twisting homomorphisms from the subgroup $\Spin(4)\sub \Spin(7)$ admitting compatible supercharges, given by the two projections $\Spin(4) \to \SU(2)$.  Each of these supercharges admits compatible rank 2 supercharges.

\textit{$\ZZ$-gradings}:
Square-zero supercharges of rank 1 can automatically be promoted to compatible twisting data using a suitable maximal torus $\mr U(1) \sub \Sp(2; \CC)$.  This twisting datum is automatically compatible with the zero twisting homomorphism. Supercharges of rank 2 cannot be promoted to a twisting datum, so the corresponding twisted theory is merely $\ZZ/2$-graded.

\begin{example}
The $\mr G_2$-structured rank 1 topological twist of 7-dimensional $\mc N=1$ super Yang--Mills studied by Acharya, O'Loughlin and Spence \cite{AcharyaOLoughlinSpence} as a dimensional reduction of an 8-dimensional $\Spin(7)$-structured theory (see Example \ref{8d_example} below). 
\end{example}

\subsection{Dimension 8} \label{dim8_section}
The Lorentz group is $\Spin(8)$. Let $V=\CC^8$ be the vector representation. One also has $S_+=\CC^8$ and $S_-=\CC^8$ the two semi-spin representations. The three representations $(V, S_+, S_-)$ are exchanged under the outer automorphism group (triality). We have a decomposition
\[S_+\otimes S_-\cong V\oplus Z,\]
where $Z$ is the irreducible $56$-dimensional representation generated by the highest weights of $S_+$ and $S_-$.

Pick a vector space $W$. The supertranslation algebra is
\[T=V\oplus \Pi(S_+\otimes W\oplus S_-\otimes W^*),\]
where $\dim W = \mc N$. The R-symmetry group $G_R$ is $\GL(\mc{N}; \CC)$.

Compactification to $d=7$ corresponds to the isomorphisms $S^7\cong S^8_+\cong S^8_-$ of $\so(7)$-representations. The auxiliary space is $W^7 = W^8\oplus (W^8)^*$ with the standard symplectic structure.

Let us recall a description of $\Spin(8)$.  One can identify $\Spin(8)$ as the subgroup of triples $\{(A_1, A_2, A_3)\}\in\SO(\mathbb O)^3$ such that $A_1(a_1)\times A_2(a_2)\times A_3(a_3) = 1$ for all $a_1,a_2,a_3\in\mathbb{O}$ such that $a_1\times a_2\times a_3=1$. The two (real) semi-spin representations and the (real) vector representation can be identified with the three projections $\Spin(8) \to \SO(\mathbb O)$ coming from this description.  The complex semi-spin and vector representations are their complexifications.

The $\Gamma$-pairing $S_+ \otimes S_- \to V$ is identified with the complexification of the octonion multiplication map.

\paragraph{$\underline{\boldsymbol{\mc N=1}}$:} Let $Q = Q_+ + Q_- \in S_+ \oplus S_-$ be a pair of complexified octonions.

\textit{Square-zero supercharges}:

\begin{enumerate}
 \item Rank (1, 0) and (0, 1): such supercharges automatically square to zero.  According to Igusa \cite[Proposition 1]{Igusa} there are two orbits consisting pure and impure spinors respectively.  Concretely:
 \begin{itemize}
   \item $N(Q) = 0$, i.e. it is a pure spinor. Such a supercharge is holomorphic.
   \item $N(Q)\neq 0$. Such a supercharge is topological.
 \end{itemize}
 
 \item Rank (1, 1). Since $Q_+\cdot Q_-=0$, we have $N(Q_+) = N(Q_-) = 0$, i.e. both $Q_+$ and $Q_-$ are pure spinors. Let $L\sub V$ be the nullspace of $Q_+$. Then we can identify
 \begin{align*}
   S_+ &\cong (\CC\oplus \wedge^2 L\oplus \wedge^4 L)\otimes \det(L)^{-1/2} \\
   S_- &\cong (L\oplus \wedge^3 L)\otimes \det(L)^{-1/2}
 \end{align*}
 as representations of $\Spin(V)_L\sub \Spin(V)$ where $\Spin(V)_L\cong \ML(L)\ltimes \wedge^2 L$ is the stabilizer of $L\sub V$. The vector-valued spinor pairing $S_+\otimes S_-\rightarrow V\cong L\oplus L^*$ has the following components:
 \begin{itemize}
 \item The $L$-component is given by pairing $\wedge^4(L)\otimes \det(L)^{-1/2}\sub S_+$ and $L\otimes\det(L)^{-1/2}\sub S_-$.
 \item The $L^*$-component is given by pairing $\CC\otimes \det(L)^{-1/2}\sub S_+$ and $\wedge^3(L)\otimes \det(L)^{-1/2}\sub S_-$.
 \end{itemize}
 Thus, $Q_-\in S_-$ such that $\Gamma(Q_+, Q_-) = 0$ lie in the subspace
 \[\wedge^3(L)\otimes \det(L)^{-1/2}\sub S_-.\]
 Since we assume $Q_-$ is nonzero, the space of such is permuted by the subgroup $\SL(L)\sub \Spin(V)$ which fixes $Q_+$. These supercharges have 5 invariant directions given by the span of $L\sub V$ and the line in $L^*\sub V$ determined by $Q_-\in L^*\otimes \det(L)^{1/2}$.
\end{enumerate}

\textit{Twisting homomorphisms}:
As we mentioned in the previous section, supercharges are automatically compatible with the zero twisting homomorphism from their stabilizer. For instance, in the case of rank (1, 0) or (0, 1) supercharges we have the following. For a topological supercharge it is isomorphic to $\Spin(7)$, and for a holomorphic supercharge (pure spinor) is isomorphic to $\SU(4)$ \cite[Theorem 3.1.10]{Charlton}.

\textit{$\ZZ$-gradings}:
Every rank $(1, 0)$ or rank $(0, 1)$ square-zero supercharge can be promoted to compatible twisting data using the maximal torus $\mr U(1)\sub \GL(1)$. These twisting data are automatically compatible with the zero twisting homomorphism. The rank $(1, 1)$ supercharges cannot be promoted to a compatible twisting datum, so the corresponding twisted theory is merely $\ZZ/2$-graded.

\begin{example} \label{8d_example}
Acharya, O'Loughlin and Spence \cite{AcharyaOLoughlinSpence} studied the topological rank 1 twist of $\mc N=1$ supersymmetric Yang--Mills defined on all 8-manifolds with a $\Spin(7)$-structure, in particular verifying that the gravitational stress-energy tensor is $Q$-exact.  They discuss a connection to the theory of $\Spin(7)$-instantons, and discuss the reduction to 7- and 6-dimensional theories.  These theories were additionally studied by Baulieu, Kanno and Singer \cite{BaulieuKannoSinger}.  It would be interesting to describe concretely the BV complex of this twisted theory, which would allow one to check for the existence of a homotopically trivial dilation action.
\end{example}

To our knowledge, theories twisted by rank (1, 1) supercharges have not been discussed in the literature.

\subsection{Dimension 9}
The Lorentz group is $\Spin(9)$. Let $V=\CC^9$ be the vector representation and $S=\CC^{16}$ the spin representation. Then
\[\sym^2(S)\cong \CC \oplus V\oplus Z,\]
where $Z$ is the irreducible $126$-dimensional representation generated by the highest weight of $\sym^2(S)$.

Pick a vector space $W$ equipped with a symmetric nondegenerate pairing. The supertranslation algebra is
\[T= V\oplus \Pi(S\otimes W),\]
where $\dim W = \mc{N}$. Since we restrict to the case where there are no fields of spins greater than one, $\mc N = 1$ and the R-symmetry group $G_R$ is $\O(1)\cong \ZZ/2$.

Compactification to $d=8$ corresponds to the isomorphism $S^9\cong S^8_+\oplus S^8_-$ of $\so(8)$-representations. The auxiliary space is $W^8 = W^9$.

One can describe the group $\Spin(9)$ explicitly as the subgroup of $\SO(\mathbb O^2)$ that preserves a certain 8-form $\Phi$ described by Berger \cite{Berger}.  Consider a set of elements $J_1, \ldots, J_9$ of $SO(\mathbb O^2)$ defined by
\[J_1 = \begin{pmatrix}0&\mathrm{Id}\\\mathrm{Id}&0\end{pmatrix}, J_i= \begin{pmatrix}0&-R_{e_{i-1}}\\R_{e_{i-1}}&0\end{pmatrix}, J_9 = \begin{pmatrix}\mathrm{Id}&0\\0&-\mathrm{Id}\end{pmatrix}\]
where $i = 2, \ldots, 8$ and $R_x$ denotes right multiplication by the octonion $x$.  For each pair $1 \le i,j \le 9$ define a 2-form by $\omega_{ij} = dx^k \wedge J_i J_j dx^k$.  The 8-form $\Phi$ is the $t^5$ coefficient of the characteristic polynomial of the 9x9 matrix $(\omega_{ij})$ of 2-forms (this construction was given by Parton and Piccinni \cite{PartonPiccinni}).  This map $\Spin(9) \to \SO(16)$ defines the (real) spin representation $S_\RR$ in 9-dimensions, whose complexification is $S$.

\paragraph{$\underline{\boldsymbol{\mc N=1}}$:}

\textit{Square-zero supercharges}:
Fix a supercharge $(Q,Q')$ in $(\bb O \otimes_{\RR} \CC)^{\oplus 2}$.  It squares to zero if it squares to zero viewed as an 8d $\mc N=1$ supercharge, i.e. if $Q\cdot Q' = 0$, plus the additional condition $N(Q) - N(Q') = 0$. The first condition implies $N(Q) = N(Q')$, so in fact $(Q, Q')$ squares to zero in 9d if and only if it squares to zero in 8d. Using the 10d pairing $\Gamma_{10}$ that we will describe in the next section, this implies that every supercharge $(Q, Q')$ is pure. These supercharges have 5 invariant directions.

\textit{Twisting homomorphisms}:
Pure spinors in 9 dimensions are again stabilized by $\SU(4) \sub \Spin(9)$, so they are compatible with the zero twisting homomorphism from $\SU(4)$.

\textit{$\ZZ$-gradings}:
Since the R-symmetry group is discrete in this dimension, twisted theories will only be $\ZZ/2$-graded.

\begin{example}
A 9-dimensional twist of $\mc N=1$ super Yang--Mills theory is described by Baulieu, Losev and Nekrasov \cite{BaulieuLosevNekrasov} on manifolds of the form $X \times S^1$ where $X$ is a $\Spin(7)$ 8-manifold.  They verify that the Lagrangian density of this theory is $Q$-exact.
\end{example}

\subsection{Dimension 10}

The Lorentz group is $\Spin(10)$. Let $V=\CC^{10}$ be the vector representation and $S_+=\CC^{16}$, $S_-=\CC^{16}$ the spin representations. We have
\[\sym^2(S_+)\cong V\oplus Z_+,\quad \sym^2(S_-)\cong V\oplus Z_-,\]
where $Z_+$ and $Z_-$ are the $126$-dimensional representations generated by the highest weights.

Pick a pair of vector spaces $W_+, W_-$ equipped with symmetric nondegenerate pairings. The translation part of the supersymmetry algebra is
\[T= V\oplus \Pi(S_+\otimes W_+\oplus S_-\otimes W_-),\]
where $\dim W_+ = \mc{N}_+$ and $\dim W_- = \mc{N}_-$. Since we restrict to the case where there are no fields of spins greater than one, we only consider the cases $\mc{N} = (1, 0)$ or $\mc{N} = (0, 1)$. In either case the R-symmetry group $G_R$ is $\ZZ/2$.

Compactification to $d = 9$ corresponds to the isomorphism $S^9\cong S^{10}_+\cong S^{10}_-$ of $\so(9)$-representations. The auxiliary space is $W = W_+\oplus W_-$.

In Lorentzian signature the group $\Spin(1,9)$ can be identified as the special linear group $\SL(2;\bb O)$ for the octonions, as explained in \cite[Chapter 6]{DeligneSpinors}.  The Majorana--Weyl spinor representation $S_{\RR+}$ can be identified as the fundamental representation $\bb O^2$, so after complexification the Weyl spinor representation can be identified as $S_+ = (\bb O \otimes_\RR \CC)^2$.  In these terms the vector-valued pairing $\Gamma_{10} \colon S_+ \otimes S_+ \to V$ is written as
\begin{align*}
\Gamma_{10}&((Q_1,Q_1'), (Q_2,Q_2')) \\
&= (Q_1 \cdot Q_2' + Q_1' \cdot Q_2, \tr(Q_1 \cdot \ol{Q_2}) - \tr(Q_1' \cdot \ol{Q_2'}), \tr(Q_1 \cdot \ol{Q_2}) + \tr(Q_1' \cdot \ol{Q_2'}))
\end{align*}
where $Q_1, Q_1', Q_2, Q_2'$ are viewed as complexified octonions and the first term on the right hand side is defined in terms of the complexified octonionic multiplication map.  

\begin{remark}
We obtain the 8d and 9d vector-valued pairings by projecting the pairing $\Gamma_{10}$ onto the first 8 or first 9 components.
\end{remark}

\paragraph{$\underline{\boldsymbol{\mc N=(1,0)}}$:}

\textit{Square-zero supercharges}:$ $

A spinor $(Q,Q')$ squares to zero if $Q\cdot Q'=0, N(Q) = 0$ and $N(Q') =0$. According to our discussion above in the 8-dimensional case, this occurs when $Q = a + ib$ for orthogonal octonions $a$ and $b$ with the same norm and $Q' = \ol Q \cdot x$ for some complexified octonion $x$, possibly zero, or similarly with $Q$ and $Q'$ reversed.

Such square-zero supercharges are always pure, and therefore always holomorphic as we discussed in Section \ref{pure_spinor_section}.  These spinors all lie in the same $\Spin(10)$ orbit.

\textit{Twisting homomorphisms}:
The stabilizer of a pure spinor in 10 dimensions is isomorphic to $\SU(5) \sub \SO(10)$ by \cite[Theorem 3.1.10]{Charlton}, so every holomorphic supercharge is compatible with the zero twisting homomorphism from $\SU(5)$.

\textit{$\ZZ$-gradings}:
Since the R-symmetry group is discrete in this dimension, twisted theories will again only be $\ZZ/2$-graded.

\begin{example}
The holomorphic twist of 10-dimensional maximal super Yang--Mills theory is discussed by Baulieu \cite{Baulieu}.  This twist is automatically $\SU(5)$-structured.  Baulieu argues that this twisted theory is equivalent to a 5-dimensional holomorphic Chern--Simons theory.
\end{example}

\pagestyle{bib}
\printbibliography

\textsc{Institut des Hautes \'Etudes Scientifiques}\\
\textsc{35 Route de Chartres, Bures-sur-Yvette, 91440, France}\\
\texttt{celliott@ihes.fr}\\
\vspace{-5pt}

\textsc{Institut f\"ur Mathematik, Universit\"at Z\"urich}\\
\textsc{Winterthurerstrasse 190, 8057 Z\"urich, Switzerland}\\
\texttt{pavel.safronov@math.uzh.ch}\\
\end{document}